\RequirePackage{fix-cm}
\documentclass{article} 

\usepackage{color}
\usepackage[numbers]{natbib}
\usepackage{enumitem}

\usepackage{hyperref}

\usepackage{amsmath, amsthm, amsfonts, amssymb,amscd,bbm}
\usepackage{amscd}
\usepackage{pxfonts}
\usepackage{etoolbox}

\usepackage{mathtools}
\mathtoolsset{showonlyrefs}
\numberwithin{equation}{section}
\usepackage[applemac]{inputenc}

\theoremstyle{plain}
\newtheorem{theorem}{Theorem}[section]
\newtheorem{proposition}[theorem]{Proposition}
\newtheorem{lemma}[theorem]{Lemma}
\newtheorem{corollary}[theorem]{Corollary}

\theoremstyle{definition}
\newtheorem{definition}{Definition}[section]
\newtheorem{example}[definition]{Example}

\theoremstyle{remark}
\newtheorem{remark}{Remark}[section]

\DeclareMathOperator{\tr}{Tr}

\DeclareMathOperator{\im}{Im}
\def\Z{\mathbb{Z}}	
\def\C{\mathbb{C}}	
\def\R{\mathbb{R}}	
\newcommand\cS{{\mathcal S}}

\def\r{{\rm r}}
\def\l{{\rm l}}
\def\c{{\rm c}}
\def\a{{\rm a}}

\renewcommand{\leq}{\leqslant} 		

\newcommand{\ldb}{\{\!\{}
\newcommand{\rdb}{\}\!\}}
\newcommand{\lsb}{[\![}
\newcommand{\rsb}{]\!]}

\def\ud{\mathrm{d}}

\newcommand{\dev}{\partial}

\let\ker\relax
\DeclareMathOperator{\ker}{Ker}
\def\A{\mathcal{A}}
\def\hA{\hat{\A}}
\def\F{\mathcal{F}}
\def\hF{\hat{\F}}

\def\hcA{\hat{\mathcal{A}}}

\def\ldb{\{\!\{ }
\def\rdb{\}\!\} }

\usepackage{authblk}
\begin{document}

	\title{Hamiltonian Structures for Integrable Nonabelian Difference Equations}
	\author[1, 2]{Matteo Casati}
	\author[2]{Jing Ping Wang}
	\affil[1]{School of Mathematics and Statistics, Ningbo University, Ningbo, 315211, China}
	\affil[2]{School of Mathematics, Statistics and Actuarial Science, University of Kent,
		Canterbury CT2 7FS, United Kingdom}

\date{}

\maketitle

	\begin{abstract}
	In this paper we extensively study the notion of Hamiltonian structure for nonabelian differential-difference systems, exploring the link between the different algebraic (in terms of double Poisson algebras and vertex algebras) and geometric (in terms of nonabelian Poisson bivectors) definitions.  We introduce multiplicative double Poisson vertex algebras (PVAs) as the suitable noncommutative counterpart to multiplicative PVAs, used to describe Hamiltonian differential-difference equations in the commutative setting, and prove that these algebras are in one-to-one correspondence with the Poisson structures defined by difference operators, providing a sufficient condition for the fulfilment of the Jacobi identity. Moreover, we define nonabelian polyvector fields and their Schouten brackets, for both finitely generated noncommutative algebras and infinitely generated difference ones: this allows us to provide a unified characterisation of Poisson bivectors and double quasi-Poisson algebra structures. Finally, as an application we obtain some results towards the classification of local scalar Hamiltonian difference structures and construct the Hamiltonian structures for the nonabelian Kaup, Ablowitz--Ladik and Chen--Lee-Liu integrable lattices.
	\end{abstract}

	\tableofcontents

	\section{Introduction}
Nonabelian systems of ordinary and partial differential equations have been previously studied, for example in \cite{Kup00,man76,ms00,ors12,OS98}. In our recent paper \cite{cw19-2} we have studied nonabelian differential-difference integrable system and described a number of examples lifting, by means of their Lax representation, well-known Abelian systems to their nonabelian counterparts. While we provided recursion operators for all those examples, hence establishing a way to produce their higher symmetries (or, equivalently, integrable hierarchies), at the time we could not construct the appropriate Hamiltonian description for some of them despite the very simple structure they possess in the Abelian case. In this paper, we answer this open question by thoroughly investigating noncommutative Hamiltonian structures.
	
In Section \ref{sec:Ham_formalism}, we review the language and the notion of Hamiltonian structure widely adopted among the Integrable Systems community. Here, the focus is the notion of Poisson bracket: an operator is called Hamiltonian if it can be used to define a Poisson bracket; the same operator defines a functional bivector whose prolongation along its flow vanishes. The latter is equivalent to the Poisson property. 

In Section \ref{sec:dPA} we present a more recent viewpoint: a suitable algebraic description of the Hamiltonian structures of nonabelian ordinary differential equations can be given in terms of double Poisson algebras \cite{vdb}. While this theory is reasonably well-established, the use of a similar formalism to go beyond ordinary differential equations and describe partial differential equations is much more recent \cite{dskv15} and encoded in the theory of double Poisson \emph{vertex} algebras (double PVAs). The main object of these theories is the so-called double (lambda) bracket, defined on associative (differential) algebras, essentially replacing the action of the Hamiltonian operators.  In this paper, we focus on differential-difference equations (D$\Delta$Es), a class of systems where the time is a continuous variable, while the spatial one takes values on a lattice. We follow the lines of multiplicative PVAs \cite{dskvw19} to define the analogue algebraic structures for noncommutative difference algebras, and hence to apply it to differential-difference nonabelian systems. We call such a structure \emph{double multiplicative Poisson vertex algebra} and we define it in Section \ref{ssec:dmPVA}, see Definition \ref{def:dmpva}, and prove that it is equivalent to the ``usual'' notion of Hamiltonian (\emph{rectius} Poisson) operator.

In Section \ref{sec:geom} we follow the spirit of classical Poisson geometry, which has been successfully adopted in the study of (Abelian) Hamiltonian PDEs and D$\Delta$Es, to find a counterpart of these algebraic structures as suitably defined functional bivector fields;  this allows us to define the whole Poisson-Lichnerowicz complex and its cohomology, and to provide a new interpretation of the double quasi-Poisson structure.
	
To do so, we first need a notion of Schouten bracket for nonabelian systems and an operational way of computing it. We define it in Section \ref{ssec:sch-ul} for systems of ODEs and Section \ref{ssec:sch-l} for D$\Delta$Es, and prove that it satisfies the graded skewsymmetry and Jacobi properties that characterise a Gerstenhaber algebra. In Section \ref{ssec:pois} we show how the Poisson bracket and its properties, as well as all the standard objects in the theory of Poisson manifolds, can be defined in terms of a Poisson bivector.

In Section \ref{sec:quasi} we present, starting from the Kontsevich system \cite{WE12}, the notion of quasi-Poisson algebras, introduced by Van Den Bergh \cite{vdb}, from which one can define a Poisson bracket even if the corresponding bivector is not Poisson. We investigate the equations that a bivector must satisfy in order to be called ``Hamiltonian'' and find a simple interpretation for this property.

Finally, in Section \ref{sec:Ham-ex} we provide several examples of scalar ultralocal (which coincides with those for ordinary differential equations, or with some double Poisson algebras) and local nonabelian Hamiltonian operators. However, many local (and even ultralocal) Hamiltonian operators that produce Abelian integrable systems do not have a local counterpart in the nonabelian case, but correspond to nonlocal difference operators. The ``missing'' structures of our recent work fall within this category; we describe for the first time the two-component Hamiltonian structure for the nonabelian Kaup, Ablowitz-Ladik and Chen-Lee-Liu lattices (see \cite{cw19-2} for their recursion operators and \cite{kmw13} for the Abelian case).		 


	\section{Hamiltonian structures and Poisson bivectors}\label{sec:Ham_formalism}
	In this section, we review the formalism of nonabelian Hamiltonian equations we used in our previous work \cite{cw19-2}. It is the straightforward generalisation of the language that P.~Olver and V.~Sokolov developed to describe Hamiltonian partial differential equations on an associative (but not commutative) algebra \cite{OS98} and, subsequently, adapted (actually, simplified) by Mikhailov and Sokolov to noncommutative ODEs \cite{ms00}. Kuperschmidt, in his vast book \cite{Kup00}, discusses partial differential, differential-difference and difference-difference noncommutative systems, including a similar formalism for Hamiltonian operators. His approach and what we present here, extending Olver and Sokolov's notation, are equivalent when applied to local Hamiltonian operators. The extension to nonlocal ones is, at the best of our knowledge, a new feature introduced in \cite{cw19-2}.
	
In particular, we want to stress the distinction between \emph{Hamiltonian} operators (characterised by their ability to define a Poisson bracket between local functionals) and \emph{Poisson} bivectors (identified by the vanishing of their Schouten torsion).	

While we refer to our previous work for the precise definitions, we summarize hereinafter the main objects necessary to state our results. The main feature of this language is the introduction of functional $0$-, $1$-, and $2$-vector fields to describe, respectively, observables, evolutionary equations and brackets, closely following the treatment that Olver systematised for partial differential equations \cite{O93}.

	\begin{definition}[Difference Laurent polynomials]
	The space $\A$ of noncommutative \emph{difference (Laurent) polynomials} is a linear associative algebra with unit
	$$\A=\frac{\R\left\langle u^i_n,\left(u^i_n\right)^{-1}\right\rangle}{\left\langle u^i_n\left(u^i_n\right)^{-1}\!-1,\enspace\left(u^i_n\right)^{-1}\!u^i_n-1\right\rangle},\qquad i=1,\ldots,\ell,\; n\in\Z,$$
		endowed with an automorphism $\cS\colon\A\to\A$ given by $\cS u^i_n=u^i_{n+1}$. 
	\end{definition}
	Note that taking the quotient with respect to the two-sided ideal allows us to have that, in $\A$, $u^i_n\left(u^i_n\right)^{-1}=\left(u^i_n\right)^{-1}u^i_n=1$ so that these variables are multiplicative inverses of each other. In the spirit of the formal calculus of variations, the variables $u^i:=u^i_0$ generate the algebra $\A$ and represent the $\ell$ components of the dependent variable of the differential-difference equations. Note that the product in $\A$ is, in general, non-commutative. Let us denote by $[-,-]$ the commutator on $\A$, i.e. $[a,b]=ab-ba$.
	\begin{definition}[Local functionals]
	The elements of the quotient space
	\begin{equation}
	\F=\frac{\A}{(\cS-1)\A+[\A,\A]}
	\end{equation}
	are called \emph{local functionals}. We denote the projection from $\A$ to $\F$ as $\int \tr -$, which satisfies
	\begin{equation}
	\int\tr \cS f=\int\tr f,\qquad\qquad\int\tr fg=\int\tr gf
	\end{equation}
	for all $f,g\in\A$.	
	\end{definition}
	In our notation, the integral sign denotes the quotient operation with respect to $(\cS-1)$ and $\tr$ (``trace'', since in the standard example the generators of $\A$ are elements of $\mathfrak{gl}_N$ with the canonical matrix product) is the quotient with respect to the commutator.

A derivation of $\A$, denoted by $D(a)$, is a linear map satisfying the Leibniz's property $D(ab)=D(a)b+aD(b)$. Note that, because of the noncommutativity of the product, for a monomial $abc$ the property becomes 
$$D(abc)=D(a)bc+aD(b)c+abD(c),$$
and so on, until $D$ acts on the single generators of $\A$. For the inverse generators, we have $D(a^{-1})=-a^{-1}D(a)a^{-1}$; this follows from $D(1)=0$.

	\begin{definition}
		An \emph{evolutionary difference vector field} $X$ is a derivation of $\A$ that commutes with the shift operator $\cS$.
	\end{definition}
	The necessary and sufficient condition for $X$ to be an evolutionary vector field is that it satisfies the property $X(u^i_n)=\cS^n X^i$ where $(X^1,\ldots,X^\ell) \in \A^\ell$ is called the \emph{characteristics} of the vector field. 

Note that, exactly as in the commutative case, a differential-difference system
	\begin{equation}\label{eq:evsyst}
	u^i_t=X^i(\ldots,\cS^{-1}\mathbf{u},\mathbf{u},\cS\mathbf{u},\ldots) \qquad i=1,\ldots,\ell,\quad \mathbf{u}=\{u^j\}_{j=1}^\ell
	\end{equation}
	can be identified with an evolutionary vector field of characteristic $\{X^i\}_{i=1}^\ell$.	

	\begin{definition}
		A \emph{local scalar difference operator} is a linear map $K\colon\A\to\A$ that can be written as a finite linear combination 
		of terms of the form $\r_f \l_g \cS^p$ for $p\in\Z$, $f,g$ in $\A$, where $\r$ and $\l$ denote, respectively, the 
		multiplication on the right and the multiplication on the left. Namely, we have
		\begin{equation}\label{eq:diffop-def}
		K h = \sum_{p=N}^M\sum_{\alpha_p}\r_{f^{(\alpha_p)}}\l_{g^{(\alpha_p)}}\cS^p h=\sum_{p=N}^M\sum_{\alpha_p} g^{(\alpha_p)} \left(\cS^p h\right) f^{(\alpha_p)}.
		\end{equation}
	We call $(N,M)$, with $N\leq M$, respectively, the minimal and maximal power of $\cS$ appearing in the expansion, the \emph{order} of the difference operator. With the notation adopted in \eqref{eq:diffop-def} we want to stress that in the linear combination there are, in general, several terms with the same number of shifts. In the rest of this paper, where there is no ambiguity, we drop the indices and the double sum to represent the difference operators.
	\end{definition}
	\begin{definition}\label{def:ul}
	We call a difference operator \emph{ultralocal} if its order is $(0,0)$, namely if it does not contains shift operators.
	\end{definition}
	The multiplication operators have the properties
	\begin{equation}
	\l_f \l_g=\l_{fg}\qquad\qquad \r_f\r_g=\r_{gf}\qquad\qquad \r_f\l_g=\l_g\r_f.
	\end{equation}
	Moreover, we define the \emph{commutator} $\c_f:=\l_f-\r_f$, that is, $[f,g]=\c_f g$, and the \emph{anticommutator} $\a_f=\l_f+\r_f$. Note that $\c_f$ is a derivation.	
	
	The formal adjoint of the scalar difference operator $K=\sum\l_f \r_g \cS^p$ is 
	\begin{equation}\label{eq:adj}
	K^\dagger:=\sum\cS^{-p}\r_f\l_g.
	\end{equation}
In the multi-component case, namely when $\ell>1$, we consider $\ell\times\ell$ matrices $(K)_{ij}$ whose entries are scalar difference operators. To avoid making the notation too heavy, we denote the entry $(K)_{ij}$ as $K^{ij}$. The formal adjoint of (multi-component) $K$ is $(K^\dagger)_{ij}=(K^{ji})^\dagger$.  We say that a difference operator in \emph{skewsymmetric} if $K^\dagger=-K$.
	
	We define the variational derivative of a local functional $F=\int\tr f$ using a generic evolutionary vector field $X$ of characteristic $\{X^i\}_{i=1}^\ell$. We have 
	\begin{equation}\label{eq:defvarder}
	\int\tr\sum_{i=1}^\ell \frac{\delta F}{\delta u^i}X^i:=\int\tr X(f).
	\end{equation}

Providing a closed formula for the variational derivative of a local functional (or density) needs some more work, because of the way in which an evolutionary vector fields acts on the single generators of the algebra $\A$, ``splitting'' the density $f$ around them. This is the crux of the matter when dealing with derivations on a noncommutative space; a possible solution is ``doubling'' the space: the theory of double Poisson algebras \cite{vdb,dskv15} stems from it, and we will discuss it at large in Section \ref{sec:dPA}.

	The operation described in \eqref{eq:defvarder} can be regarded as a pairing between (evolutionary) vector fields and (variational) 1-forms; we use as a shorthand notation for such a paring $\langle \delta F, X\rangle$. 
	
The definition of local $p$-vector fields (see \cite{CW19} for the difference Abelian case) must be postponed to Section \ref{sec:dPA}; however, it is possible to adopt a tailored version of the so-called $\theta$ formalism following Olver and Sokolov's treatment \cite{OS98}.
 
We introduce the basic uni-vectors $\theta_{i,n}$, where $\theta_{i,n}=\cS^n\theta_i$; these objects (contrasting with the commutative case, where they are Grassmann variables) do not have any special parity with respect to the product. However, they are odd with respect to the permutations under the trace operation.

\begin{definition}\label{def:thetadens}
	The elements of the space $$\hA:=\A\langle\{\theta_{i,n}\}_{i=1,n\in\Z}^\ell\rangle$$ are called \emph{densities of (functional) polyvector fields}. The space $\hA$ is a graded algebra where $\deg_\theta\theta_{i,n}=1$, $\deg_\theta u^i_n=0$. Homogeneous elements of $\hA$ of degree $p$ in $\theta$ are densities of $p$-vector fields.
\end{definition}
\begin{definition}\label{def:thetafield}
A local \emph{functional polyvector field} is an element of the quotient space

	\begin{equation}
	\hF=\frac{\hA}{(\cS-1)\hA+[\hA,\hA]},
	\end{equation}
	where the commutator $[-,-]$ is $[a,b]=a b -(-1)^{|a||b|}b a$ and we denote $|a|:=\deg_\theta a$ and $|b|:=\deg_\theta b$. This commutator coincides with the standard commutator on $\A$, since $\deg_\theta\A=0$.
\end{definition}

	The trace form (and as a consequence the quotient operation $\hA\twoheadrightarrow\hF$) is then graded commutative, namely
	\begin{equation}\label{eq:tr}
	\tr \left(a\,b \right)=(-1)^{|a||b|}\tr \left(b\,a\right).
	\end{equation}

We denote as $\hA^p$ (respectively, $\hF^p$) the homogeneous component of degree $p$ in $\hA$ (resp., $\hF$).

Take notice of the abuse of language in Definition \ref{def:thetadens} and \ref{def:thetafield}: the original definition of functional polyvector field does not require the $\theta$ formalism, and one must normally prove that this formalism induces an isomorphism between densities of (classically defined) polyvector fields and polynomials in $\theta$. We leave it to Section \ref{sec:dPA} and exploit the formalism for our computations, following Olver and Sokolov's lead.
 	
For simplicity, we denote $\theta_i=\theta_{i,0}$ in the multi-component case, and -- in the scalar $\ell=1$ case -- $\theta_n=\theta_{1,n}$, $\theta=\theta_{1,0}$. To avoid confusion, in the following sections we will introduce different Latin ($u,v,\dots$) and Greek ($\theta,\zeta,\dots$) letters denoting, respectively, different $u^i$'s and $\theta_j$'s.
	
The formal evolutionary vector field of characteristics $K\Theta=(\sum_{j}K^{ij}\theta_j)_{i=1}^\ell$, where $K$ is a difference operator with entries $K^{ij}$,  is denoted $\mathbf{pr}_{K\Theta}$ and it is a graded derivation of degree 1. We have
$$\mathbf{pr}_{K\Theta}(u^i)=\sum_j K^{ij}\theta_j,\qquad\mathbf{pr}_{K\Theta}(\theta_i)=0,$$ and $$\mathbf{pr}_{K\Theta} (ab)=\mathbf{pr}_{K\Theta}(a)b+(-1)^{|a|}a\,\mathbf{pr}_{K\Theta}(b).$$

Moreover, we can associate to any difference operator (in particular, skewsymmetric) $K$ the functional \emph{bivector}
	\begin{equation}\label{eq:defBiv}
	P=\frac12\int\tr\left(\sum_{i,j=1}^\ell \theta_i \,K^{ij}\theta_j\right).
	\end{equation}
	Similarly, for $K$ a difference operator we can define a bracket between local functionals
	\begin{equation}\label{eq:defBra}
	\{F,G\}:=\langle \delta F,K\delta G\rangle=\int\tr \left(\sum_{i,j=1}^\ell\frac{\delta F}{\delta u^i} K^{ij} \frac{\delta G}{\delta u^j}\right).
	\end{equation}

\begin{definition}
A skewsymmetric difference operator $K$ is \emph{Hamiltonian} if and only if the bracket \eqref{eq:defBra} endows the space of local functionals with the structure of a Lie algebra, namely if and only if the bracket is skewsymmetric and satisfies the Jacobi identity
\begin{equation}
\{A,\{B,C\}\}+\{B,\{C,A\}\}+\{C,\{A,B\}\}=0,\qquad \forall A,B,C\in\F.
\end{equation}
\end{definition} 
\begin{definition}\label{def:Pois}
We say that a bivector $P$, associated to the skewsymmetric operator $K$, defined as in \eqref{eq:defBiv} is a \emph{Poisson} bivector if and only if
	\begin{equation}\label{eq:HamProp}
		\mathbf{pr}_{K\Theta}P=\frac12 \int\tr\left(\sum_{i,j=1}^\ell \mathbf{pr}_{K\Theta}(\theta_i K^{ij}\theta_j)		\right)=0.
		\end{equation}
We also call an operator whose associated bivector is Poisson a \emph{Poisson operator}, or Poisson structure.\end{definition}
\begin{example}
The skewsymmetric operators $K=\cS-\cS^{-1}$ and $H=\r_u-\l_u $ are Poisson. For $K$, we have that
\begin{equation}
P=\frac12\int\tr\left(\theta\left(\cS-\cS^{-1}\right)\theta\right)=\int\tr\theta\theta_1.
\end{equation}
The bivector does not depend on any generators $u$, and hence $\mathbf{pr}_{K\Theta}P=0$.
For $H$ we have $H\theta=\theta u-u\theta$ and $P=\int\tr u\theta\theta$. Then condition \eqref{eq:HamProp} is
\begin{equation}
\int\tr\left(\mathbf{pr}_{\theta u}\left(u\theta\theta\right)-\mathbf{pr}_{u \theta}\left(u\theta\theta\right)\right)=\int\tr\left(\theta u\theta\theta-u\theta\theta\theta\right)=0.
\end{equation}
\end{example}
The relation between Poisson geometry (the geometry of manifolds endowed with a Poisson bivector) and Hamiltonian systems is well known; Poisson bivectors always define Hamiltonian structures: for ODEs, PDEs and differential-difference systems, both Abelian and nonabelian (see for instance \cite{O93,s20} and references therein).
	
However, a remarkable difference between (standard) Abelian operators and the nonabelian ones is that, while in the former case the notion of Hamiltonian operator and of Poisson bivector are equivalent (namely, identity \eqref{eq:HamProp} holds for all Hamiltonian operators and all the Poisson bivectors are defined by Hamiltonian operators, see \cite{D94} for reference), in the noncommutative setting the Poisson property \eqref{eq:HamProp} is a sufficient but not necessary condition for an operator to be Hamiltonian: such is the case for operators defined in terms of double quasi-Poisson algebras (see Section \ref{sec:quasi}). This is why we argue that the terms ``Hamiltonian'' and ``Poisson'' should cease to be used interchangeably.

The identity \eqref{eq:HamProp} is essentially due to Olver (it is used for the Abelian differential case in \cite{O93} and for the nonabelian differential case in \cite{OS98}); we call bivectors for which it holds true ``Poisson'' because the left hand side of the identity is equivalent to the Schouten torsion of $P$, in analogy to the finite dimensional and commutative frameworks. We will show the equivalence of the notions in Theorem \ref{thm:PoisBiv}. 

Finally, we say that an evolutionary system \eqref{eq:evsyst} is a \emph{Hamiltonian system} if and only if
		\begin{equation}\label{eq:HamSyst_expl}
		u^i_t=X^i=\sum_{j=1}^\ell H^{ij}\frac{\delta}{\delta u^j}\left(\int\tr h\right)
		\end{equation}
		with $H$ a Hamiltonian operator and for a local functional $\int\tr h$ which is called ``the Hamiltonian'' of the system.

\section{Double Poisson algebras and Hamiltonian operators}\label{sec:dPA}
Van Den Bergh gave an axiomatization of noncommutative Poisson geometry in terms of \emph{double Poisson algebras} \cite{vdb}; in analogy with the connection between (classical) Poisson geometry and (commutative) Hamiltonian ODEs, they provide an effective framework to study noncommutative ODEs.

The theory of double Poisson vertex algebras is a formalism for noncommutative PDEs developed by De Sole, Kac and Valeri \cite{dskv15} that closely follows Van De Bergh's approach; the same formalism for noncommutative differential-difference system has not been discussed yet, at the best of our knowledge; we extend it along the lines of multiplicative PVAs \cite{dskvw19} and present its axioms in Section \ref{ssec:dmPVA}.

In this section, without claiming to be exhaustive, we illustrate how the formalism we use, and in particular the defining property for Poisson operators \eqref{eq:HamProp}, is equivalent to the notion of double Poisson (vertex) algebras. 

For simplicity, we start with the ultralocal case, which coincides with the original Van Den Bergh's notion of double Poisson algebras. For our short exposition of double Poisson algebras we broadly follow \cite{vdb}; since it is easier to generalise it to the vertex case using some of the notation of \cite{dskv15}, we take some definition from that paper. In our exposition we omit proofs and technical lemmas, that can be found in the aforementioned \cite{dskv15,vdb}.
\subsection{Double derivations and brackets}\label{ssec:double}
Let us consider the linear associative algebra $\A$, obtained as the quotient of the free algebras $\R\langle u^i_n,(u^i_n)^{-1}\rangle$, $i=1,\ldots\ell$, $n\in\Z$ by the two-sided ideals $\langle u^i_n (u^i_n)^{-1}-1\rangle$ and $\langle (u^i_n)^{-1}u^i_n -1\rangle$. This is tantamount to considering the symbol $(u^i_n)^{-1}$ as left and right inverse of the symbol $u^i_n$. We regard elements of $\A$ as noncommutative Laurent polynomials. The product, associative but not commutative, in $\A$ is denoted by simple juxtaposition. We endow $\A^{\otimes n}$ with the structure of \emph{outer} bimodule
\begin{equation}
b(a_1\otimes\cdots\otimes a_n)c=ba_1\otimes \cdots\otimes a_nc
\end{equation}
and $(n-1)$ left and right module structures
\begin{gather}
b\star_i(a_1\otimes a_2\otimes\cdots\otimes a_n)=a_1\otimes \cdots\otimes a_i\otimes ba_{i+1}\otimes\cdots\otimes a_n,\\
(a_1\otimes a_2\otimes\cdots\otimes a_n)\star_i c=a_1\otimes \cdots \otimes a_{n-i}c\otimes \cdots\otimes a_n.
\end{gather}
Note in particular that for $\A^{\otimes 2}$ we have $a\star_1(b\otimes c)=b\otimes ac$ and $(a\otimes b)\star_1c=ac\otimes b$, so that $\star_1=\star$ endows $\A\otimes \A$ with the structure of inner bimodule. Choosing a similar notation for the product $\A\times \A^{\otimes n}\to \A^{\otimes(n+1)}$ we have
\begin{gather}
b\otimes(a_1\otimes\cdots\otimes a_n)=b\otimes a_1\otimes\cdots\otimes a_n,\\
(a_1\otimes\cdots\otimes a_n)\otimes c=a_1\otimes\cdots\otimes a_n\otimes c,\\
b\otimes_i(a_1\otimes a_2\otimes\cdots\otimes a_n)=a_1\otimes \cdots\otimes a_i\otimes b\otimes a_{i+1}\otimes\cdots\otimes a_n,\\
(a_1\otimes a_2\otimes\cdots\otimes a_n)\otimes_i c=a_1\otimes \cdots \otimes a_{n-i}\otimes c\otimes \cdots\otimes a_n.
\end{gather}

An element of $\A\otimes\A$, in general, is of the form $B=\sum_i B'_{(i)}\otimes B''_{(i)}$, for $B_{(i)}^{\prime,\prime\prime}\in\A$. We adopt the so-called Sweedler's notation (see for instance \cite[pag.~35]{u11}), which keeps the summation implicit, and write $B=B'\otimes B''$. This allows us to define an associative product $\bullet$ in $\A\otimes \A$ given by
\begin{equation}\label{eq:bullproduct}
B\bullet C=B'C'\otimes C''B''.
\end{equation} 
We then introduce the \emph{multiplication map} $m\colon \A\otimes \A\to \A$, $m(a\otimes b):=ab$. We also denote by $\circ$ the application (or composition) of operators from the left in $\A\otimes\A$, by $(A\otimes B)\circ(C\otimes D)=A(C)\otimes B(D)$.

Moreover, let us denote the permutation of the factors in $\A\otimes \A$ by $\sigma$, namely $(B'\otimes B'')^\sigma=B''\otimes B'$; it is an antiautomorphism of the bullet product:
\begin{equation}
(B\bullet C)^\sigma=C^\sigma\bullet B^\sigma.
\end{equation}

Similarly, we denote the cyclic permutations of the factors of an element of $A^{\otimes n}$ with $\tau$:
\begin{equation}
\tau(a_1\otimes\a_2\otimes \cdots\otimes a_n)=a_n\otimes a_1 \otimes \cdots \otimes a_{n-1}.
\end{equation}
A \emph{$n$-fold} (double, triple, \dots) \emph{derivation} is a linear map $\A\to \A^{\otimes n}$ fulfilling the Leibniz property $D(ab)=D(a)b+aD(b)$. In particular, we define a noncommutative version of the partial derivative that is a double derivation:
\begin{align}\label{eq:doubleder}
\frac{\dev}{\dev u^l}\left(u^{i_1}u^{i_2}\cdots u^{i_p}\right)&=\sum_{k=1}^p\delta_{i_k,l} u^{i_1}\cdots u^{i_{k-1}}\otimes u^{i_{k+1}}\cdots u^{i_p},\\
\frac{\dev }{\dev u^l}(u^l)^{-1}&=-(u^l)^{-1}\otimes (u^l)^{-1}.
\end{align}
Using the Sweedler's notation, we denote the (sum of the) two factors produced by the double derivative as
\begin{equation}
\frac{\dev f}{\dev u^l}=\left(\frac{\dev f}{\dev u^l}\right)'\otimes\left(\frac{\dev f}{\dev u^l}\right)''.
\end{equation}

Note that $m\circ\dev_u$ is a derivation on $\A$. Indeed, we can read the action of an evolutionary vector field using a formula involving this double partial derivative, which closely resembles the standard formulae in the theory of evolutionary equations.

Let $\{X^i\}_{i=1}^\ell$ be the characteristics of an evolutionary vector field. Then its action on a difference polynomial $f$ is given by
\begin{equation}\
X(f)=\sum_{i,n}m\left(\left(\cS^n X^i\right)\star \frac{\dev f}{\dev u^i_n}\right).
\label{eq:evvfield-action}
\end{equation}
The validity of the formula can be easily checked; it is more interesting comparing equation \eqref{eq:evvfield-action} with formula \eqref{eq:defvarder}, holding true in $\F$. We have
\begin{equation}
\begin{split}
\int\tr X(f)&=\int\tr\sum_{i,n} m\left(\left(\cS^n X^i\right)\star\frac{\dev f}{\dev u^i_n}\right)=\int\tr\sum_{i,n} m\left(X^i\star\cS^{-n}\frac{\dev f}{\dev u^i_n}\right)\\
&=\int\tr\sum_{i,n}\left(\cS^{-n}\frac{\dev f}{\dev u^i_n}\right)'X^i\left(\cS^{-n}\frac{\dev f}{\dev u^i_n}\right)''\\
&=\int\tr\sum_{i,n}\left(\cS^{-n}\frac{\dev f}{\dev u^i_n}\right)''\left(\cS^{-n}\frac{\dev f}{\dev u^i_n}\right)'X^i,
\end{split}\label{eq:varderproof}
\end{equation}
from which one can read an explicit formula for the variational derivative, namely
\begin{equation}\label{eq:varder_expl}
\frac{\delta f}{\delta u^i}=\sum_{n}\cS^{-n}m\left(\frac{\dev f}{\dev u^i_n}\right)^\sigma.
\end{equation}
We simply wrote $\delta f\,/\,\delta u$ to denote the RHS of \eqref{eq:varder_expl}, without direct reference to the local functional $\int\tr f$, because, both explicitly and as a consequence of the well-posedness of the definition of variational derivative, we have
\begin{align}
\frac{\delta}{\delta u^i}(\cS-1)f&=0,& 
\frac{\delta}{\delta u^i}[f,g]&=0. \label{eq:varder12}
\end{align}
Note that the Sweedler's notation in \eqref{eq:varderproof} leaves a further sum implicit (e.g., $\dev_u u^3=1\otimes u^2+u\otimes u+u^2\otimes 1$).

\begin{lemma}\label{lem:comm_vf}
The commutator of two evolutionary vector fields is a vector field with characteristics
\begin{equation}\label{eq:comm_vf}
[X,Y]^i=\sum_{n,j}\left[\left(\frac{\dev Y^i}{\dev u^j_n}\right)'\left(\cS^nX^j\right)\left(\frac{\dev Y^i}{\dev u^j_n}\right)''-\left(\frac{\dev X^i}{\dev u^j_n}\right)'\left(\cS^nY^j\right)\left(\frac{\dev X^i}{\dev u^j_n}\right)''\right].
\end{equation}
\end{lemma}
\begin{proof}
Let us omit the summation symbol for repeated indices with their natural boundaries ($j=1,\ldots,\ell$ and $n\in\Z$), and introduce a shorthand notation for the derivative $\dev_{u^i_p}f:=\dev f/\dev u^i_p$.

A straightforward computation leads to
\begin{equation}\label{eq:comm_vf_pf1}
X(Y(f))=m\left(\left(\cS^nX^i\right)\star\frac{\dev}{\dev u^i_n}\left[\left(\dev_{u^j_m} f\right)'\left(\cS^mY^j\right)\left(\dev_{u^j_m} f\right)''\right]\right).
\end{equation}
When the partial derivative acts on the inside of the square bracket it produces two kinds of terms, with first and second derivatives. We need to prove that the expression with the second derivatives vanishes and that the remaining part acts as an evolutionary vector field.
Let us consider the terms with first derivatives only. We have
\begin{equation}
\left(\frac{\dev}{\dev u^i_n}(\cS^mY^j)\right)'\left(\cS^n X^i\right)\left(\frac{\dev}{\dev u^j_n}(\cS^m Y^j)\right)''\star\dev_{u^j_m} f,
\end{equation}
that, since the partial derivatives have the following commutation rule with the shift operator
$$
\dev_{u^i_n}(\cS f)=\cS\left(\dev_{u^i_{n-1}} f\right),
$$
can be rewritten as
\begin{equation}
\cS^m\left[\left(\dev_{u^i_{n-m}} Y^j\right)'(\cS^{n-m}X^i)\left(\dev_{u^i_{n-m}} Y^j\right)''\right]\star\dev_{u^j_m} f.
\end{equation}
Subtracting the same expression with $X$ and $Y$ exchanged we get \eqref{eq:comm_vf}.

We must now prove that the terms with the second derivatives vanish. They appear from \eqref{eq:comm_vf_pf1} when the derivative acts on one of the factors produced by the double derivative of $f$ -- when we subtract the same expression with $X$ and $Y$ exchanged, this produces two families of terms. One of these is
\begin{multline}\label{eq:comm_vf_pf4}
\left(\frac{\dev (\dev_{u^j_m} f)'}{\dev u^i_n}\right)'\left(\cS^n X^i\right)\left(\frac{\dev (\dev_{u^j_m} f)'}{\dev u^i_n}\right)''\left(\cS^mY^j\right)(\dev_{u^j_m} f)''\\
-(\dev_{u^i_n} f)'\left(\cS^nX^i\right)\left(\frac{\dev (\dev_{u^i_n} f)''}{\dev u^j_m}\right)'\left(\cS^m Y^j\right)\left(\frac{\dev (\dev_{u^i_n} f)''}{\dev u^j_m}\right)''.
\end{multline}
Note that the expression \eqref{eq:comm_vf_pf4} can be written as
\begin{equation}\label{eq:comm_vf_pf5}
m\left((m\otimes 1)\left(\left(1\otimes \cS^n X^i\otimes \cS^m Y^j\right)\left(\left(\frac{\dev}{\dev u^i_n}\otimes 1\right)\circ\frac{\dev f}{\dev u^j_m}-\left(1\otimes\frac{\dev}{\dev u^j_m}\right)\circ\frac{\dev f}{\dev u^i_n}\right)\right)\right),
\end{equation}
where the product from the left acts on each factor of $\A^{\otimes 3}$ and we apply the multiplication map twice to obtain an element in $\A$. From \cite[Lemma 3.7]{dskv15} we have that the expression in the innermost bracket vanishes (the double derivatives act on an algebra of differential polynomials, but the definition itself is the same as \eqref{eq:doubleder}). The same happens to the other terms involving second derivatives.
\end{proof}

\subsection{Double Poisson algebras and ultralocal Poisson brackets}\label{ssec:dPA}
In this section we briefly recall the notion of double Poisson algebra \cite{dskv15,vdb}, introduced to describe the Hamiltonian structure of noncommutative ODEs. Our aim is establishing its equivalence with the notion of ultralocal Poisson operators (see Definition \ref{def:Pois}) and the noncommutative Poisson bracket they define.

Let us now focus on the \emph{ultralocal} case, namely we disregard the shift operation and drop the shifted variables. Let $\A_0\subset\A$ be the space of the ultralocal Laurent polynomials, generated by the variables $u^i_0:=u^i$ only. We introduce the main object in the theory of double Poisson algebras, the \emph{double} bracket.
\begin{definition}
A \emph{$n$-fold (double, triple, \dots) bracket} is a $n$-linear map 
$$\ldb -,-,\cdots,-\rdb\colon\underbrace{\A_0 \times \cdots\times \A_0}_{n \text{ times}} \to \A_0^{\otimes n},$$
 which is a $n$-fold derivation in the last entry
\begin{equation}
\ldb a_1,\ldots,a_{n-1},bc\rdb=b\ldb a_1,\ldots,c\rdb+\ldb a_1,\ldots,b \rdb c
\end{equation}
and it is cyclically skewsymmetric, namely
\begin{equation}
\tau \ldb -,-,\ldots,-\rdb \tau^{-1}=(-1)^{n+1}\ldb -,-,\ldots,-\rdb.
\end{equation}
\end{definition}
For example, we have $\ldb a ,b\rdb=-\ldb b,a\rdb^\sigma$, $ \ldb a,b,c\rdb=\tau \ldb b,c,a\rdb$. In particular, a double bracket is also a derivation in the first entry for the inner bimodule structure of $\A_0\otimes \A_0$, namely
\begin{equation}
\ldb ab,c\rdb=a\star\ldb b,c\rdb+\ldb a,c\rdb\star b.
\end{equation}
To write the double bracket between elements of $\A_0$ in terms of the brackets between its generators we have the explicit formula (called \emph{master formula} in \cite{dskv15})
\begin{equation}\label{eq:master-ul}
\ldb a,b\rdb=\sum_{i,j=1}^\ell \frac{\dev b}{\dev u^j}\bullet \ldb u^i,u^j\rdb \bullet \left(\frac{\dev a}{\dev u^i}\right)^\sigma.
\end{equation}

Now let $a \in \A_0$, $B=b_1 \otimes  b_2 \in \A_0^{\otimes 2}$. We introduce the additional notation
\begin{align}
\ldb a, B\rdb_L&=\ldb a, b_1\rdb \otimes b_2, &\ldb a, B\rdb_R&=b_1\otimes \ldb a, b_2\rdb;\\
\ldb B, a\rdb_L&=\ldb b_1 ,a\rdb\otimes_1 b_2, & \ldb B,a\rdb_R&=b_1\otimes_1\ldb b_2,a\rdb. 
\end{align}

\begin{definition}
A linear associative algebra $\A_0$, endowed with a double bracket $\ldb-,-\rdb$, is a \emph{double Poisson algebra} if the triple bracket
\begin{equation}\label{eq:tripleb}
\begin{split}
\ldb a,b,c\rdb&:=\ldb a,\ldb b,c\rdb\rdb_L+\tau\ldb b,\ldb c,a\rdb\rdb_L+\tau^2\ldb c,\ldb a,b\rdb\rdb_L\\
&=\ldb a,\ldb b,c\rdb\rdb_L-\ldb b,\ldb a,c\rdb\rdb_R-\ldb\ldb a,b\rdb,c\rdb_L
\end{split}
\end{equation}
vanishes for any $a,b,c \in \A_0$.
\end{definition}
From the definition and properties of a triple bracket \cite{vdb}, this is equivalent to the vanishing of the brackets for all the triples of generators of $\A_0$.

We associate an ultralocal operator (see Definition \ref{def:ul}) to a double Poisson bracket and vice versa. We observe that the bullet product \eqref{eq:bullproduct} has the same structure of the composition of multiplication operators: we can then identify a multiplication operator as an element of $\A_0\otimes \A_0$, by
\begin{equation}
\l_f\mapsto f\otimes 1,\quad\r_g\mapsto 1\otimes g,\quad \l_f\r_g\mapsto f\otimes g,
\end{equation}
and the composition of such operators by the multiplication \eqref{eq:bullproduct} on the left.
Ultralocal Hamiltonian operators are compositions and linear combinations of left and right multiplication operators only. In the scalar case, $\A_0$ has a single generator $u$, hence all the double brackets are defined by
\begin{equation}\label{eq:dbrack}
\ldb u,u\rdb=\sum_\alpha f_\alpha\otimes g_\alpha
\end{equation} for some $f_\alpha, g_\alpha\in \A_0$. From a double bracket as in \eqref{eq:dbrack} we define the multiplication operator $K=\sum_\alpha \l_{f_\alpha}\r_{g_\alpha}$, which is skewsymmetric. Indeed, from the skewsymmetry of the double bracket we have
\begin{equation}
\ldb u,u\rdb=-\ldb u,u\rdb^\sigma=-\sum_{\alpha}g_\alpha\otimes f_\alpha
\end{equation}
corresponding to the operator $-K^\dagger=-\sum_\alpha \l_{g_\alpha}\r_{f_\alpha}$. Similarly, for an algebra $\A_0$ with $\ell$ generators we can define an $\ell\times\ell$ matrix of operators, whose entries $K^{ij}$ correspond to the brackets between $u^j$ and $u^i$ (note the exchange of indices, which is the convention adopted in \cite{dskv15}). Equivalently, given a skewsymmetric multiplication operator $K$ we can define a double bracket, identifying the bracket between the generators with its entries and extending it to the full algebra $\A_0$ by means of \eqref{eq:master-ul}. Indeed, the bracket on $\A_0$ is uniquely determined by the bracket between generators.

Observe the striking analogy between the double bracket structure and the bivector defined by $K$. If we denote $K^{ij}$ as $\sum_\alpha\l_{K^{(\alpha)ij}_L}\r_{K^{(\alpha)ij}_R}$ we have on one hand
\begin{align}
P&=\frac12 \tr\sum_{i,j}\sum_\alpha \theta_i K^{(\alpha)ij}_L\theta_j K^{(\alpha)ij}_R\\&=-\frac12 \tr \sum_{i,j}\sum_\alpha\theta_j K^{(\alpha)ij}_R\theta_iK^{(\alpha)ij}_L=-\frac12\tr\sum_{i,j}\sum_\alpha \theta_i K^{(\alpha)ji}_R\theta_jK^{(\alpha)ji}_L,
\end{align}
(note that in the ultralocal case we have dropped the integral operation) and on the other hand
\begin{equation}\label{eq:brackbiv}
\ldb u^j,u^i\rdb = \sum_\alpha K^{(\alpha)ij}_L\otimes K^{(\alpha)ij}_R=-\sum_\alpha K^{(\alpha)ji}_R\otimes K^{(\alpha)ji}_L.
\end{equation}
We can read one expression from the other by replacing the tensor product in \eqref{eq:brackbiv} by $\theta_j$ and multiplying the result on the left by $\theta_i$. 

After these preliminary observations, we are ready to address the equivalence between double \emph{Poisson} brackets and ultralocal Poisson operators.
\begin{proposition}\label{thm:poisdouble-ul}
Let $H$ be an ultralocal operator with entries
\begin{equation}
H^{ij}=\sum_{\alpha}\l_{H^{(\alpha)ij}_{L}}\r_{H^{(\alpha)ij}_{R}}
\end{equation} 
Then $H$ is Poisson if and only if the double bracket defined on generator as
\begin{equation}
\ldb u^i,u^j\rdb =\sum_\alpha H^{(\alpha)ji}_L\otimes H^{(\alpha)ji}_R
\end{equation}
is the bracket of a double Poisson algebra.
\end{proposition}
\begin{proof}
The bracket of a double Poisson algebra is, in particular, skewsymmetric; the skewsymmetry of the bracket defined by $H$ is then equivalent to
\begin{equation}
\sum_\alpha H^{(\alpha)ji}_L\otimes H^{(\alpha)ji}_R=-\sum_{\alpha} H^{(\alpha)ij}_R\otimes H^{(\alpha)ij}_L,
\end{equation}
where the LHS and RHS correspond, respectively, to the entries of the bracket defined  by $\sum_{\alpha}\l_{H^{(\alpha)}_L}\r_{H^{(\alpha)}_R}$ and minus its adjoint. 
On the other hand, the property defining Poisson bivectors \eqref{eq:HamProp} and the vanishing of the triple bracket \eqref{eq:tripleb} are not linear identities. Indeed, taken an operator $H=\sum_\alpha H^{(\alpha)}$ (and the corresponding double bracket $\sum \ldb -,-\rdb_{(\alpha)}$), we have
\begin{equation}\label{eq:JacLin_pf1}
\begin{split}
\mathbf{pr}_{H\Theta} \tr(\theta H \theta)&=\sum_{\alpha,\beta}\mathbf{pr}_{H^{(\alpha)}\Theta} \tr(\theta H^{(\beta)} \theta)
\end{split}
\end{equation}
and
\begin{align}
\ldb a,b,c\rdb&=\sum_{\alpha,\beta}\left[\ldb a,\ldb b,c\rdb_{(\alpha)}\rdb_{(\beta),L}+\tau\left(\ldb b,\ldb c,a\rdb_{(\alpha)}\rdb_{(\beta),L}\right)+\tau^2\left(\ldb c,\ldb a,b\rdb_{(\alpha)}\rdb_{(\beta),L}\right)\right]\\ \label{eq:JacLin_pf2}
&=\sum_{\alpha,\beta}\ldb a,b,c\rdb_{(\alpha,\beta)}.
\end{align}

Each term of the summation \eqref{eq:JacLin_pf1} is
\begin{equation}\label{eq:HamPropUL}
\mathbf{pr}_{H^{(\alpha)}\Theta}\tr\left(\theta_i H^{(\beta)ij}\theta_j\right)=\mathbf{pr}_{H^{(\alpha)}\Theta}\tr\left(\theta_i H^{(\beta)ij}_L\theta_jH^{(\beta)ij}_R\right).
\end{equation}
A direct computation for $\text{\eqref{eq:HamPropUL}}$ gives, up to a sign and denoting the characteristics of $\mathbf{pr}_{H^{(\alpha)}\Theta}$ by $H^{(\alpha)lk}\theta_k=H^{(\alpha)lk}_L\theta_kH^{(\alpha)lk}_R$ and the partial derivative with respect to $u^l$ as $\dev_l$,
\begin{multline}\label{eq:HamPropULpf-1}
\tr\left[\theta_i\left(\dev_l H^{(\beta)ij}_L\right)'H^{(\alpha)lk}_L\theta_kH^{(\alpha)lk}_R\left(\dev_l H^{(\beta)ij}_L\right)''\theta_jH^{(\beta)ij}_R-\theta_i H^{(\beta)ij}_L\theta_j\left(\dev_l H^{(\beta)ij}_R\right)'H^{(\alpha)lk}_L\theta_kH^{(\alpha)lk}_R\left(\dev_l H^{(\beta)ij}_R\right)''\right]\\
=\tr\left[\theta_i\left(\dev_l H^{(\beta)ik}_L\right)'H^{(\alpha)lj}_L\theta_jH^{(\alpha)lj}_R\left(\dev_l H^{(\beta)ik}_L\right)''\theta_kH^{(\beta)ik}_R+\theta_i H^{(\beta)ji}_R\theta_j\left(\dev_l H^{(\beta)ji}_L\right)'H^{(\alpha)lk}_L\theta_kH^{(\alpha)lk}_R\left(\dev_l H^{(\beta)ji}_L\right)''\right],
\end{multline}
where we exchanged indices in the first term and used the skewsymmetry of $H^{(\beta)}$ in the second one.

Because of the trace operation, the expression does not vary (up to a constant) if we replace the RHS with an expression obtained taking the cyclic sums over the factors of the form $(\theta_i f,\theta_k g,\theta_j h)$, when we denote each term of the RHS of \eqref{eq:HamPropULpf-1} as $\tr[\theta_if\theta_jg\theta_kh]$. Moreover, in the result we can change the summation indices so that the relative position of $(\theta_i,\theta_j,\theta_k)$ is the same in each monomial. The terms we get can be summed in pairs, so that we obtain
\begin{multline}\label{eq:HamProof}
-\frac32\mathbf{pr}_{H^{(\alpha)}\Theta}\tr(\theta H^{(\beta)}\theta)=\tr\left[\theta_i\left(\dev_l H^{(\beta)ik}_L\right)'H^{(\alpha)lj}_L\theta_jH^{(\alpha)lj}_R\left(\dev_l H^{(\beta)ik}_L\right)''\theta_kH^{(\beta)ik}_R\right.\\
\left.+\theta_iH^{(\beta)ji}_R\theta_j\left(\dev_l H^{(\beta)ji}_L\right)'H^{(\alpha)lk}_L\theta_kH^{(\alpha)lk}_R\left(\dev_l H^{(\alpha)ji}_L\right)''+\theta_iH^{(\alpha)li}_R\left(\dev_l H^{(\beta)kj}_L\right)''\theta_jH^{(\beta)kj}_R\theta_k\left(\dev_l H^{(\beta)kj}_L\right)'H^{(\alpha)li}_L\right].
\end{multline}

On the other hand, let us compare \eqref{eq:HamProof} with $\ldb u^j,u^k,u^i\rdb_{(\beta,\alpha)}$ as given in \eqref{eq:JacLin_pf2}. For the first term we have
\begin{align}
\ldb u^j,\ldb u^k,u^i\rdb_{(\beta)}\rdb_{(\alpha)L}&=\ldb u^j,H^{(\beta)ik}_L\rdb_{(\alpha)}\otimes H^{(\beta)ik}_R=\left(\dev_l H^{(\beta)ik}_L\right)\bullet \left(H^{(\alpha)lj}_L\otimes H^{(\alpha)lj}_R\right)\otimes H^{(\beta)ik}_R\\\label{eq:dJac-1}&=
\left(\dev_l H^{(\beta)ik}_L\right)'H^{(\alpha)lj}_L\otimes H^{(\alpha)lj}_R\left(\dev_l H^{(\beta)ik}_L\right)''\otimes H^{(\beta)ik}_R.
\end{align}
The three factors of \eqref{eq:dJac-1} are exactly the three factors $f,g,h$ of $\theta_i f \theta_j g \theta_k h$ in the first summand of \eqref{eq:HamProof}; computing the full triple bracket produces two more terms that reproduce the second and third summand.

We have then established the equivalence between each term of \eqref{eq:JacLin_pf1} and \eqref{eq:JacLin_pf2} for the triple of generators $(u^j,u^i,u^k)$. The summation over all the pairs $(\alpha,\beta)$ is then equivalent, too. In particular, identity \eqref{eq:HamProof} includes an implicit sum over all the triples of indices $(i,j,k)$. Then the vanishing of each individual term is a sufficient and necessary condition for
\begin{equation}\label{eq:HamProp-UL-pf2}
\mathbf{pr}_{H\Theta}\tr\left(\theta_i H^{ij}\theta_j\right)=0
\end{equation}
and \eqref{eq:HamProp-UL-pf2} is equivalent to the vanishing of \eqref{eq:tripleb} for all the triples of generators.
\end{proof}
\begin{example}[\cite{vdb}]
The operator $H=\l_u^2\r_u-\l_u\r_u^2$ (which is the operator presented in Theorem \ref{thm:ul-scal} with $\alpha=\beta=0$, $\gamma=1$) is Poisson. We can easily show it with the language of double Poisson algebras: the operator corresponds to the double bracket $\ldb u,u\rdb=u^2\otimes u-u\otimes u^2$. The triple bracket in the scalar case is
\begin{align}
\ldb u,u,u\rdb&=(1+\tau+\tau^2)\ldb u, \ldb u,u\rdb\rdb_L\\
&=(1+\tau+\tau^2)\left[\ldb u,u^2\rdb \otimes u-\left(u^2\otimes u-u\otimes u^2\right)\otimes u^2\right]\\
&=(1+\tau+\tau^2)\left[u\left(u^2\otimes u-u\otimes u^2\right)\otimes u+\left(u^2\otimes u-u\otimes u^2\right)u\otimes u\right.\\&\quad\left.-\left(u^2\otimes u-u\otimes u^2\right)\otimes u^2\right]\\
&=(1+\tau+\tau^2)\left(u^3\otimes u\otimes u-u\otimes u^3\otimes u-u^2\otimes u\otimes u^2+u\otimes u^2\otimes u^2\right)\\
&=0.
\end{align}
Since the operator is Poisson, in particular it is Hamiltonian.
\end{example}
The language of double Poisson algebras can be used not only to characterise the operators, but to replace the whole Hamiltonian formalism. The Hamiltonian equations defined by the Hamiltonian structure $H$ and Hamiltonian functional $\tr f$ is, with this language,
\begin{equation}
u^i_t=m\left(\ldb  f,u^i\rdb\right).
\end{equation}
Similarly, the Poisson bracket between two local functionals $\tr f$ and $\tr g$ is
\begin{equation}\label{eq:PoisB-PA-def}
\{\tr f,\tr g\}=-\tr m\left(\ldb f, g \rdb\right).
\end{equation}
The two statements can be easily verified by a direct computation. For instance, let us consider the operator $H^{ij}=\sum_\alpha\l_{H^{(\alpha)ij}_L}\r_{H^{(\alpha)ij}_R}$ and the functional $\tr f$. From \eqref{eq:HamSyst_expl}  and the formula for the variational derivative \eqref{eq:varder_expl}, we have that the characteristics of the Hamiltonian vector field has the same form prescribed by  the master formula \eqref{eq:master-ul}, namely
\begin{equation}
\begin{split}
u^i_t&=m\left(\sum (H^{(\alpha)ij}_L\otimes H^{(\alpha)ij}_R)\bullet \left(\frac{\dev f}{\dev u^j}\right)^\sigma\right)=m\left(\sum H^{(\alpha)ij}_L\left(\frac{\dev f}{\dev u^j}\right)''\otimes\left(\frac{\dev f}{\dev u^j}\right)'H^{(\alpha)ij}_R\right)\\
&=\sum H^{(\alpha)ij}_L\left(\frac{\dev f}{\dev u^j}\right)''\left(\frac{\dev f}{\dev u^j}\right)'H^{(\alpha)ij}_R.
\end{split}
\end{equation}

\subsection{Difference operators, multiplicative double Poisson vertex algebras and Poisson bivectors} \label{ssec:dmPVA}

In this section we present the axioms of the what we call \emph{double multiplicative Poisson vertex algebras} (double multiplicative PVAs). While (standard) double PVAs \cite{dskv15} are tailored for nonabelian PDEs, double multiplicative PVAs are the structures corresponding to the Hamiltonian formalism for nonabelian D$\Delta$Es. Their axioms are modelled on those of multiplicative Poisson vertex algebras, \cite{dskvw19, dskvw20} which describe the structure of \emph{Abelian} differential-difference equations.

\begin{definition}\label{def:dlbracket}
Let $\A$ be the space of difference (Laurent) polynomials with $\ell$ generators as in Section \ref{sec:Ham_formalism}. A multiplicative $\lambda$ bracket is a bilinear operation
\begin{equation}
\ldb-_{\lambda}-\rdb\colon \A\times\A\to\A\otimes\A\lsb\lambda,\lambda^{-1}\rsb
\end{equation}
such that the following properties hold:
\begin{enumerate}
\item $\ldb \cS a_{\lambda} b\rdb=\lambda^{-1}\ldb a_{\lambda} b\rdb$ and $\ldb a_{\lambda}\cS b\rdb=\lambda\cS\ldb a_{\lambda} b \rdb$ (sesquilinearity)
\item $\ldb a_{\lambda} bc\rdb=\ldb a_{\lambda}b\rdb c +b\ldb a_{\lambda}c\rdb$ (Left Leibniz property)
\item $\ldb ab_{\lambda}c\rdb=\ldb a_{\lambda \cS}c\rdb_\to\star b+ ({}_\to a)\star\ldb b_{\lambda\cS}c\rdb$ (Right Leibinz property)
\end{enumerate}
\end{definition}
The notation used for the right Leibniz property means that, for $\ldb a_{\lambda}b\rdb=\sum B(a,b)_{p}'\otimes B(a,b)_{p}''\lambda^p$, we have
\begin{align}
\ldb a_{\lambda \cS} c\rdb_\to \star b&=\sum B(a,c)_p'\left(\cS^p b\right)\otimes B(a,c)_p''\lambda^p,\\
({}_\to a)\star\ldb b_{\lambda\cS}c\rdb&=\sum B(b,c)_p'\otimes\left(\cS^p a\right)B(b,c)''_p\lambda^p.
\end{align}

\begin{definition}\label{def:dmpva}
A double multiplicative PVA is an algebra of difference polynomials $\A$ endowed with a multiplicative $\lambda$ bracket satisfying the additional properties
\begin{enumerate}
\item Skewsymmetry: $\ldb b_{\lambda} a\rdb=-{}_{\to}\ldb a_{(\lambda\cS)^{-1}}b\rdb^\sigma$;
\item Double Jacobi identity $\ldb a_{\lambda}\ldb b_{\mu} c\rdb\rdb_L-\ldb b_{\mu}\ldb a_{\lambda}c\rdb\rdb_R=\ldb \ldb a_{\lambda}b\rdb{}_{\lambda\mu}c\rdb_L$.
\end{enumerate}
The notation for the skewsymmetry property is
\begin{equation}\label{eq:skew-detail}
{}_\to\ldb a_{(\lambda \cS)^{-1}}b \rdb^\sigma:=\sum\cS^{-p}\left(B(a,b)_p''\otimes B(a,b)_p'\right)\lambda^{-p},
\end{equation}
namely the arrow on the left of the bracket denotes that the shift operator acts on the terms of the bracket themselves.
Similarly to the definitions used for double Poisson algebras in Section \ref{ssec:dPA} and adopting the same notation we used for the Leibniz property, the entries of the double Jacobi identity are
\begin{align}
\ldb a_\lambda b\otimes c\rdb_L&:=\ldb a_\lambda b\rdb\otimes c,&\ldb a_\lambda b\otimes c\rdb_R&:=b\otimes\ldb a_\lambda c\rdb,\\
\ldb a\otimes b_\lambda c\rdb_L&:=\ldb a_{\lambda S}c\rdb_\to\otimes_1 b,&\ldb a\otimes b_\lambda c\rdb_R&:=\ldb b_{\lambda S}c\rdb_\to\otimes_1 a.
\end{align}
\end{definition}

Note that the skewsymmetry property, together with the left Leibniz property, implies the right Leibniz property and generalises the notion of double bracket.

The \emph{master formula} for the double multiplicative $\lambda$ bracket is
\begin{equation}\label{eq:master}
\ldb a_\lambda b\rdb=\sum_{i=1}^\ell\sum_{n,m\in \Z} \frac{\dev b}{\dev u^j_m}(\lambda\cS)^m\bullet \ldb u^i{}_{\lambda\cS} u^j\rdb\bullet (\lambda\cS)^{-n}\left(\frac{\dev a}{\dev u^i_n}\right)^\sigma,
\end{equation}
where $\ldb a_{\lambda\cS}b\rdb\bullet (c\otimes d)=B(a,b)'_p(\cS^{p} c)\otimes (\cS^p d)B(a,b)''_p\lambda^p$.

Exactly as we showed for the ultralocal case in Proposition \ref{thm:poisdouble-ul}, there is a one-to-one correspondence between $\lambda$ brackets and difference operators; the skewsymmetry property of the bracket is equivalent to the skewsymmetry of the operator and the double Jacobi identity is equivalent to condition \eqref{eq:HamProp}.

A generic scalar difference operator has the form \eqref{eq:diffop-def}; we also want to consider nonlocal operators, allowing the summation of \eqref{eq:diffop-def} to either run from $N$ to $\infty$, as a Taylor series, or from $-\infty$ to $M$ as a Laurent series. As an example, let us consider an operator of the form $K=J(\cS-1)^{-1}L$ for $J$ and $L$ ultralocal difference operators (this will allow us to focus on a single expansion). Its corresponding double $\lambda$ bracket can be written as
\begin{equation}
\ldb u_{\lambda}u\rdb=\left(J'\otimes J''\right)\bullet \frac{1}{\lambda\cS-1}\left(K'\otimes K''\right),
\end{equation}
which we can expand, alternatively, as
\begin{align}
\ldb u_{\lambda}u\rdb&=\sum_{k=-\infty}^{-1} J'(\cS^k K')\otimes(\cS^k K'')J''\lambda^k& \text{or}&&\ldb u_{\lambda}u\rdb&=\sum_{k=0}^\infty -J'(\cS^k K')\otimes(\cS^k K'')J''\lambda^k.
\end{align}
Note that some difficulties with the definition of skewsymmetry arise from the choice of a direction for the expansion, since the one prescribed in \eqref{eq:skew-detail} is in the opposite direction than the one chosen for $\ldb b_\lambda a\rdb$. For the class of nonlocal operators we consider, as detailed in the proof of the following Theorem \ref{thm:doublePVA-Poiss}, the skewsymmetry of the double $\lambda$ bracket follows from the expansion in the same direction of the bracket associated to its  adjoint operator. For a more detailed account regarding the manipulation of nonlocal operators as double $\lambda$ brackets (including their skewsymmetry) the interested reader shall consult the recent preprint \cite{FV21}, which appeared while this manuscript was under revision.

Let
\begin{equation}\label{eq:matrdiffop-compo}
H^{ij}=\sum_{p,\alpha_p} \l_{H^{(\alpha_p)ij}_L}\r_{H^{(\alpha_p)ij}_R}\cS^p
\end{equation} be the entries of a difference operator-valued matrix for the $\ell$-components case. To such a difference operator we associate a double $\lambda$ bracket defined on the generator $\{u^i\}$ of $\A$
\begin{equation}\label{eq:fromOptoBr}
\ldb u^i_{\lambda }u^j\rdb=\sum_{p,\alpha_p} H^{(\alpha_p)ji}_L\otimes H^{(\alpha_p)ji}_R\lambda^p;
\end{equation}
conversely, given a double $\lambda$ brackets on the generators of $\{u^i\}$ of $\A$, we can define the matrix of difference operators
\begin{equation}
H^{ij}=\ldb u^j_{\lambda}u^i\rdb\Big|_{\lambda=\cS}.
\end{equation}
\begin{theorem}\label{thm:doublePVA-Poiss}
The $\lambda$ bracket associated to a skewsymmetric difference operator according to \eqref{eq:fromOptoBr} is the bracket of a double multiplicative PVA if and only if the operator $H=\left(H^{ij}\right)_{i,j=1}^\ell$ with entries of the form \eqref{eq:matrdiffop-compo} is Poisson.
\end{theorem}
\begin{proof}
The equivalence between the skewsymmetry of $H$ and the double bracket is due to an elementary computation. The double bracket associated to the entry $H^{ij}$ is
\begin{equation}
\sum H^{(\alpha_p)ji}_{R}\otimes H^{(\alpha_p)ji}_{L}\lambda^p=-\sum (\lambda\cS)^{-p}\left[H^{(\alpha_p)ij}_{L}\otimes H^{(\alpha_p)ij}_{R}\right],
\end{equation}
which indeed corresponds to the difference operator
\begin{equation}
-\sum\cS^{-p}\,\l_{H^{(\alpha_p)ji}_{L}}\r_{H^{(\alpha_p)ji}_{R}}=-(H^{ji})^\dagger.
\end{equation}
Note that, if the difference operator is nonlocal, then the skewsymmetry property should be checked by expanding $\ldb a_{\lambda}b\rdb$ and ${}_\to\ldb b_{(\lambda\cS)^{-1}}a\rdb$ in the same direction.

Establishing the equivalence between the double Jacobi identity and the Poisson condition for a difference operator requires some long but straightforward computations. They have the same structure as in the ultralocal case: we exhibit them for a scalar difference operator of the form
\begin{equation}\label{eq:diffop-dmpva-pf1}
H=\sum\l_{H_L^{(\alpha_p)}}\r_{H_R^{(\alpha_p)}}\cS^p
\end{equation}
in Appendix \ref{app:equivPVAJac}. The multi-component case behaves in the same way.
\end{proof}

As for the ultralocal case, we can use the double $\lambda$ brackets to describe the Hamiltonian action and the Poisson brackets defined by the operator $H$, by the identities
\begin{equation}
u^i_t=m\left(\ldb f_\lambda u^i\rdb\right)\Big|_{\lambda=1}
\end{equation}
and
\begin{equation}\label{eq:dmpva-bracket}
\left\{\smallint\tr f,\smallint\tr g\right\}=-\int\tr m\left(\ldb f_\lambda g\rdb\right)\big|_{\lambda=1}.
\end{equation}
The well-posedness of the two identities above, and in particular that the RHS of \eqref{eq:dmpva-bracket} defines a Poisson bracket on $\F$ can be easily proved, following the same lines of the proof of the similar result for double Poisson vertex algebras (namely, the partial differential case) given in \cite[Theorem 3.6]{dskv15}.

We sketch here the proof of the well-posedness of the definition of the bracket with respect to the integral operation (namely that replacing $f$ (or $g$) with $\cS f$ (or $\cS g$) in the RHS of \eqref{eq:dmpva-bracket} does not affect the result, and show that the formula \eqref{eq:defBra} coincides with the one obtained starting from the double $\lambda$ bracket.
\begin{proposition}
Let $\ldb-_\lambda -\rdb$ be the $\lambda$ bracket of a multiplicative double Poisson vertex algebra. Then
\begin{equation}\label{eq:wellJac}
\int\tr m\left(\ldb \cS f_{\lambda} g\rdb\right)\Big|_{\lambda=1}=
\int\tr m\left(\ldb  f_{\lambda}\cS g\rdb\right)\big|_{\lambda=1}=\int\tr m\left(\ldb  f_{\lambda} g\rdb\right)\big|_{\lambda=1}.
\end{equation}
\end{proposition}
\begin{proof}
Using the sesquilinearity properties we have
\begin{align}
\int\tr m\left(\ldb \cS f_{\lambda} g\rdb\right)\Big|_{\lambda=1}&=\int\tr m\left(\lambda^{-1}\ldb f_{\lambda} g\rdb\right)\Big|_{\lambda=1},\\
\int\tr m\left(\ldb  f_{\lambda}\cS g\rdb\right)\Big|_{\lambda=1}&=\int\tr m\left(\cS\lambda\ldb  f_{\lambda} g\rdb\right)\Big|_{\lambda=1}=\int\tr m\left(\lambda\ldb  f_{\lambda} g\rdb\right)\Big|_{\lambda=1}.
\end{align}
Setting $\lambda=1$ we have \eqref{eq:wellJac}.
\end{proof}
\begin{proposition}\label{thm:bracketsame}
Let $H$ be a Poisson difference operator of the form \eqref{eq:matrdiffop-compo} and let $\ldb u^i_{\lambda} u^j\rdb$ the corresponding $\lambda$ bracket of a multiplicative double Poisson vertex algebra as in  \eqref{eq:fromOptoBr}. Then the Poisson bracket defined by $H$ as in \eqref{eq:defBra} is equal to \eqref{eq:dmpva-bracket}.
\end{proposition}
\begin{proof}
Using the master formula \eqref{eq:master} with \eqref{eq:dmpva-bracket} we obtain
\begin{equation}
\begin{split}
-\left\{\smallint\tr f,\smallint\tr g\right\}&=\int\tr \left[\left(\dev_{u^j_n} g\right)'\left(\cS^{m} H^{(\alpha_p)ji}_L\right)\left[\cS^{m+p-n}\left(\dev_{u^i_n} f\right)''\left(\dev_{u^i_n} f\right)'\right]\left(\cS^m H^{(\alpha_p)ji}_R\right)\left(\dev_{u^j_n} g\right)''\right]\\
&=\int\tr \left[\left(\dev_{u^j_n} g\right)'\left(\cS^{m} H^{(\alpha_p)ji}_L\right)\left(\cS^{m+p}\frac{\delta f}{\delta u^i}\right)\left(\cS^m H^{(\alpha_p)ji}_R\right)\left(\dev_{u^j_n} g\right)''\right]\\
&=\int\tr\left[\left[\cS^{-m}\left(\dev_{u^j_n} g\right)''\left(\dev_{u^j_n} g\right)'\right]H^{(\alpha_p)ji}_L\left(\cS^{p}\frac{\delta f}{\delta u^i}\right) H^{(\alpha_p)ji}_R\right]\\
&=\int\tr\left[\frac{\delta g}{\delta u^j}H^{(\alpha_p)ji}_L\left(\cS^p\frac{\delta f}{\delta u^i}\right)H^{(\alpha_p)ji}_R\right].
\end{split}
\end{equation}
From the skewsymmetry of the operator $H$, this turns out to be the same as
\begin{equation}
\left\{\smallint\tr f,\smallint\tr g\right\}=\int\tr\left[\left(\cS^{-p}\frac{\delta f}{\delta u^i}\right)\left(\cS^{-p}H^{(\alpha_p)ij}_L\right)\frac{\delta g}{\delta u^j}\left(\cS^{-p}H^{(\alpha_p)ij}_R\right)\right].
\end{equation}
Finally, by shifting all the integrand by $p$, we obtain
\begin{equation}
\left\{\smallint\tr f,\smallint\tr g\right\}=\int\tr\frac{\delta f}{\delta u^i}H^{(\alpha_p)ij}_L\left(\cS^p\frac{\delta g}{\delta u^j}\right)H^{(\alpha_p)ij}_R=\int\tr \frac{\delta f}{\delta u^i}H^{ij}\left(\frac{\delta g}{\delta u^j}\right)
\end{equation}
as in the definition we gave in \eqref{eq:defBra}.
\end{proof}
Proposition \ref{thm:bracketsame} establishes the equivalence between the language of Poisson operators and double multiplicative PVAs not only at the level of the operators (namely, between Poisson operators and $\lambda$ brackets), but also at the level of the Poisson brackets themselves.
\section{A path to noncommutative Poisson geometry}\label{sec:geom}
The $\theta$ formalism, introduced as a mere computational tool in Section \ref{sec:Ham_formalism}, deserves to be investigated in further detail. In this section we  find a characterisation of Poisson bivectors in terms of a noncommutative version of the Schouten torsion -- this allows us to exploit the standard machinery of Poisson geometry to treat Hamiltonian and integrable systems, even defining the Poisson-Lichnerowicz complex for nonabelian Poisson manifolds and its cohomology.

\subsection{$\theta$ formalism and functional polyvector fields}\label{ssec:theta}
The complex of functional polyvector fields, which is presented in \cite{CW19,Kup85} for the difference Abelian case, generalises naturally to the nonabelian case we are dealing with; the $\theta$ formalism for the nonabelian differential case is presented in \cite{OS98} and the one for the difference case in \cite{cw19-2}.

The aim of this section is reviewing the $\theta$ formalism used in Section \ref{sec:Ham_formalism} and, in analogy with the well-established theory for the commutative case, establishing the isomorphism between polynomials in $\theta$ and polyvector fields we implicitly use when representing skewsymmetric operators with degree 2 polynomials (as in Equation \eqref{eq:defBiv}), providing explicit formulae for their Schouten brackets. 

\begin{definition}A local $p$-vector field is a $p$-alternating map from $\F^p$ to $\F$; it is then of the form
\begin{equation}
B(F_1,\dots,F_p)=\int\tr\left[B^{i_1,\dots,i_p}_{(1)n_1,\ldots,n_p}\left(\cS^{n_1}\frac{\delta F_1}{\delta u^{i_1}}\right)B^{i_1,\dots,i_p}_{(2)n_1,\ldots,n_p}\cdots B^{i_1,\dots,i_p}_{(p)n_1,\ldots,n_p}\left(\cS^{n_p}\frac{\delta F_p}{\delta u^{i_p}}\right)\right]\label{eq:defpolyv}
\end{equation}
with $B_{(r)}$ in $\A$ such that $B(\sigma(F_1),\ldots,\sigma(F_p))=(-1)^{|\sigma|}B(F_1,\dots,F_p)$.
\end{definition}
A $0$-vector field is clearly just an element of $\F$, namely a local functional. A $1$-vector field has the form
$$
X(F)=\int\tr X^i_n\cS^n\frac{\delta F}{\delta u^i},\qquad F\in\F
$$
which is equivalent, under the integral and cyclic permutation, to the action of the (sum of) evolutionary vector field of characteristics $\cS^{-n}X^i_n$ on $F$ (see Equation \eqref{eq:varderproof}): the notion of local $1$-vector fields coincides with the notion of evolutionary vector fields on $\F$.
More interestingly, let us consider a local 2-vector field:
\begin{equation}
\begin{split}
B(F,G)&=\int\tr B^{ij}_{(1)mn}\left(\cS^{m}\frac{\delta F}{\delta u^i}\right)B^{ij}_{(2)mn}\left(\cS^{n}\frac{\delta G}{\delta u^j}\right)\\
&=\int\tr\frac{\delta F}{\delta u^i}\tilde{B}^{ij}_{(2)n'} \left(\cS^{n'}\frac{\delta G}{\delta u^j}\right)\tilde{B}^{ij}_{(1)n'}
\end{split}
\end{equation}
where $n'=n-m$, $\tilde{B}^{ij}_{(1)n'}=\cS^{-m}B^{ij}_{(1)mn}$ and $\tilde{B}^{ij}_{(2)n'}=\cS^{-m}B^{ij}_{(2)mn}$, by the cyclic property of the trace and integrating by parts to get rid of $\cS^{m}$. It matches the definition of a bracket given in \eqref{eq:defBra}, for an operator $K$ of the form $\r_{\tilde{B}_{(1)n'}}\l_{\tilde{B}_{(2)n'}}\cS^{n'}$. We also observe that the skewsymmetry property for a bivector field,
$B(F,G)=-B(G,F)$, can be written as
\begin{equation}
\begin{split}
\int\tr\frac{\delta F}{\delta u^i}B^{ij}_{(2)n}\left(\cS^{n}\frac{\delta G}{\delta u^j}\right)B^{ij}_{(1)n}&=-\int\tr\frac{\delta G}{\delta u^j}B^{ji}_{(2)n'}\left(\cS^{n'}\frac{\delta F}{\delta u^i}\right)B^{ji}_{(1)n'}\\
&=-\int\tr \frac{\delta F}{\delta u^i}\cS^{-n'}\left(B^{ji}_{(1)n'}\frac{\delta G}{\delta u^j}B^{ji}_{(2)n'}\right),
\end{split}
\end{equation}
which is exactly the skewsymmetry of the aforementioned operator $K$. One can read how the requirement that a polyvector field is an alternating map affects the coefficients in the same fashion. Moreover, since in $\F$ one can always integrate by parts to get rid of the shift operator acting on one (``the first'') argument of the polyvector field, it is apparent that a local $p$-vector field is defined by a totally skewsymmetric difference operator with $p-1$ arguments.

Let us now consider the spaces $\hA$ and $\hF$ introduced in Definition \ref{def:thetadens} and \ref{def:thetafield}; the $\theta$ \emph{variational derivative} can be defined, in analogy with the standard one given in \eqref{eq:varder_expl}, as
\begin{equation}
\frac{\delta f}{\delta \theta_i}:=\cS^{-n}m\left(\frac{\dev f}{\dev \theta_{i,n}}\right)^\sigma
\end{equation}
and it satisfies the same properties \eqref{eq:varder12}.

\begin{proposition}
The space $\hF$ is isomorphic to the space of local polyvector fields; in particular, an element of $\hF^p$ is in a one-to-one correspondence with a local $p$-vector field.
\end{proposition}
\begin{proof}
The lines of the proof of the analogue theorem for the commutative case \cite[Proposition 2]{CW19} are still valid in the noncommutative setting. Given an element $\int \tr B\in\hF^p$, the corresponding $p$-vector field is given by
\begin{equation}
\tilde{B}(F_1,\ldots,F_p)=\int\tr m\left[\frac{\dev}{\dev\theta_{i_p,n_p}}\cdots m\left(\frac{\dev}{\dev\theta_{i_1,n_1}}B\star\cS^{n_1}\frac{\delta F_1}{\delta u^{i_1}}\right)\star\cdots\cS^{n_p}\frac{\delta F_p}{\delta u^{i_p}}\right].
\end{equation}
Conversely, given a $p$-vector field of the form \eqref{eq:defpolyv}, the corresponding element in $\hF^p$ is
\begin{equation}\label{eq:pvect-gen}
\frac{1}{p!}\int\tr\left[B^{i_1,\dots,i_p}_{(1)n_1,\ldots,n_p}\theta_{i_1,n_1}B^{i_1,\dots,i_p}_{(2)n_1,\ldots,n_p}\cdots B^{i_1,\dots,i_p}_{(p)n_1,\ldots,n_p}\theta_{i_p,n_p}\right].
\end{equation}
\end{proof}
Note that a $p$-vector of the form \eqref{eq:pvect-gen} can always be integrated by parts and cyclically permuted in such a way that it can be written as 
$$
\int\tr\theta_{i_1}K\left(\theta_{i_2},\ldots,\theta_{i_p}\right),
$$
with $K$ a $(p-1)$ difference operator acting on $\theta$'s. 

As anticipated in Section \ref{sec:Ham_formalism}, the existence of the isomorphism between the space of local polyvector fields and $\hF$ allows us to perpetrate an abuse of language and identify the former with the latter one. This leads us to introduce a tailored notion of Schouten bracket on $\hF$. In Section \ref{ssec:sch-ul} we present the construction for double Poisson algebras, namely in the case of ODEs; in Section \ref{ssec:sch-l} we repeat the construction for multiplicative double Poisson vertex algebras, namely those tailored on nonabelian differential-difference systems. In both cases, we start defining a double bracket on $\hA$, and then use it to obtain a well-defined bracket on $\hF$ satisfying the axioms of a Gerstenhaber algebra, namely proving that it is indeed the Schouten bracket for the corresponding local polyvector fields.
\subsection{Schouten brackets for nonabelian ODEs}\label{ssec:sch-ul}
Let us consider the ``$\theta$ version" of an ultralocal algebra $\A_0$, as the one introduced in Section \ref{ssec:dPA}. We have  $\hat{\A_0}=\A_0[\theta_1,\ldots,\theta_\ell]$. Similarly we have the space $\hat{\F_0}=\hat{\A_0}/[\hat{\A_0},\hat{\A_0}]$.

First, we need to introduce a graded version of the operations we have defined on $\A^{\otimes n}$ (since the shift operation is not involved, the following are the same for both $\hat{\A_0}$ and $\hat{\A}$). We have
\begin{gather}\label{eq:superswap}
(a\otimes b)^\sigma=(-1)^{|a||b|}b\otimes a,\\
\tau\left(a_1\otimes\a_{n-1}\otimes a_n\right)=(-1)^{|a_n|(|a_1|+\cdots+|a_{n-1}|)}a_n\otimes a_1 \cdots \otimes a_{n-1},\\
a \otimes b\otimes_1 c=(-1)^{|b||c|}a\otimes c\otimes b,\\\label{eq:superinner}
(a\otimes b)\star c=(-1)^{|b||c|}ac\otimes b,\quad a\star(b\otimes c)=(-1)^{|a||b|}b\otimes ac,\\ \label{eq:superbullet}
(a \otimes b)\bullet(c\otimes d)=(-1)^{|b|(|c|+|d|)}ac\otimes db.
\end{gather}
Moreover,
$$
\frac{\dev}{\dev z}(ab)=\frac{\dev a}{\dev z}b+(-1)^{|z||a|}a\frac{\dev b}{\dev z}
$$
and, by the graded version of the trace,
\begin{equation}\label{eq:supertr}
\int\tr ab=(-1)^{|a||b|}\int\tr ba.
\end{equation}

Our aim is to define the Schouten bracket among ultralocal polyvector fields, namely in $\hat{\F_0}$. For this, we introduce a degree -1 double bracket $\hat{\A}_0\times \hat{\A}_0\to \hat{\A}_0\otimes \hat{\A}_0$ given by the formula
\begin{equation}\label{eq:ds_master_ul}
\lsb a,b\rsb :=\sum_{i=1}^\ell\left(\frac{\dev b}{\dev u^i}\bullet\left(\frac{\dev a}{\dev \theta_i}\right)^\sigma+(-1)^{|b|}\frac{\dev b}{\dev \theta_i}\bullet\left(\frac{\dev a}{\dev u^i}\right)^\sigma\right).
\end{equation}

\begin{proposition}\label{thm:ds_prop}
The bracket \eqref{eq:ds_master_ul} satisfies the following basic properties:
\begin{enumerate}[label=(\roman*)]
\item $\lsb b,a\rsb=-(-1)^{(|a|-1)(|b|-1)}\lsb a,b\rsb^\sigma$ (graded skewsymmetry);
\item $\lsb a,bc\rsb=b\lsb a,c\rsb+(-1)^{(|a|-1)|c|}\lsb a,b\rsb c$ (graded Leibniz property).
\end{enumerate}
\end{proposition}
Note that the signs for the Leibniz property imply that this bracket is a derivation \emph{from the right}, i.e. $D(ab)=aD(b)+(-1)^{|D||b|}D(a)b$.
\begin{proof}
(i). A direct computations using the definition \eqref{eq:ds_master_ul} shows that
\begin{equation}\label{eq:skew_dsb}
\begin{split}
\lsb b,a\rsb&=\sum(-1)^{b_\theta' b_\theta''}\left(\left(\frac{\dev a}{\dev u^i}\right)'\otimes\left(\frac{\dev a}{\dev u^i}\right)''\right)\bullet\left(\left(\frac{\dev b}{\dev \theta_i}\right)''\otimes\left(\frac{\dev b}{\dev \theta_i}\right)'\right)\\
&\quad+(-1)^{|a|+b_u' b_u''}\left(\left(\frac{\dev a}{\dev \theta_i}\right)'\otimes\left(\frac{\dev a}{\dev \theta_i}\right)''\right)\bullet\left(\left(\frac{\dev b}{\dev u^i}\right)''\otimes\left(\frac{\dev b}{\dev u^i}\right)'\right),
\end{split}
\end{equation}
where for shorthand we denote $\deg_\theta (\dev_ub)'=b_u'$ and so on. Then, by combining \eqref{eq:superbullet} and \eqref{eq:superswap} we get
\begin{equation}\label{eq:superswapbull}
\left[(a\otimes b)\bullet (c\otimes d)\right]^\sigma=(-1)^{(|a|+|b|)(|c|+|d|)}(c\otimes d)^\sigma\bullet(a\otimes b)^\sigma.
\end{equation}
This means that for \eqref{eq:skew_dsb} we have
\begin{equation}
\begin{split}
\lsb b,a\rsb&=\sum(-1)^{b_\theta'b_\theta''+(a_u'+a_u'')(b_\theta'+b_\theta'')+b_\theta'b_\theta''+a'_ua''_u}\left[\left(\left(\frac{\dev b}{\dev \theta_i}\right)'\otimes\left(\frac{\dev b}{\dev \theta_i}\right)''\right)\bullet \left(\left(\frac{\dev a}{\dev u^i}\right)''\otimes\left(\frac{\dev a}{\dev u^i}\right)'\right)\right]^\sigma\\
&\quad+(-1)^{|a|+b_u' b_u''+(a_\theta'+a_\theta'')(b_u'+b_u'')+b'_ub''_u+a'_\theta a''_\theta}\left[\left(\left(\frac{\dev b}{\dev u^i}\right)'\otimes\left(\frac{\dev b}{\dev u^i}\right)''\right)\bullet\left(\left(\frac{\dev a}{\dev \theta_i}\right)''\otimes\left(\frac{\dev a}{\dev \theta_i}\right)'\right)\right]^{\sigma}\\
&=\sum\left[(-1)^{(a_u'+a_u'')(b_\theta'+b_\theta'')}\frac{\dev b}{\dev \theta_i}\bullet\left(\frac{\dev a}{\dev u^i}\right)^\sigma+(-1)^{|a|+(a_\theta'+a_\theta'')(b_u'+b_u'')}\frac{\dev b}{\dev u^i}\bullet\left(\frac{\dev a}{\dev \theta_i}\right)^\sigma\right]^\sigma.
\end{split}
\end{equation}
We finally observe that $(z'_u+z''_u)=|z|$ and $(z'_\theta+z''_\theta)=|z|-1$ for $z=a,b$, so that the two exponents are respectively $|a|(|b|-1)$ and $|a|+(|a|-1)|b|$. By rearranging the two terms of the sum and collecting common factors we obtain $\lsb b,a\rsb=(-1)^{|a||b|+|a|+|b|}\lsb a,b\rsb^\sigma$ as claimed.

(ii) From a straightforward application of the formula \eqref{eq:ds_master_ul} we have
\begin{equation}
\begin{split}
\lsb a, bc\rsb&=\left(\frac{\dev b}{\dev u^i}c+b\frac{\dev c}{\dev u^i}\right)\bullet\left(\frac{\dev a}{\dev\theta_i}\right)^\sigma+(-1)^{|b|+|c|}\left(\frac{\dev b}{\dev\theta_i}c+(-1)^{|b|}b\frac{\dev c}{\dev\theta_i}\right)\bullet\left(\frac{\dev a}{\dev u^i}\right)^\sigma\\
&=b\left[\frac{\dev c}{\dev u^i}\bullet\left(\frac{\dev a}{\dev \theta_i}\right)^\sigma+(-1)^{|c|}\frac{\dev c}{\dev\theta_i}\bullet\left(\frac{\dev a}{\dev u^i}\right)^\sigma\right]\\
&\quad+\left(\left(\frac{\dev b}{\dev u^i}\right)'\otimes \left(\frac{\dev b}{\dev u^i}\right)''c\right)\bullet\left(\frac{\dev a}{\dev\theta_i}\right)^\sigma+(-1)^{|b|+|c|}\left(\left(\frac{\dev b}{\dev \theta_i}\right)'\otimes \left(\frac{\dev b}{\dev \theta_i}\right)''c\right)\bullet\left(\frac{\dev a}{\dev u^i}\right)^\sigma\\
&=b\lsb a,c\rsb +(-1)^{|c|(|a|-1)}\left[\frac{\dev b}{\dev u^i}\bullet\left(\frac{\dev a}{\dev \theta_i}\right)^\sigma\right]c+(-1)^{|b|+|c|+|a||c|}\left[\frac{\dev b}{\dev \theta_i}\bullet\left(\frac{\dev a}{\dev u^i}\right)^\sigma\right]c.
\end{split}
\end{equation}
Collecting the common factor in the second and third summands we obtain $(-1)^{|c|(|a|-1)}\lsb a,b\rsb c$ as claimed.
\end{proof}

\begin{proposition} The bracket \eqref{eq:ds_master_ul} enjoys the following ``right Leibniz property'':
\begin{equation}\label{eq:ds_leib_l}
\lsb ab,c\rsb=\lsb a,c\rsb\star b +(-1)^{|a|(|c|-1)}a\star\lsb b,c\rsb.
\end{equation}
\end{proposition}
\begin{proof}
Using skewsymmetry and the graded Leibniz property of Proposition \ref{thm:ds_prop} we have
\begin{equation}
\lsb ab,c\rsb=-(-1)^{(|a|+|b|-1)(|c|-1)}\left(a\lsb c,b\rsb\right)^\sigma-(-1)^{(|a|-1)(|c|-1)}\left(\lsb c,a\rsb b\right)^\sigma.
\end{equation}
Denoting $(cb)'=|\lsb c,b\rsb'|$, and similarly for $(cb)''$, $(ca)'$, and $(ca)''$, then
\begin{align}
\lsb ab,c\rsb&=-(-1)^{(|a|+|b|-1)(|c|-1)+(cb)'(cb)''}a\star\lsb c,b\rsb''\otimes\lsb c,b\rsb'\\
&\quad-(-1)^{(|a|-1)(|c|-1)+(ca)'(ca)''}\lsb c,a\rsb''\otimes\lsb c,a\rsb'\star b\\
&=(-1)^{|a|(|c|-1)}a\star\lsb b,c\rsb+\lsb a,c\rsb\star b.\qedhere
\end{align}
\end{proof}
The bracket \eqref{eq:ds_master_ul} satisfies graded versions of skewsymmetry property, graded Leibniz property and, as we prove for Proposition \ref{thm:dJac} in Appendix \ref{app:Jacobi}, double Jacobi identity. Since these properties are analogue to the properties that a (standard) Schouten bracket enjoys, we call $\lsb -,-\rsb$ the \emph{double Schouten bracket}. It is the fundamental building block for the Schouten bracket in the nonabelian setting.
\begin{remark}
The structure of double Schouten bracket we have presented is the same as the double Schouten--Nijenhuis bracket on the derivations of $A$, with $A=\A_0$, defined by Van Den Bergh in \cite[Section 3.2]{vdb} after the identification
\begin{equation}
\theta_i:=\frac{\dev}{\dev u^i}\colon \A_0\to \A_0\otimes\A_0.
\end{equation}
Both the brackets define a structure of double Gerstenhaber algebra and they coincide on the generators. Indeed, the bracket is trivial on generators of $\A_0$ and $\frac{\dev}{\dev u^i}(u^j)=\delta_i^j1\otimes1=\lsb \theta_i,u^j\rsb$. On the other hand, in Van Den Bergh's double bracket we have 
\begin{align}
\Bigg\{\!\Bigg\{ \frac{\dev}{\dev u^i},\frac{\dev}{\dev u^j}\Bigg\}\!\Bigg\}&=\tau_{(23)}\left(\left(\frac{\dev}{\dev u^i}\otimes 1\right)\circ \frac{\dev}{\dev u^j}-\left(1\otimes \frac{\dev}{\dev u^j}\right)\circ\frac{\dev}{\dev u^i}\right)\\
&\qquad+\tau_{(12)}\left(\left(1\otimes \frac{\dev}{\dev u^i}\right)\circ\frac{\dev}{\dev u^j}-\left(\frac{\dev}{\dev u^j}\otimes 1\right)\circ\frac{\dev}{\dev u^i}\right),
\end{align}
where $\tau_{(12)}$ and $\tau_{(23)}$ denote permutations of the factors in $\A_0^{\,\otimes3}$. Both the terms in the RHS vanish because of the commutation of double derivatives \cite[Lemma 2.6]{dskv15}, so that the double bracket between the derivations coincides with $\lsb \theta_i,\theta_j\rsb=0$.
\end{remark}
\begin{proposition}\label{def:sch}
Let $a,b\in\hF_0$. The bracked defined by
\begin{equation}\label{eq:schdef}
[a,b]:=\tr m(\lsb a,b\rsb)
\end{equation}
is a well-defined bilinear map  $\hF_0^p\times\hF_0^q\to\hF_0^{p+q-1}$
\end{proposition}
\begin{proof}
The grading of the bracket is obvious, because the derivative with respect to $\theta$ is of degree $-1$, while all the other operations in the definition are of degree 0. We need to prove that such an operation is well-defined, namely that it is vanishes on elements of $[\hcA_0,\hcA_0]$. To do so, let's recall that the graded commutator in $\hcA_0$ is $[a,b]=a b-(-1)^{|a||b|}ba$. Then
\begin{align}
[a,bc-(-1)^{|b||c|}cb]&=\tr m(\lsb a,bc\rsb)-(-1)^{|b||c|}\tr m(\lsb a,cb\rsb)\\
&=\tr m\left(b\lsb a,c\rsb+(-1)^{|a||c|+|c|}\lsb a,b\rsb c\right.\\
&\quad\quad\left.-(-1)^{|b||c|}c\lsb a,b\rsb-(-1)^{|b||c|+|a||b|+|b|}\lsb a,c\rsb b\right).
\end{align}
Observing that $\tr a b=(-1)^{|a||b|}\tr b a$ we have then
\begin{align}
[a,bc-(-1)^{|b||c|}cb]&=\tr m\left(b\lsb a,c\rsb+(-1)^{|a||c|+|c|}\lsb a,b\rsb c-(-1)^{|b||c|+|a||c|+|b||c|+|c|}\lsb a,b\rsb c\right.\\
&\quad\quad\left.-(-1)^{|b||c|+|a||b|+|b|+|a||b|+|b||c|+|b|}b\lsb a,c\rsb\right)=0.
\end{align}
Note that the vanishing of the expression is due to the overall trace.
We also need to show that the bracket vanishes for elements of $[\hA_0,\hA_0]$ in its first entry. By definition we have
\begin{align}
[ab-(-1)^{|a||b|}ba,c]&=\tr m\left(\lsb ab-(-1)^{|a||b|}ba,c\rsb\right)\\
&=-(-1)^{(|a|+|b|-1)(|c|-1)}\tr m\left(\lsb c,ab-(-1)^{|a||b|}ba\rsb^\sigma\right),
\end{align}
because of the skewsymmetry of the double Schouten bracket. We can now go back to the computation for the commutator in the second entry because $\tr m(A^\sigma)=\tr m(A)$, which is straightforward given \eqref{eq:superswap} and \eqref{eq:supertr}. In conclusion, the bracket vanishes for elements of $[\hA_0,\hA_0]$ in either entry, which means that it is well-defined on $\hF_0$.
\end{proof}
\begin{definition}\label{def:sch-gen}
Let $[-,-]$ be a degree -1 bilinear bracket among polyvector fields. We call such a bracket a \emph{Schouten bracket} if it coincides with the commutator of vector fields among 1-vectors, with the action of a vector field on a functional when computed between a 0- and a 1-vector, and satisfies the following graded versions of skewsymmetry and Jacobi identity:
\begin{enumerate}[label=(\roman*)]
\item $[P,Q]=-(-1)^{(p-1)(q-1)}[Q,P]$
\item $[P,[Q,R]]=[[P,Q],R]+(-1)^{(p-1)(q-1)}[Q,[P,R]]$
\end{enumerate}
for a $p$-vector $P$, a $q$-vector $Q$ and a $r$-vector $R$.
\end{definition}
We are now going to prove that bracket \eqref{eq:schdef} fulfils the properties outlined in Definition \ref{def:sch-gen}, starting from the fact that it coincides with the action of a vector field on a local functionals when evaluated between elements of $\F_0$ and $\hF^1_0$. Indeed, if we compute $[X,f]$ using \eqref{eq:ds_master_ul} for $X=\tr X^i\theta_i$ and $f\in \F_0$, we obtain exactly the same expression as in \eqref{eq:evvfield-action}.

\begin{proposition}\label{thm:Sch_vf}
For any pair of vector fields $X,Y$, the bracket $[X,Y]$ defined as in \eqref{eq:schdef} coincides with the commutator of vector fields, namely produces a vector field whose action on the functional $f$ is $X(Y(f))-Y(X(f))$.
\end{proposition}
\begin{proof}
Note that in the ultralocal setting a vector field $\{X^i\}$ acts on $f$ simply by
$$
\tr m\left( X^i\star\frac{\dev f}{\dev u^i}\right)=\tr\left(\left(\frac{\dev f}{\dev u^i}\right)'X^i\left(\frac{\dev f}{\dev u^i}\right)''\right),
$$
where the sum over $i$ is left implicit. From Lemma \ref{lem:comm_vf}, the commutator of vector fields has characteristics
\begin{equation}\label{eq:sch_vf_pf1}
[X,Y]^i=\left(\frac{\dev Y^i}{\dev u^j}\right)'X^j\left(\frac{\dev Y^i}{\dev u^j}\right)''-\left(\frac{\dev X^i}{\dev u^j}\right)'Y^j\left(\frac{\dev X^i}{\dev u^j}\right)''.
\end{equation}
The elements of $\hF_0$ corresponding to those vector fields are $\tr X^i\theta_i$ and $\tr Y^j\theta_j$; using formula \eqref{eq:ds_master_ul} we have
\begin{equation}
\begin{split}
[X^i\theta_i,Y^j \theta_j]&=\tr m\left[\left(\frac{\dev Y^j}{\dev u^i}\theta_j\right)\bullet (1\otimes X^i)-\left(Y^i\otimes 1\right)\bullet\left(\frac{\dev X^j}{\dev u^i}\theta_j\right)^\sigma\right]\\
&=\tr\left[\left(\frac{\dev Y^j}{\dev u^i}\right)'X^i\left(\frac{\dev Y^j}{\dev u^i}\right)''\theta_j-Y^i\left(\frac{\dev X^j}{\dev u^i}\right)''\theta_j\left(\frac{\dev X^j}{\dev u^i}\right)'\right].
\end{split}
\end{equation}
Using the trace operation to bring $\theta_j$ in the second term to the rightmost position we easily read the expression
\begin{equation}
[X^i\theta_i,Y^j\theta_j]=\tr [X,Y]^j\theta_j
\end{equation}
with $[X,Y]^j$ as in \eqref{eq:sch_vf_pf1}.
\end{proof}

\begin{proposition}\label{thm:Sch_skew}
The bracket \eqref{eq:schdef} is graded skewsymmetric as a Schouten bracket, namely 
\begin{equation}
[b,a]=-(-1)^{(|a|-1)(|b|-1)}[a,b]\label{eq:schskew}
\end{equation}
and fulfils the graded version of the Jacobi identity 
\begin{equation}
[a,[b,c]]=[[a,b],c]+(-1)^{(|a|-1)(|b|-1)}[b,[a,c]].\label{eq:schJac}
\end{equation}
\end{proposition}
\begin{proof}
Let us first prove \eqref{eq:schskew}. From the definition and the graded supersymmetry of $\lsb a,b\rsb$ we have
$$
[b,a]=-(-1)^{(|a|-1)(|b|-1)}\tr m(\lsb a,b\rsb^\sigma).
$$
Let us assume that $\lsb a,b\rsb= A\otimes B$; then
\begin{equation}
\begin{split}
[b,a]&=-(-1)^{(|a|-1)(|b|-1)+|A||B|}\tr BA=-(-1)^{(|a|-1)(|b|-1)}\tr AB\\
&=-(-1)^{(|a|-1)(|b|-1)}[a,b].
\end{split}
\end{equation}
The proof for \eqref{eq:schJac} is long and left to Appendix \ref{app:Jacobi}.
\end{proof}

From Proposition \ref{thm:Sch_vf} and \ref{thm:Sch_skew} we conclude that the bracket \eqref{eq:schdef} is the \emph{Schouten bracket} for nonabelian ultralocal polyvector fields. 

\subsection{Schouten (lambda) brackets for the local case}\label{ssec:sch-l}
Having defined the Schouten bracket for ultralocal polyvector fields, we now address the more general case of difference ones. The underlying space of densities is the space of Laurent difference polynomials $\A$. Similarly to what we did in the previous section, we first introduce a double $\lambda$ bracket on $\hA$ and use it to define the Schouten bracket on the space $\hF$.

Let us define the degree -1 double $\lambda$ bracket $\hA\times\hA\to\hA\otimes\hA\lsb\lambda\rsb$ given by the formula
\begin{equation}\label{eq:ds_master}
\lsb a_{\lambda}b\rsb:=\sum_{i=1}^\ell\sum_{m,n}\left(\frac{\dev b}{\dev u^i_m}\bullet(\lambda\cS)^{m-n}\left(\frac{\dev a}{\dev \theta_{i,n}}\right)^\sigma+(-1)^{|b|}\frac{\dev b}{\dev \theta_{i,m}}\bullet(\lambda\cS)^{m-n}\left(\frac{\dev a}{\dev u^i_n}\right)^\sigma\right).
\end{equation}
Apart from the introduction of a grading, \eqref{eq:ds_master} is a multiplicative $\lambda$ bracket: in particular it satisfies the sesquilinearity property (see Definition \ref{def:dlbracket}).
\begin{proposition}
For the bracket defined in \eqref{eq:ds_master}, we have the following:
\begin{align}\label{eq:lsesq-dsb}
\lsb \cS a_{\lambda} b\rsb&=\lambda^{-1}\lsb a_{\lambda}b\rsb,\\ \label{eq:rsesq-dsb}
\lsb a_{\lambda} \cS b\rsb&=\lambda\cS\lsb a_{\lambda}b\rsb.
\end{align}
\end{proposition}
\begin{proof}
We recall that, for both the derivatives with respect to $u$'s and $\theta$'s, we have
$$
\frac{\dev \cS f}{\dev u^i_n}=\cS\frac{\dev f}{\dev u^i_{n-1}},\qquad\frac{\dev \cS f}{\dev \theta_{i,n}}=\cS\frac{\dev f}{\dev \theta_{i,n-1}}.
$$
Let us just consider the first of the two summands in \eqref{eq:ds_master}, since the behaviour is the same in both. For the expansion of \eqref{eq:lsesq-dsb} we have
\begin{equation}
\frac{\dev b}{\dev u^i_m}\bullet(\lambda\cS)^{m-n}\left(\frac{\dev \cS a}{\dev \theta_{i,n}}\right)^\sigma=\lambda^{-1}\frac{\dev b}{\dev u^i_m}\bullet(\lambda\cS)^{m-(n-1)}\left(\frac{\dev a}{\dev \theta_{i,n-1}}\right)^\sigma,
\end{equation}
and similarly for the second term. On the other hand, computing the bracket for \eqref{eq:rsesq-dsb} we obtain
\begin{equation}
\begin{split}
\frac{\dev\cS b}{\dev u^i_m}\bullet(\lambda\cS)^{m-n}\left(\frac{\dev a}{\dev\theta_{i,n}}\right)^\sigma&=\left(\cS\frac{\dev b}{\dev u^i_{m-1}}\right)\bullet(\lambda\cS)^{m-n}\left(\frac{\dev a}{\dev\theta_{i,n}}\right)^\sigma\\
&=\lambda\cS\left(\frac{\dev b}{\dev u^i_{m-1}}\bullet(\lambda\cS)^{m-1-n}\left(\frac{\dev a}{\dev\theta_{i,n}}\right)^\sigma\right).
\end{split}
\end{equation}
The same happens for the second term in \eqref{eq:ds_master}, giving us the full result \eqref{eq:rsesq-dsb}.
\end{proof}
\begin{proposition}\label{thm:property-ds}The bracket \eqref{eq:master} enjoys the following version of the properties for the $\lambda$ brackets:
\begin{enumerate}[label=(\roman*)]
\item Skewsymmetry
\begin{equation}\label{eq:skew-sch-diff}
\lsb b_{\lambda}a\rsb=-(-1)^{(|a|-1)(|b|-1)}{}_\to\lsb a_{(\lambda\cS)^{-1}}b\rsb^\sigma.\end{equation}
\item Left Leibniz property
\begin{equation} \label{eq:lLeib-ds}
\lsb a_\lambda bc\rsb=b\lsb a_\lambda c\rsb+(-1)^{(|a|-1)|c|}\lsb a_{\lambda}b\rsb c.\end{equation}
\item Right Leibniz property
\begin{equation}\label{eq:rLeib-ds}
\lsb ab_{\lambda}c\rsb=\lsb a_{\lambda\cS}c\rsb\star\left({}_\to b\right)+(-1)^{|a|(|c|-1)}\left({}_\to a\right)\star\lsb b_{\lambda\cS}c\rsb.\end{equation}
\item Jacobi identity
\begin{equation}
\label{eq:jac-scj-diff}
\lsb a_\lambda\lsb b_\mu c\rsb\rsb_L-(-1)^{(|a|-1)(|b|-1)}\lsb b_{\mu}\lsb a_{\lambda}c\rsb\rsb_R-\lsb\lsb a_{\lambda}b\rsb_{\lambda\mu}c\rsb_L=0.\end{equation}
\end{enumerate}

The notation $({}_\to b)$ (resp., $({}_\to a)$) used in \eqref{eq:rLeib-ds} means that the shift operators appearing in the double Schouten brackets act also on $b$ (resp. $a$). 
\end{proposition}
\begin{proof}
For the sake of compactness, we omit the symbol for the summation over all the indices $i=1,\ldots,\ell$ and $m,n\in \Z$. To prove \eqref{eq:skew-sch-diff}, we compute using \eqref{eq:ds_master}

\begin{equation}
\begin{split}
\lsb b_{\lambda}a\rsb&=-(-1)^{(|a|-1)(|b|-1)}\left(\left(\cS^{n-m}\frac{\dev b}{\dev u^i_m}\right)\bullet \left(\frac{\dev a}{\dev \theta_{i,n}}\right)^\sigma\right.\\
&\qquad\qquad\qquad\qquad\left.+(-1)^{|b|}\left(\cS^{n-m}\frac{\dev b}{\dev \theta_{i,m}}\right)\bullet \left(\frac{\dev a}{\dev u^i_n}\right)^\sigma\right)^\sigma\lambda^{n-m}\\
&=-(-1)^{(|a|-1)(|b|-1)}\lambda^{n-m}\cS^{n-m}\left(\frac{\dev b}{\dev u^i_m}\bullet \cS^{m-n}\left(\frac{\dev a}{\dev \theta_{i,n}}\right)^\sigma\right.\\
&\qquad\qquad\qquad\qquad\left.+(-1)^{|b|}\frac{\dev b}{\dev \theta_{i,m}}\bullet \cS^{m-n}\left(\frac{\dev a}{\dev u^i_n}\right)^\sigma\right)^\sigma.
\end{split}
\end{equation}
If we write, symbolically, $\lsb a_{\lambda}b\rsb=B'\otimes B''\lambda^{m-n}$, then in the last passage we can read $\lambda^{n-m}\cS^{n-m}\left(B'\otimes B''\right)^\sigma$, namely ${}_\to\lsb a_{(\lambda\cS)^{-1}}b\rsb$. We have therefore obtained \eqref{eq:skew-sch-diff} as claimed.

For \eqref{eq:lLeib-ds}, we have explicitly from \eqref{eq:ds_master} that
\begin{equation}
\begin{split}\lsb a_\lambda bc\rsb&=\left(\frac{\dev b}{\dev u^i_m}c+b\frac{\dev c}{\dev u^i_m}\right)\bullet\left(\lambda\cS\right)^{m-n}\left(\frac{\dev a}{\dev\theta_{i,n}}\right)^\sigma\\
&\qquad+(-1)^{|b|+|c|}\left(\frac{\dev b}{\dev\theta_{i,m}}c+(-1)^{|b|}b\frac{\dev c}{\dev\theta_{i,m}}\right)\bullet\left(\lambda\cS\right)^{m-n}\left(\frac{\dev a}{\dev u^i_n}\right)^\sigma\\
&=b\lsb a_{\lambda} c\rsb+(-1)^{|c|(|a|-1)}\left[\frac{\dev b}{\dev u^i_m}\bullet (\lambda\cS)^{m-n}\left(\frac{\dev a}{\dev\theta_{i,n}}\right)^\sigma\right]c\\&\qquad+(-1)^{|b|+|c|+|c||a|}\left[\frac{\dev b}{\dev\theta_{i,m}}\bullet(\lambda\cS)^{m-n}\left(\frac{\dev a}{\dev u^i_n}\right)^\sigma\right]c.
\end{split}
\end{equation}
Note that the set of sign rules \eqref{eq:superswap}, \eqref{eq:superinner}, \eqref{eq:superbullet} can be summarized as if the formulae followed the standard rule for a $\theta$-graded commutative product. This means that we produce a $(-1)^{|a||b|}$ factor whenever we exchange the position of two elements $a$ and $b$ following the definition of bullet, outer or inner module product, or we swap the factors in $\hA\otimes\hA$. Observe that the sign factor in the third term of the previous equality can be written as $(-1)^{|c|(|a|+1)+|b|}=(-1)^{|c|(|a|-1)+|b|}$, that we can collect the overall $(-1)^{|c|(|a|-1)}$ factor in the second and third terms and express them as $(-1)^{|c|(|a|-1)}\lsb a_{\lambda}b\rsb c$, namely the second term in the RHS of \eqref{eq:lLeib-ds}. 

Property \eqref{eq:rLeib-ds} follows from \eqref{eq:lLeib-ds} together with the skewsymmetry \eqref{eq:skew-sch-diff}.

The proof for \eqref{eq:jac-scj-diff} follows the same lines as the one for Proposition \ref{thm:dJac} in Appendix \ref{app:Jacobi}. The formal parameters $\lambda$ and $\mu$ play little role in the overall computations: note that in all the instances where the left Leibniz property for the double Schouten bracket (Part (ii) of Proposition \ref{thm:ds_prop}) is used, we replace it with the same property for the double Schouten $\lambda$ bracket \eqref{eq:lLeib-ds}.
\end{proof}

Similarly to what we observed in the previous section, the bracket \eqref{eq:ds_master} is a double $\lambda$ bracket satisfying a skewsymmetry and a double Jacobi identity with suitable and consistent grading. We call this bracket a \emph{double Schouten $\lambda$ bracket} and use it to define the Schouten bracket on the space $\hF$. As before, we give a well-defined bracket on it and we prove that it satisfies the properties of Definition \ref{def:sch-gen}.

\begin{proposition}\label{thm:schbr-l}
Given two elements $a,b\in\hF$, the bracket
\begin{equation}\label{eq:schdef-diff}
[a,b]:=\int\tr m\left(\lsb a_{\lambda}b\rsb\right)\big|_{\lambda=1}
\end{equation}
is a bilinear map $\hF^p\times\hF^q\to\hF^{p+q-1}$ that satisfies the graded skewsymmetry \eqref{eq:schskew} and the graded Jacobi identity \eqref{eq:schJac}.
\end{proposition}
\begin{proof}
We prove that \eqref{eq:schdef-diff} is well-defined. To check that the result does not change if we perform cyclic permutations of $a$ and $b$, we can follow the proof of Proposition \ref{def:sch}, using the Leibniz properties \eqref{eq:lLeib-ds} and \eqref{eq:rLeib-ds}. In addition, we have to check that the same happens if we replace $a$ (resp. $b$) with $\cS a$ (resp. $\cS b$), since in $\hF$ the two are identified. 
From \eqref{eq:lsesq-dsb} we have
$$
\lsb \cS a_{\lambda}b\rsb\big|_{\lambda=1}=\lsb a_{\lambda}b\rsb\big|_{\lambda=1},
$$
so that the value of the bracket in $\hF$ does not change. Similarly, from \eqref{eq:rsesq-dsb} we have
$$
\lsb a_{\lambda}\cS b\rsb\big|_{\lambda=1}=\cS \lsb a_{\lambda} b\rsb\big|_{\lambda=1},
$$
which yields the same result as $\lsb a_{\lambda}b\rsb|_{\lambda=1}$ after the integration.

The skewsymmetry and the Jacobi identity for \eqref{eq:schdef-diff} follow from Proposition \ref{thm:property-ds} as in the analogue ultralocal case (compare with Proposition \ref{thm:Sch_skew} and Appendix \ref{app:Jacobi}).
\end{proof}

From Proposition \ref{thm:schbr-l} we can conclude that the bracket defined in \eqref{eq:schdef-diff} is the \emph{Schouten bracket} for nonabelian difference polyvector fields. Moreover, observe that the double Schouten bracket \eqref{eq:ds_master_ul} (and its corresponding Schouten bracket \eqref{eq:schdef}) can be regarded as a special case of the double Schouten $\lambda$ bracket \eqref{eq:ds_master}. This is the reason why we use the same notation for \eqref{eq:schdef} and \eqref{eq:schdef-diff}.
\subsection{Poisson bracket revisited}	\label{ssec:pois}
The identification of $\hF$ with the complex of polyvector fields and the introduction of a Schouten bracket \eqref{eq:schdef-diff} on it allows us to replace the definition of Poisson operator \eqref{eq:HamProp} with the standard language of Poisson geometry.

Let us consider a bivector $B$ of the form 
\begin{equation}\label{eq:Pbiv_def}
B=\sum_{p,\alpha_p}\int\tr \theta_i H^{(\alpha_p)ij}_{L}\theta_{j,p}H^{(\alpha_p)ij}_{R},
\end{equation}
defined, according to \eqref{eq:defBiv}, by the skewsymmetric operator
$$
H=\sum\left(\l_{H^{(\alpha_p)ij}_{L}}\r_{H^{(\alpha_p)ij}_{R}}\cS^{p}-\cS^{-p}\r_{H^{(\alpha_p)ji}_{L}}\l_{H^{(\alpha_p)ji}_{R}}\right).
$$

As in classical Poisson geometry, we say that $H$ is a Hamiltonian operator if it has two main properties (which are shared by double Poisson brackets and double Poisson $\lambda$ bracket, see \cite{dskv15}):
\begin{enumerate}
\item It defines a Lie algebra on $\F$, namely it defines a skewsymmetric bracket which fulfils the Jacobi identity (a Poisson bracket)
\item Defines an action of $\F$ on $\A$ by derivations (Hamiltonian vector fields).
\end{enumerate}
We can define and interpret both these structures using only the Schouten bracket and a bivector satisfying some constraints, starting from the action of $\F$ on $\A$.

\begin{proposition}
Let $B$ be a bivector and $F=\int\tr f$ a local functional. Then the evolutionary vector field associated to $F$ and produced by the operator $H$ is
\begin{equation}\label{eq:hamvf_sch}
X_F=-[B,F].
\end{equation}
\end{proposition}
\begin{proof}
From the definition of Schouten bracket \eqref{eq:schdef-diff} and the master formula \eqref{eq:ds_master} we have
\begin{equation}\label{eq:hamvf-sch-pf1}
[B,\smallint\tr f]=\int\tr\left(\frac{\dev  f}{\dev u^i_m}\right)'\left(\cS^{m-n}\left(\frac{\dev B}{\dev \theta_{i,n}}\right)''\left(\frac{\dev B}{\dev \theta_{i,n}}\right)'\right)\left(\frac{\dev f}{\dev u^i_m}\right)''.
\end{equation}
The derivative of $B$ with respect to $\theta$ is 
\begin{equation}\label{eq:PoisBiv-pf4}
\left(\frac{\dev B}{\dev \theta_{l,n}}\right)^\sigma=\delta_{n,0} H^{(\alpha_p)lj}_{L}\theta_{j,p}H^{(\alpha_p)lj}_{R}\otimes 1 - H^{(\alpha_n)jl}_{R}\otimes \theta_j H^{(\alpha_n)jl}_{L}.
\end{equation}
After some elementary manipulations we can rewrite \eqref{eq:hamvf-sch-pf1} as
\begin{equation}\label{eq:hamvf-sch-pf2}
[B,\smallint\tr f]=\int\tr\left(H^{(\alpha_p)li}_{L}\theta_{i,p}H^{(\alpha_p)li}_{R}-\left(\cS^{-p}H^{(\alpha_p)il}_{R}\right)\theta_{i,-p}\left(\cS^{-p}H^{(\alpha_p)il}_{L}\right)\right)\frac{\delta f}{\delta u^l}
\end{equation}
Finally, normalising \eqref{eq:hamvf-sch-pf2} collecting $\theta_i$ we obtain the evolutionary vector field
\begin{equation}
[B,\smallint\tr f]=\int\tr\left[\left(\cS^{-p}H^{(\alpha_p)li}_{R}\right)\left(\cS^{-p}\frac{\delta f}{\delta u^l}\right)\left(\cS^{-p}H^{(\alpha_p)li}_{L}\right)-H^{(\alpha_p)il}_{L}\left(\cS^p\frac{\delta f}{\delta u^l}\right)H^{(\alpha_p)il}_{R}\right]\theta_i,
\end{equation}
which is exactly the evolutionary vector field of characteristics
\begin{equation}
X^i=-\sum_lH^{il}\left(\frac{\delta f}{\delta u^l}\right).
\end{equation}
If $H$ is a Hamiltonian operator, this is, up to the sign, the expression for the characteristic of a Hamiltonian vector field \eqref{eq:HamSyst_expl}. Because of this, we put the minus sign in \eqref{eq:hamvf_sch}.
\end{proof}

Given an evolutionary vector field as defined in \eqref{eq:hamvf_sch}, we use the Schouten bracket to define a bracket in $\F$, too. The operational definition of Poisson bracket \eqref{eq:defBra} we have used throughout the paper can be read, indeed, as
\begin{equation}\label{eq:bra-cl}
\{F,G\}=X_G(F),
\end{equation}
or
\begin{equation}\label{eq:defBra-sch0}
\{F,G\}=-[[B,G],F].
\end{equation}
We can call \eqref{eq:bra-cl} a Poisson bracket only if it satisfies skewsymmetry and Jacobi identity (or, equivalently, if $H$ is a Hamiltonian operator); however, any bivector $B$ can be used to define a bracket in $\F$ according to \eqref{eq:defBra-sch0}. 
\begin{proposition}\label{thm:skew}
The bracket \eqref{eq:defBra-sch0} is skewsymmetric.
\end{proposition}
\begin{proof}
It simply follows from the Jacobi identity for the Schouten bracket. Indeed, we have
\begin{equation}
\{F,G\}=-[[B,G],F]=-[G,[B,F]]-[B,[G,F]]
\end{equation}
Then, using the fact that the Schouten bracket between local functionals vanishes and the skewsymmetry \eqref{eq:schskew}, we can conclude that
\[
\{F,G\}=-[G,[B,F]]=[[B,F],G]=-\{G,F\}.\qedhere
\]
\end{proof}
We can exploit the skewsymmetry of the Poisson bracket to define, equivalently to \eqref{eq:defBra-sch0},
\begin{equation}\label{eq:defBra-sch}
\{F,G\}=[[B,F],G].
\end{equation}
\begin{lemma}\label{lem:Jac}
The Jacobiator of the bracket \eqref{eq:defBra-sch} is
\begin{equation}\label{eq:lemjac}
\{F,\{G,H\}\}+\{G,\{H,F\}\}+\{H,\{F,G\}\}=-\frac12\left[\left[\left[\left[B,B\right],F\right],G\right],H\right].
\end{equation}
Then, the Jacobi identity is equivalent to
\begin{equation}\label{eq:jac-gen}
[[[[B,B],F],G],H]=0,\qquad\forall F,G,H\in\F.
\end{equation}
\end{lemma}
\begin{proof}
Let us first consider the two innermost brackets on the RHS of \eqref{eq:jac-gen}. By the graded Jacobi identity for the Schouten bracket we have
\begin{equation}\label{eq:lemJac-pf1}
\left[\left[B,B\right],F\right]=2\left[B,\left[B,F\right]\right]\stackrel{\mathrm{\eqref{eq:hamvf_sch}}}{=}-2\left[B,X_F\right].
\end{equation}
Then
\begin{equation}\label{eq:lemJac-pf2}
\begin{split}
\left[\left[\left[B,B\right],F\right],G\right]&=-2\left[\left[B,X_F\right],G\right]=-2\left[B,\left[X_F,G\right]\right]+2\left[X_F,\left[B,G\right]\right]\\
&\stackrel{\mathrm{\eqref{eq:bra-cl}}}{=}-2\left[B,\left\{G,F\right\}\right]-2\left[X_F,X_G\right]=2X_{\{G,F\}}-2\left[X_F,X_G\right].
\end{split}
\end{equation}
Moving to the outermost bracket, we obtain
\begin{equation}\label{eq:lemJac-pf3}
\begin{split}
\left[\left[\left[\left[B,B\right],F\right],G\right],H\right]&=2X_{\{G,F\}}(H)-2X_F\left(X_G(H)\right)+2X_G\left(X_F(H)\right)\\
&=2\left(\left\{H,\left\{G,F\right\}\right\}-\left\{\left\{H,G\right\},F\right\}+\left\{\left\{H,F\right\},G\right\}\right).
\end{split}
\end{equation}
We then obtain the LHS of \eqref{eq:lemjac} using skewsymmetry. The vanishing of the RHS is the Jacobi identity for the bracket $\{-,-\}$, hence it is equivalent to \eqref{eq:jac-gen} as claimed. 
\end{proof}
We discussed the condition that $B$ must satisfy in order to define a Poisson bracket, which endows $\F$ with a Lie algebra structure. On the other hand, $B$ allows us to define evolutionary vector fields associated to local functionals by \eqref{eq:hamvf_sch}. If we want it to define an \emph{action} of $\F$ on $\A$ we must also ascertain that there is a Lie algebra morphism between the Lie algebra of local functionals $(\F,\{-,-\})$ and that of (evolutionary) vector fields $(\A,[-,-])$ (note that the we have already proved that the commutator of vector fields in $\A$ is equivalent to the Schouten bracket of 1-vectors in $\hF$ in Proposition \ref{thm:Sch_vf}). The condition is
\begin{equation}\label{eq:act-gen0}
X_{\{F,G\}}=-[X_F,X_G],
\end{equation}
which is equivalent (see \eqref{eq:lemJac-pf2}) to
\begin{equation}\label{eq:act-gen}
[[[B,B],F],G]=0.
\end{equation}
We saw that a generic bivector $B$ defines a skewsymmetric bracket, and how the properties we require from a Poisson bracket are expressed in terms of the Schouten bracket.
\subsubsection{Poisson bivectors} We introduced Poisson bivectors in Definition \ref{def:Pois}, but now we give a characterisation in terms of Schouten brackets, as in classical Poisson geometry.
\begin{theorem}\label{thm:PoisBiv}
Let $P$ be a bivector defined by the operator $H$. Then
\begin{equation}
2[P,P]=\mathbf{pr}_{H\Theta}P.
\end{equation}
Thus, $P$ is Poisson if and only if $[P,P]=0$.
\end{theorem}
The quantity $[P,P]$ is, in general, a 3-vector which is called \emph{Schouten torsion of $P$}. Before proving our claim, let us start with a preliminary Lemma.
\begin{lemma}\label{lem:PoisBiv-pf}
Let $P$ be a bivector. Then
\begin{equation}\label{eq:PoisBiv-pf}
\int\tr m\left[\frac{\dev P}{\dev u^i_m}\bullet\cS^{m-n}\left(\frac{\dev P}{\dev \theta_{i,n}}\right)^\sigma\right]=\int\tr m\left[\frac{\dev P}{\dev \theta_{i,m}}\bullet\cS^{m-n}\left(\frac{\dev P}{\dev u^i_n}\right)^\sigma\right]
\end{equation}
\end{lemma}
\begin{proof}
The computations are essentially the same performed in the proof of \eqref{eq:skew-sch-diff} for the double Schouten bracket. Note that $|P|=2$, so that $|(\dev_u P)'|\equiv|(\dev_u P)''|$, $|(\dev_\theta P)'|+|(\dev_\theta P)''|=1$, and $|(\dev_\theta P)'||(\dev_\theta P)''|=0$. Then, from the definition of the graded version of the bullet product and the swap operation we have that the LHS of \eqref{eq:PoisBiv-pf} is
\begin{equation}
(-1)^{|(\dev_u P)''|}\int\tr \left(\frac{\dev P}{\dev u^i_m}\right)'\cS^{m-n}\left(\left(\frac{\dev P}{\dev \theta_{i,n}}\right)''\left(\frac{\dev P}{\dev \theta_{i,n}}\right)'\right)\left(\frac{\dev P}{\dev u^i_m}\right)''.
\end{equation}
Taking a graded cyclic permutations of the integrand and keeping into account the possible grading of each factor we obtain
\begin{gather}
(-1)^{|(\dev_u P)'||(\dev_u P)''|}\int\tr\left(\cS^{m-n}\left(\frac{\dev P}{\dev \theta_{i,n}}\right)'\right)\left(\frac{\dev P}{\dev u^i_m}\right)''\left(\frac{\dev P}{\dev u^i_m}\right)'\left(\cS^{m-n}\left(\frac{\dev P}{\dev \theta_{i,n}}\right)''\right)\\
=(-1)^{|(\dev_u P)'||(\dev_u P)''|}\int\tr \left(\frac{\dev P}{\dev \theta_{i,n}}\right)'\left(\cS^{n-m}\left(\frac{\dev P}{\dev u^i_m}\right)''\left(\frac{\dev P}{\dev u^i_m}\right)'\right)\left(\frac{\dev P}{\dev \theta_{i,n}}\right)''\\
=\int\tr m \frac{\dev P}{\dev \theta_{i,m}}\bullet\cS^{m-n}\left(\frac{\dev P}{\dev u^i_n}\right)^\sigma.\tag*{\mbox{\qedhere}}
\end{gather}
\end{proof}
\begin{proof}[Proof of Theorem \ref{thm:PoisBiv}]
Combining the master formula for the double Schouten bracket \eqref{eq:ds_master} with \eqref{eq:schdef-diff} we compute
\begin{equation}\label{eq:PoisBiv-pf1}
[P,P]=\int\tr m\left[\frac{\dev P}{\dev u^i_m}\bullet\cS^{m-n}\left(\frac{\dev P}{\dev \theta_{i,n}}\right)^\sigma+\frac{\dev P}{\dev \theta_{i,m}}\bullet\cS^{m-n}\left(\frac{\dev P}{\dev u^i_n}\right)^\sigma\right].
\end{equation}
Lemma \ref{lem:PoisBiv-pf} tells us that it is sufficient to compute just one of the two terms of \eqref{eq:PoisBiv-pf1}, namely
\begin{equation}\label{eq:PoisBiv-pf2}
[P,P]=2\int\tr m\left[\frac{\dev P}{\dev u^i_m}\bullet\cS^{m-n}\left(\frac{\dev P}{\dev \theta_{i,n}}\right)^\sigma\right].
\end{equation}
According to the definition of $P$ \eqref{eq:Pbiv_def}, we have
\begin{equation}
\frac{\dev P}{\dev u^l_m}=\theta_i \left(\dev_{u^l_m}H^{(\alpha_p)ij}_{L}\right)'\otimes\left(\dev_{u^l_m}H^{(\alpha_p)ij}_{L}\right)''\theta_{j,p} H^{(\alpha_p)ij}_{R}+\theta_i H^{(\alpha_p)ij}_{L}\theta_{i,p}\left(\dev_{u^l_m}H^{(\alpha_p)ij}_{R}\right)'\otimes\left(\dev_{u^l_m}H^{(\alpha_p)ij}_{R}\right)'',
\end{equation}
while we have already a formula for the derivative with respect to $\theta$ in \eqref{eq:PoisBiv-pf4}. The computation for \eqref{eq:PoisBiv-pf2}, then, produces
\begin{equation}\label{eq:PoisBiv-pf3}
\begin{split}
\frac12[P,P]&=-\int\tr\left[\theta_i\left(\dev_{u^l_m}H^{(\alpha_p)ij}_{L}\right)'\left(\cS^{m}\left(H^{(\beta_q)lk}_{L}\theta_{k,q}H^{(\beta_q)lk}_{R} \right.\right.\right.\\
&\qquad\quad\left.\left.\left.- \left(\cS^{-q}H^{(\beta_q)kl}_{R}\right)\theta_{k,-q}\left(\cS^{-q} H^{(\beta_q)kl}_{L}\right)\right)\right)\left(\dev_{u^l_m}H^{(\alpha_p)ij}_{L}\right)''\theta_{j,p}H^{(\alpha_p)ij}_{R}\right]\\
&+\quad\int\tr\left[\theta_iH^{(\alpha_p)ij}_{L}\theta_{j,p}\left(\dev_{u^l_m}H^{(\alpha_p)ij}_{R}\right)'\left(\cS^{m}\left(H^{(\beta_q)lk}_{L}\theta_{k,q}H^{(\beta_q)lk}_{R} \right.\right.\right.\\
&\qquad\qquad\left.\left.\left.- \left(\cS^{-q}H^{(\beta_q)kl}_{R}\right)\theta_{k,-q}\left(\cS^{-q} H^{(\beta_q)kl}_{L}\right)\right)\right)\left(\dev_{u^l_m}H^{(\alpha_p)ij}_{R}\right)''\right].
\end{split}
\end{equation}
Observe that the term which is ``sandiwched'' between the two factors of the derivative of $H^{ij}$ is of the form $\cS^m (H\Theta)^l$, for
\begin{equation}
(H\Theta)^l=H^{(\beta_q)lk}_{L}\theta_{k,q}H^{(\beta_q)lk}_{R}-\left(\cS^{-q} H^{(\beta_q)kl}_{R}\right)\theta_{k,-q}\left(\cS^{-p} H^{(\beta_q)kl}_{L}\right).
\end{equation}
On the other hand, by definition we have that \eqref{eq:HamProp} reads
\begin{equation}
\begin{split}
\mathbf{pr}_{H\Theta} P&=-\int\tr\theta_i \left(\dev_{u^l_m} H^{(\alpha_p)ij}_{L}\right)'\left(\cS^m (H\Theta)^l\right)\left(\dev_{u^l_m} H^{(\alpha_p)ij}_{L}\right)''\theta_{j,p}H^{(\alpha_p)ij}_{R}\\
&\quad+\int\tr \theta_i H^{(\alpha_p)ij}_{L}\theta_{j,p}\left(\dev_{u^l_m} H^{(\alpha_p)ij}_{R}\right)'\left(\cS^m (H\Theta)^l\right)\left(\dev_{u^l_m} H^{(\alpha_p)ij}_{R}\right)''=0.
\end{split}
\end{equation}
Then the vanishing of the Poisson property is equivalent to $[P,P]=0$ (or we can regard $2[P,P]$ as an alternative way to write $\mathbf{pr}_{H\Theta}P$). 
\end{proof}

\begin{proposition}\label{thm:pois-suf}
Let $P$ be a Poisson bivector. Then the bracket defined on $\F$ as in \eqref{eq:defBra-sch} is a Poisson bracket, namely
\begin{gather}
\{G,F\}=-\{F,G\},\\
\{F,\{G,H\}\}=\{\{F,G\},H\}+\{G,\{F,H\}\}.
\end{gather}
\end{proposition}
\begin{proof}
Since $P$ is a bivector, the bracket is skewsymmetric (see Proposition \ref{thm:skew}). The Jacobi identity is equivalent to \eqref{eq:jac-gen} which follows from $[P,P]=0$.
\end{proof}
This means, in particular, that a Poisson operator (namely, an operator defining a Poisson bivector) is always Hamiltonian (namely, it defines a Poisson bracket). Finally, by the same property $[P,P]=0$ we have \eqref{eq:act-gen0}. 

\subsubsection{The Poisson cohomology}
It is a well-known fact in Poisson geometry that the adjoint action of the Poisson bivector defines a cochain complex on the space of polyvector field, the so-called Poisson-Lichnerowicz complex. The cohomology of this complex plays a crucial role both in the study of the Poisson manifolds themselves and for the theory of the integrable systems. Indeed, it characterizes the Casimir functions and the deformations of the Poisson bracket; moreover, the vanishing of the first cohomology group guarantees the integrability for a bi-Hamiltonian system (see \cite{dsk13,kra}). In this section, we choose to define the Poisson cohomology of $\F$ using a Poisson bivector and the notion of Schouten bracket, exactly as in the classical finite dimensional case or for commutative PDEs and D$\Delta$Es, without a direct reference to the underlying noncommutative structure. For the ultralocal case, Pichereau and Van de Weyer \cite{pvdw08,vdw} defined the double Poisson cohomology on $\hA_0$ and showed how this maps on the space $ \hF_0$.

The Schouten bracket we have defined on $\hF$ allows to define the adjoint action of a bivector $P$ on the space,
\begin{equation}\label{eq:adjact-def}
\mathrm{ad}_P\colon B\mapsto [P,B]\qquad\quad B\in \hF^p,\; \mathrm{ad}_P B\in \hF^{p+1}
\end{equation}

\begin{proposition}
Let $P$ be a Poisson bivector. Then $(\mathrm{ad}_P)^2=0$.
\end{proposition}
\begin{proof}
The proposition is an immediate consequence of the Jacobi identity for the Schouten bracket. Let $B$ be a $p$-vector field. We have
\begin{equation}
\left(\mathrm{ad}_P\right)^2 B=[P,[P,B]]=[[P,P],B]+(-1)^{1\cdot 1}[P,[P,B]],
\end{equation}
and the first term of the RHS vanishes because $P$ is a Poisson bivector. Then we have $[P,[P,B]]=-[P,[P,B]]=0$.
\end{proof}
This proposition allows us to call the adjoint action of $P$ the \emph{Poisson differential} and to denote it as $\ud_P$.
\begin{definition}
The space of local polyvector field $\hF$, endowed with a Poisson differential $\ud_P$, is the Poisson-Lichnerowicz complex of $(\F,P)$.
\begin{equation}
\begin{CD}
0 @>>> \F @>\ud_P>> \hF^1 @>\ud_P>> \hF^2\ @>\ud_P>>\cdots @>>> \cdots
\end{CD}
\end{equation}
The cohomology of the complex is called the \emph{Poisson cohomology} of $(\F,P)$, that is,
\begin{equation}
H(P,\F)=\bigoplus_{p=0}^\infty H^p(P,\F)=\frac{\ker \ud_P\colon\hF^p\to\hF^{p+1}}{\im \ud_P\colon\hF^{p-1}\to\hF^p}.
\end{equation}
\end{definition}

The computation of the Poisson cohomology, even in the commutative case, is a challenging task. For the commutative differential and difference case, see for instance \cite{ccs,CW19,lz11}; for the noncommutative ultralocal case,  several examples of double Poisson cohomology and their relation with the Poisson cohomology on the space of local functionals have been obtained and discussed \cite{akkn20, pvdw08, vdw}. An investigation of the Poisson cohomology for the local nonabelian case will be discussed in a forthcoming work. 
\section{Quasi-Poisson structures and Hamiltonian structures}\label{sec:quasi}
In this section we focus on ultralocal operators, which -- as we have seen in Section \ref{ssec:dPA} -- coincide with the class of operators used to describe ordinary differential equations.

In 2012, T. Wolf and O. Efimovaskaya investigated the integrability of a two-component system of ODEs proposed by Kontsevich \cite{WE12}:
\begin{equation}\label{eq:kont}
\begin{cases}
u_t=uv-uv^{-1}-v^{-1}&\\
v_t=-vu+vu^{-1}+u^{-1}
\end{cases}
\end{equation}
This system is integrable, possessing a Lax pair representation; this allows to compute its infinite series of conserved quantities and to find the corresponding hierarchy of symmetries. It can be cast in ``Hamiltonian'' form \eqref{eq:HamSyst_expl}, using the operator (first identified by Mikhailov and Sokolov in \cite{ms00})
\begin{equation}\label{eq:kont-op}
H=\begin{pmatrix}
\r_{u^2}-l_{u^2} & \l_{uv}+\l_u\r_v-\l_v\r_u+\r_{vu}\\
-\r_{uv}+\l_u\r_v-\l_v\r_u-\l_{vu} & \l_{v^2}-\r_{v^2}
\end{pmatrix}
\end{equation}
and the local functional
\begin{equation}\label{eq:h}
h=\frac12\tr\left(u+v+u^{-1}+v^{-1}+u^{-1}v^{-1}\right).
\end{equation}
Let us call $Q$ the bivector defined by the operator $H$,
\begin{equation}
Q=\frac12\tr\sum_{i,j}\left(\theta_i H^{ij}(\theta_j)\right).
\end{equation}
The bivector does not satisfy the Poisson condition \eqref{eq:HamProp}, namely $[Q,Q]\neq 0$. However, in their paper \cite[\S 3]{WE12}, Wolf and Efimovaskaya observe that $H$ enjoys the property
\begin{equation}\label{eq:pre-Ham}
\mathcal{L}_{X_h}(Q)=0,
\end{equation}
where $\mathcal{L}$ is the Lie derivative and $X_h$ is the ``Hamiltonian'' vector field $H(\delta h)$. This property for the operator $H$, which is system-dependent, allows us to employ it in most of the constructions which would normally involve a \emph{bona fide} Hamiltonian operator.

We recall that a \emph{conserved quantity} for the system defined by an evolutionary vector field $X$ is a functional $f$ such that $X(f)=0$. It is well known that the Lie derivative of a polyvector field can be written in terms of the Schouten bracket, by $\mathcal{L}_X(B)=[X,B]$. Hence, for a conserved functional $f$ we can write
$$\mathcal{L}_X(f)=[X,f]=0.$$
Similarly, a vector field $Y$ is a \emph{symmetry} of the system $X$ if and only if $[X,Y]=0$.

\begin{proposition}\label{thm:preH-1}
Let $Q$ be the bivector defined by an operator $H$, such that it satisfies \eqref{eq:pre-Ham} for the system $X_h$. Then $Q$ maps conserved quantities of the system into symmetries.
\end{proposition}
\begin{proof}
Let $f$ be a conserved quantity for the system $X_h$.
From \eqref{eq:hamvf_sch} we have then $-[[Q,h],f]=0$. The bivector $Q$ maps the conserved quantity $f$ into the vector field $X_f=-[Q,f]$; equivalently, the operator $H$ maps the conserved quantity into the characteristics of the vector field by $H(\delta f)$. By the Jacobi identity for the Schouten bracket we have
\begin{equation}
[X_f,X_h]=[[Q,f],[Q,h]]= [Q,[f,[Q,h]]]+[f,[Q,[Q,h]]=0,
\end{equation}
namely the vector field $X_f$ is a symmetry of the system $X_h$. Indeed, the first term of the RHS vanishes because $[f,[Q,h]]=-[[Q,h],f]=X_h(f)$ and $f$ is a conserved quantity of $X_h$, while the second one does because $[Q,[Q,h]]=-[[Q,h],Q]=\mathcal{L}_{X_h}(Q)=0$ by \eqref{eq:pre-Ham}.
\end{proof}
\begin{proposition}
Let $f$ and $g$ be conserved quantities for the system $X_h$. Then the bracket
\begin{equation}
\{f,g\}:=X_g(f)=[[Q,f],g]
\end{equation}
is a conserved quantity of the system, too.
\end{proposition}
\begin{proof}
We need to prove
\begin{equation}\label{eq:preh-th-1}
X_h\left(\{f,g\}\right)=0.
\end{equation}
Using the definition of ``Hamiltonian'' vector field and of bracket we rewrite the LHS of \eqref{eq:preh-th-1} as
\begin{equation}
-[[Q,h],[[Q,f],g]]
\end{equation}
which, because of the Jacobi identity for the Schouten bracket, is equal to
\begin{equation}
-[[[Q,h],[Q,f]],g]+[[Q,f],[[Q,h],g]].
\end{equation}
The first term vanishes because $[[Q,h],[Q,f]]=[X_h,X_f]$ and we have proved in Proposition \ref{thm:preH-1} that the vector field associated to a conserved quantity $f$ commutes with $X_h$. The second term vanishes, too, because $[[Q,h],g]=-X_h(g)$ and $g$ is a conserved quantity.
\end{proof}

Property \eqref{eq:pre-Ham} is sufficient to explain why $H$ maps conserved quantities into commuting symmetries; however, if we use the bivector $Q$ to define a bracket according to \eqref{eq:defBra-sch}, we obtain an operation which satisfies Jacobi identity, namley \emph{a Poisson brackets defined by a non-Poisson bivector}. The property identified by Wolf and Efimovskaya is not sufficient to guarantee this outcome.

A class of non-Poisson structures giving rise to Poisson brackets in some quotient space (as our space of local functional $\F$ is) was originally introduced by Alekseev, Kosmann-Schwarzbach and Meinrenken as quasi-Poisson manifolds \cite{akm02}. In the non-commutative case, Van Den Bergh \cite{vdb} introduced a twisted version of double Poisson algebras called \emph{double quasi-Poisson algebras}.

\subsection{Double quasi-Poisson algebras}
The definition of double quasi-Poisson algebra we present in this section is given in the form proposed by Fairon \cite{f19}. The notion has been introduced by Van Den Bergh \cite{vdb} in his seminal work on double Poisson algebras, but the more modern version is equivalent and requires less background material.

Let $\A_0$ be an associative but non commutative algebra whose identity admits a finite decomposition in terms of orthogonal idempotents,
$$
\mathbbm{1}_{\A_0}=\sum_{s=1}^n e_s,
$$
with $e_se_t=\delta_{st}e_s$. We can then regard $\A_0$ as a $\mathcal{B}$-algebra for $\mathcal{B}=\oplus_s \mathbb{K}e_s$. Let us now consider a $\mathcal{B}$-linear double bracket $\ldb-,-\rdb$ on $\A_0$; its associated triple bracket is
\begin{equation}\label{eq:tripleb-int}
\ldb -,-,-\rdb=\sum_{s=1}^3 \tau^s \ldb -,\ldb -,-\rdb\rdb_L\tau^{-s},
\end{equation}
which is an alternative way of writing \eqref{eq:tripleb}.
\begin{remark}
We are reproducing Fairon's definition, that allows $\A_0$ to have a decomposable unit. The standard example for this is the double quasi-Poisson algebra realised on the path algebra of a quiver; the identity in such an algebra is obtained as the sum of the ``stationary paths'' associated to each vertex of the quiver \cite{vdb}.
\end{remark}
\begin{definition}\label{def:qP}
We say that a $\mathcal{B}$-algebra $\A_0$, endowed with a $\mathcal{B}$-linear double bracket $\ldb-,-\rdb$, is a \emph{double quasi-Poisson bracket} if it satisfies
\begin{equation}\label{eq:qP-def}
\begin{split}
\ldb a,b,c\rdb&=\alpha\sum_s\left( ce_sa\otimes e_sb\otimes e_s-ce_sa\otimes e_s\otimes be_s-ce_s\otimes ae_sb\otimes e_s+ce_s\otimes ae_s\otimes be_s\right.\\
& \left.-e_sa\otimes e_sb\otimes e_s c +e_sa\otimes e_s\otimes be_sc+e_s\otimes ae_sb\otimes e_sc-e_s\otimes ae_s\otimes be_s c\right),
\end{split}
\end{equation}
for some $\alpha\neq 0$ and all triples $a,b,c\in \A_0$.
\end{definition}

The remarkable feature of double quasi-Poisson algebras is that, despite their triple bracket does not vanish, the bracket defined on the space $\F_0=\A_0/[\A_0,\A_0]$ is a Poisson bracket (in particular, it satisfies the Jacobi identity for any triple of entries). Moreover, since the vanishing of \eqref{eq:qP-def} is equivalent to the vanishing of the triple brackets among all the generators of $\A_0$ \cite{vdb}, we have a quick and explicit way to verify whether an ultralocal operator is quasi-Poisson (or, more precisely, defines the bracket of a double quasi-Poisson algebra).
\begin{theorem}[\cite{vdb}]\label{thm:qP-Ham}
Let $(\A_0,\ldb-,-\rdb)$ be a double quasi-Poisson algebra. Then the bracket defined on $\F_0$ as in \eqref{eq:PoisB-PA-def} is a Poisson bracket.
\end{theorem}
\begin{proof}
The skewsymmetry of the bracket on $\F_0$ is guaranteed by the skewsymmetry of the double bracket on $\A_0$.
From \cite[Proposition 2.4.2 and Corollary 2.4.4]{vdb}, a graded version of which we prove as Equation \eqref{eq:sch_Jac_pf4} in Theorem \ref{thm:schJac} of Appendix \ref{app:Jacobi}, we have
\begin{multline}\label{eq:qP-Ham-pf1}
\{\tr a,\{\tr b,\tr c\}\}+\{\tr b,\{\tr c,\tr a\}\}+\{\tr c,\{\tr a,\tr b\}\}\\
=\tr m\left((m\otimes 1) \ldb a,b,c\rdb-(1\otimes m)\ldb b,a,c\rdb\right).
\end{multline}
From \eqref{eq:qP-def} we have that the RHS of \eqref{eq:qP-Ham-pf1} vanishes, so that the bracket $\{-,-\}$ satisfies the Jacobi identity and it is, hence, a Poisson bracket on $\F_0$. 
\end{proof}
Theorem \ref{thm:qP-Ham} states that a double quasi-Poisson algebra on $\A_0$, whose bracket \emph{does not} define a Poisson operator (in the sense of \eqref{eq:HamProp}), defines nevertheless a Poisson bracket on $\F_0$.

Definition \ref{def:qP} is more general than what we need in our discussion. The algebra $\A_0$ we consider is the space of Laurent polynomials in the generators $\{u^i\}$ and $\mathcal{B}=\R$ with $1$ as the only idempotent: we can henceforth drop the sum over the different idempotents $e_s$ from our formulae.

\begin{proposition}\label{thm:kont-quasi}
The double bracket defined by the skewsymmetric ultralocal operator \eqref{eq:kont-op} is the bracket of a double quasi-Poisson algebra.
\end{proposition}
\begin{proof}
According to Proposition \ref{thm:poisdouble-ul}, the double bracket among the generators is
\begin{align}
\ldb u,u\rdb&=1\otimes u^2-u^2\otimes 1, & \ldb u,v\rdb&=-1\otimes uv+u\otimes v-v\otimes u-vu\otimes 1,\\
\ldb v,v\rdb&=v^2\otimes 1-1\otimes v^2, & \ldb v,u\rdb&=uv\otimes 1+u\otimes v-v\otimes u+1\otimes vu.
\end{align}

This double bracket is a special case of the quasi-Poisson bracket obtained in \cite[Theorem 3.5]{f19}, for $i=1$. The result we explicitly show is due to Massuyeau and Turaev \cite{mt}. In principle, we would have to compute the four triple bracket $\ldb u,u,u\rdb$, $\ldb u,u,v\rdb$, $\ldb u,v,v\rdb$, and $\ldb v,v,v\rdb$. However, given the apparent symmetry of the expression in the exchange of $u$ with $v$ we restrict ourselves to the first two ones.

Computing $\ldb u,u,u\rdb$ gives us
\begin{equation}\label{eq:kont-qP-pf1}
\begin{split}
\ldb u,u,u\rdb&=\left(1+\tau+\tau^2\right)\ldb u,\ldb u,u\rdb\rdb_L\\
&=-\left(1+\tau+\tau^2\right)\left(\ldb u,u^2\rdb\otimes 1\right)\\
&=\left(1+\tau+\tau^2\right)\left(u^2\otimes u\otimes 1-u\otimes u^2\otimes 1\right).
\end{split}
\end{equation}
On the other hand, the RHS of \eqref{eq:qP-def} is
\begin{equation}
\alpha\left(u^2\otimes u\otimes 1-u^2\otimes 1\otimes u-u\otimes u^2\otimes 1+u\otimes 1\otimes u^2+1\otimes u^2\otimes 1-1\otimes u\otimes u^2\right),
\end{equation}
which is equal to \eqref{eq:kont-qP-pf1} for $\alpha=1$. The same computation for $\ldb u,u,v\rdb$ gives us
\begin{equation}
\begin{split}
\ldb u,u,v\rdb&=1\otimes u^2\otimes v+v\otimes u\otimes u-u\otimes u\otimes v-v\otimes u^2\otimes 1+vu\otimes u\otimes 1\\&\quad+u\otimes 1\otimes uv-1\otimes u\otimes uv-vu\otimes 1\otimes u,
\end{split}
\end{equation}
which is, again, the RHS of \eqref{eq:qP-def} for $\alpha=1$.
\end{proof}
It follows from Theorem \ref{thm:qP-Ham} that the Mikhailov and Sokolov's operator \eqref{eq:kont-op} defines a Poisson bracket on $\F_0$ and it is, hence, Hamiltonian.

\subsection{Hamiltonian operators and quasi-Poisson bivectors}\label{ssec:HamBrack}
In Section \ref{sec:Ham_formalism} we have defined a Hamiltonian operator as an operator on $\A$, which induces a Lie algebra structure on $\F$ by means of the Poisson bracket. In Section \ref{ssec:pois} we showed the properties that a bivector must satisfy in order to do the same. Proposition \ref{thm:pois-suf} says that having a Poisson bivector is a sufficient condition. However, a bivector $B$ satisfying the condition
\begin{equation}\label{eq:qP-def0}
[[B,B],F]=0 \qquad\qquad\forall F\in\F,
\end{equation}
which appears to be less strict than $[B,B]=0$, still defines a Hamiltonian structure: indeed, condition \eqref{eq:qP-def0}  guarantees both the validity of the Jacobi identity by \eqref{eq:jac-gen} and the existence of the Lie algebra action by \eqref{eq:act-gen}.

Note that property \eqref{eq:pre-Ham} can be rewritten as $[[Q,Q],h]=0$ (we showed this in the proof of Proposition \ref{thm:preH-1}), namely as identity \eqref{eq:qP-def0} for a particular $h\in\F_0\subset\F$. Identity \eqref{eq:qP-def0} is a natural generalisation of the Poisson property, relaxing the condition and still obtaining a Hamiltonian structure. Indeed, we have the following theorem:
\begin{theorem}\label{thm:qp-gen}
Let $Q$ be a quasi-Poisson bivector. Then $[[Q,Q],f]=0$ for all $f\in\F_0$. Thus, $Q$ defines a Hamiltonian structure.
\end{theorem}
\begin{proof}
From the computation performed in the proof of Proposition \ref{thm:poisdouble-ul} and Lemma \ref{lem:PoisBiv-pf}, we have that then
\begin{equation}\label{eq:qp-gen-pf-00}
[Q,Q]=-\frac23\sum_{i,j,k}\tr\left(\theta_i\ldb u^j, u^k, u^i \rdb'\theta_j\ldb u^j, u^k, u^i \rdb''\theta_k\ldb u^j, u^k, u^i \rdb'''\right),
\end{equation}
where we extend Sweedler's notation to elements of $A^{\otimes 3}$ by $\ldb u,u,u\rdb=\ldb u,u,u\rdb'\otimes\ldb u,u,u\rdb''\otimes\ldb u,u,u\rdb'''$. For convenience, we drop the sum over the indices from Equation \eqref{eq:qp-gen-pf-00}. From the quasi-Poisson property \eqref{eq:qP-def} for $\alpha=1$, after changing the indices and reordering the terms we have
\begin{equation}
[Q,Q]=-\frac23\tr\left(3\theta_i\theta_ju^ju^k\theta_ku^i-3\theta_i\theta_ju^j\theta_ku^ku^i+\theta_iu^i\theta_ju^j\theta_ku^k-\theta_iu^j\theta_ju^k\theta_ku^i\right).
\end{equation}
A direct computation with \eqref{eq:ds_master_ul} gives
\begin{equation}\label{eq:qp-gen-pf0}
\begin{split}
-\frac12[[Q,Q],f]=\tr&\left[\theta_i\theta_ju^j\left(u^lf_l''f_l'-f''_lf'_lu^l\right)u^i+\theta_iu^i\theta_ju^j\left(f''_lf'_lu^l-u^lf''_lf'_l\right)\right.\\
&\left.+\theta_iu^j\theta_j\left(f''_lf'_lu^l-u^lf''_lf'_l\right)u^i+\theta_iu^iu^j\theta_j\left(u^lf''_lf'_l-f''_lf'_lu^l\right)\right],
\end{split}
\end{equation}
where we denote $\dev_{u^l} f=f'_l\otimes f''_l$. The main point of the proof is showing that the terms in the brackets vanish, when summed over all $l$'s and all the terms in the double derivative of $f$.

Let $f$ be a linear combination of monomials of the form
\begin{equation}
f=(u^{i_1})^{\pm 1}(u^{i_2})^{\pm 1}\cdots(u^{i_d})^{\pm 1}.
\end{equation}
Then
\begin{equation}\label{eq:qp-gen-pf1}
\begin{split}
\frac{\dev f}{\dev u^l}&={\sum_{s=1}^d}'\delta_{l,i_s}(u^{i_1})^{\pm 1}\cdots (u^{i_{s-1}})^{\pm 1}\otimes (u^{i_{s+1}})^{\pm 1}\cdots (u^{i_{d}})^{\pm 1}\\
&\quad -{\sum_{s=1}^d}''\delta_{l,i_s}(u^{i_1})^{\pm 1}\cdots (u^{i_{s-1}})^{\pm 1}(u^{i_s})^{-1}\otimes (u^{i_s})^{-1}(u^{i_{s+1}})^{\pm 1}\cdots (u^{i_{d}})^{\pm 1},
\end{split}
\end{equation}
where with $\sum'$ we denote the sum over all the factors $u^{i_s}$ with power $+1$ and $\sum''$ the one over all the factors with power $-1$. The expressions within the parentheses in the RHS of \eqref{eq:qp-gen-pf0} have then the form
\begin{equation}
u^l f_l''f_l'-f_l''f_l'u^l=\sum_{s=1}^d\left[(u^{i_{s+1}})^{\pm1}\cdots(u^{i_d})^{\pm1}(u^{i_1})^{\pm1}\cdots (u^{i_s})^{\pm1}-(u^{i_s})^{\pm1}(u^{i_{s+1}})^{\pm1}\cdots(u^{i_{s-1}})^{\pm1} \right],
\end{equation}
where distinguishing between the two different forms of the partial derivative \eqref{eq:qp-gen-pf1} is no longer needed. The four expressions $(u^l f_l''f_l'-f_l''f_l'u^l)$ (resp. $(f_l''f_l'u^l- u^l f_l''f_l')$) vanish when summing over all the cyclic permutations of the monomial $f$, therefore proving the vanishing of $[[Q,Q],f]$ for any $f$.
\end{proof}

We provided a geometric interpretation for double quasi-Poisson algebras; we have showed that the operator \eqref{eq:kont-op} is not Poisson (having a nonvanishing triple bracket, see Proposition \ref{thm:kont-quasi}), but it is quasi-Poisson as for Proposition \ref{thm:kont-quasi}. Then we can conclude that it is Hamiltonian. Moreover, we showed that property \eqref{eq:pre-Ham}, noted by Wolf and Efimovskaya, is in fact a consequence of the stronger property \eqref{eq:qP-def0}. However, note that the notion of quasi-Poisson algebra is defined on $\A_0$ and not on $\A$, since it is given in terms of double Poisson brackets and \emph{not $\lambda$ brackets}. Noncommutative Hamiltonian systems of ODEs defined in terms of quasi-Poisson brackets have been studied in the last few years by several authors \cite{ar18,aros20,cf17,cf20}; we are not aware of any example of non-Poisson operators defining Poisson brackets, and hence being labelled Hamiltonian, for systems of PDEs or of D$\Delta$Es. However, should a differential or difference operator exist such that it fulfills \eqref{eq:qP-def0}, we could still call the corresponding bivector a quasi-Poisson one. 

\section{Nonabelian Hamiltonian operators for difference systems}\label{sec:Ham-ex}
In this section we will present several examples  of nonabelian Hamiltonian structures, applying the results we we have presented in the previous sections. All the results are described using the bivector formalism recalled in Section \ref{sec:Ham_formalism}, which is better known among the Integrable Systems community. However, most of the computations were performed using the Schouten bracket described in our ``geometric'' setting  of Section \ref{sec:geom}.

The operators we discuss, some of which are not previously known, contribute to the study of nonabelian differential-difference integrable systems. More in detail, we investigate \emph{scalar} ultralocal and local Hamiltonian operators. In the ultralocal case, we show that all the Hamiltonian structures coincide with Hamiltonian structures for nonabelian ODEs, for which we proved in Section \ref{sec:dPA} and \ref{sec:quasi} the relation with double Poisson algebras and quasi-Poisson algebras. We then study local Hamiltonian structures and present a class of nonlocal ones; finally, we provide an answer to a question left open in our recent work \cite{cw19-2}, exhibiting the Hamiltonian structures for, respectively, the nonabelian Kaup, Ablowitz-Ladik, and Chen-Lee-Liu lattices.

	\subsection{Scalar ultralocal and local Hamiltonian operators}

	A scalar ($\ell=1$) ultralocal (see Definition \ref{def:ul}) skewsymmetric operator must be of the form
	$$
	K=\sum_\alpha \left(\l_{f^{(\alpha)}}\r_{g^{(\alpha)}}-\r_{f^{(\alpha)}}\l_{g^{(\alpha)}}\right),
	$$
	with $f^{(\alpha)}, g^{(\alpha)} \in \A$. We have the following Lemma:
	\begin{lemma}\label{lem:ul}
	A necessary condition for a skewsymmetric scalar ultralocal operator $K$ to be Poisson is that $f^{(\alpha)}=f^{(\alpha)}(u)$ and $g^{(\alpha)}=g^{(\alpha)}(u)$, namely the operator must multiply on the left and on the right for polynomials of $u$'s only.
	\end{lemma}
	\begin{proof}
	The Poisson property in this case reads
	\begin{equation}
	\sum\mathbf{pr}_{K\theta}\int\tr\theta f^{(\alpha)}\theta g^{(\alpha)}=0.
	\end{equation}
	Let
	\begin{align}
	f^{(\alpha)}&=c_{(\alpha)} u_{i^{(\alpha)}_1}\ldots u_{i^{(\alpha)}_r}&\text{and}&&g^{(\alpha)}&=u_{j^{(\alpha)}_1}\ldots u_{j^{(\alpha)}_s}
	\end{align}
	for constants $c_{(\alpha)}$. We define $p=\max_{l,\alpha}\{i^{(\alpha)}_l\}$ and $q=\max_{l,\alpha}\{j^{(\alpha)}_l\}$.

	Let us first consider the case $p>q>0$ (or, equivalently, $q>p>0$ switching the role of $p$ and $q$ in the proof). Then, the only terms in $\mathbf{pr}_{K\theta}P$ including $\theta_p$ is of the form
	\begin{multline}\label{eq:lemul-pf}
	\sideset{}{'}\sum_{\alpha,\beta, i^{(\alpha)}_l=p}\int\tr c_{(\alpha)}\left((\cS^p g^{(\beta)}) u_{i^{(\alpha)}_{l+1}}\cdots u_{i^{(\alpha)}_r}\theta g^{(\beta)} \theta u_{i^{(\alpha)}_1}\cdots u_{i^{(\alpha)}_{l-1}}(\cS^p f^{(\beta)})\right.\\\left.-(\cS^p f^{(\beta)}) u_{i^{(\alpha)}_{l+1}}\cdots u_{i^{(\alpha)}_r}\theta g^{(\beta)} \theta u_{i^{(\alpha)}_1}\cdots u_{i^{(\alpha)}_{l-1}}(\cS^p g^{(\beta)})\right)\theta_p,
	\end{multline}
	where the sum runs for all $\beta$ and for $\alpha$ such that $f^{(\alpha)}$ depends on $u_p$ and for the indices $l$ such that $i^{(\alpha)}_l=p$. Note that the presence of $\theta_p$, for $p\neq 0$, fixes the position of all the terms with respect to the cyclic permutations, and hence that the expression can vanish only if $f^{(\beta)}=g^{(\beta)}$. If $p=q\neq0$ the picture is similar, and in the sum there will be present additional terms with expression multiplying (from both the left and the right) the expression $\theta f^{(\alpha)}\theta$. An analogue result holds if we consider the variables with the minimum negative index; the only way for the expression \eqref{eq:lemul-pf} to vanish without requiring $f^{(\alpha)}=g^{(\alpha)}$ is by allowing $p=q=0$, so that we can exploit the cyclic permutations of the products. 
	\end{proof}
	Lemma \ref{lem:ul} implies that all the scalar ultralocal Poisson operators \emph{in the differential-difference setting} are Poisson operators for nonabelian \emph{ordinary differential equations}, too. In Section \ref{ssec:dPA} we discussed the equivalence between Poisson structures for nonabelian ODEs and the notion of double Poisson algebras. The classification results for the latter ones provide an equivalent classification of ultralocal Poisson operators: we can then provide a list of ultralocal Poisson operators based on \cite{dskv15, ms00, p16, vdb}.
	
	\begin{theorem}\label{thm:ul-scal}
	(1) All the scalar Hamiltonian ultralocal operators are of the form
	\begin{equation}\label{eq:ul-scal}
	H=\alpha\c_u+\beta\c_{u^2}+\gamma\left(\l_{u^2}\r_u-\l_u\r_{u^2}\right)
	\end{equation}
	These operators are Poisson if and only if $\beta^2-\alpha\gamma=0$.
	
	(2) The Poisson operators $H_1$ and $H_2$ (with their respective constants $\alpha_i$, $\beta_i$,$\gamma_i$) form a bi-Hamiltonian pair if and only if $2\beta_1\beta_2-\alpha_2\gamma_1-\alpha_1\gamma_2=0$.
	\end{theorem}
	\begin{proof}
From Lemma \ref{lem:ul} we know that we must investigate only operators without shifted variables. Then, for part (1) we can rely on the result due to Van Der Bergh \cite{vdb} in the context of double Poisson algebras. It can easily verified by \eqref{eq:HamProp} that condition $\beta^2-\alpha\gamma=0$ is necessary and sufficient for the Poisson property. Powell \cite{p16} gives a full proof of this fact. However, for any value of the constants the operator corresponds to a double quasi-Poisson algebra \cite{f19}.

On the other hand, $2\beta_1\beta_2-\alpha_2\gamma_1-\alpha_1\gamma_2=0$ is equivalent to the Poisson property for $H_1-\lambda H_2$, for any $\lambda$, giving (2).
\end{proof}

	\begin{definition}
	We call a difference operator \emph{local} if its entries (or itself in the scalar case) are Laurent polynomials in $\cS$. 	\end{definition}

	Because of the skewsymmetry requirement, a scalar operator is of order $(-N,N)$ for some $0<N<\infty$.
	
	The Poisson property imposes very rigid constraints on the form of such operators, so that very few of them are known in the literature: in two components, for instance, the only nonconstant local (in particular, not ultralocal) Hamiltonian operator we know is the first Hamiltonian structure of the Toda lattice \eqref{eq:HToda1}. By solving the equation \eqref{eq:HamProp} for a scalar operator of order $(-1,1)$ we have found a new class of examples, that are novel at the best of our knowledge.
	
	\begin{theorem}
	For $\ell=1$, all the Poisson operators of order $(-1,1)$ are, up to linear transformations of the dependent variable $u$, of one of the following forms
	\begin{align}
	H_1&=\l_{uu_1}\r_{u_1u}\cS-\cS^{-1}\r_{uu_1}\l_{u_1u},\\
	H_{\mathrm{c}}&=\cS-\cS^{-1}.
	\end{align}
	
	\end{theorem}

	\begin{proof}
	The skewsymmetry condition implies that a candidate Hamiltonian operator must be of the form
	$$
	H=\sum\left(\l_{H^{(\alpha_1)}_L}\r_{H^{(\alpha_1)}_R}\cS-\cS^{-1}\r_{H^{(\alpha_1)}_L}\l_{H^{(\alpha_1)}_R} + \l_{H^{(\alpha_0)}_L}\r_{H^{(\alpha_0)}_R}-\r_{H^{(\alpha_0)}_L}\l_{H^{(\alpha_0)}_R}\right).
	$$
	The Poisson bivector associated to the operator is
	$$
	P=\int\tr\sum\left(\theta H^{(\alpha_1)}_L\theta_1 H^{(\alpha_1)}_R+\theta H^{(\alpha_0)}_L\theta H^{(\alpha_0)}_R\right),
	$$
	and a computation similar to the one performed in the proof of Lemma \ref{lem:ul} shows that $H^{(\alpha_1)}_L=H^{(\alpha_1)}_L(u,u_1)$, $H^{(\alpha_1)}_R=H^{(\alpha_1)}_R(u,u_1)$, $H^{(\alpha_0)}_L=H^{(\alpha_0)}_L(u)$, and $H^{(\alpha_0)}_R=H^{(\alpha_0)}_R(u)$. Moreover, a necessary condition emerging from comparing the terms of the expression \eqref{eq:HamProp} for expressions containing $(\theta, \theta_1, \theta_2)$ is that $H^{(\alpha_1)}_L$ (respectively, $H^{(\alpha_1)}_R$) must be all equal (or at least proportional) and $H_L=(\lambda u+\mu)(\lambda u_1+\mu)$ and $H_R=(\nu u_1+\rho)(\nu u+\rho)$. Note in particular that we obtain $H_c$ for $\lambda=\nu=0$. The conditions $\lambda=\nu$ and $\mu=\rho$ come from the vanishing of the terms with $(\theta,\theta_1,\theta_1)$. The terms containing $(\theta,\theta,\theta)$ in the identity can come only from the ultralocal term, which must on its own be Poisson: they can be then either $\c_u$ or $\l_u\r_u^2-\l_u^2\r_u$, but not a linear combination of these two. We then obtain the statement checking case by case. In particular, $H_1$ and $H_\mathrm{c}$ are not compatible.
	\end{proof}
This class of examples can be extended to an arbitrary operator of order $(-N,N)$, closely resembling the so-called multiplicative Poisson $\lambda$-bracket of general type defined in \cite{dskvw19}. It is easy to verify that
\begin{equation}
H_p=\l_{uu_p}\r_{u_pu}\cS^p-\cS^{-p}\l_{u_pu}\r_{uu_p}
\end{equation}
is Hamiltonian for any $p>0$ and that any linear combination of $H_p$, for different $p$'s, is Hamiltonian too. Note, however, that the condition for the nonabelian case is much more rigid: the form of the operators in the commutative case depends on an arbitrary function of the variable $u$ (see \cite{dskvw19}).


	\subsection{Nonlocal Hamiltonian operators}
	Similarly to the differential case, many systems, whose Hamiltonian structures are local in the Abelian case, are Hamiltonian only with respect to nonlocal Hamiltonian operators in the noncommutative case. In this section we recall some results already presented in \cite{cw19-2}, and then we exhibit new two-components ($\ell=2$) Hamiltonian structures which reduce to ultralocal brackets in the commutative case. In particular, they provide the Hamiltonian structures for the Nonabelian Ablowitz-Ladik, Chen-Lee-Liu and Kaup lattices which had not been constructed before.
	
	The prototypical example of nonlocal Hamiltonian structures reducing to local ones in the commutative case is the Hamiltonian structure of nonabelian Volterra chain, that we presented in \cite{cw19-2} (the sign difference is due to the opposite definition of $\c_u$):
		\begin{equation}\label{eq:hamVolt}
		H_V=\r_u\cS\r_u-\l_u\cS^{-1}\l_u-\r_u\c_u-\c_u(1-\cS)^{-1}\c_u.
		\end{equation}
		
		Note that the last two terms can be written in a form which is skewsymmetric at sight, namely $H_V=\r_u\cS\r_u-\l_u\cS^{-1}\l_u-H_0^{(\mathrm{sc})}$ with
		\begin{equation}\label{eq:Hsc0}
		H_0^{(\mathrm{sc})}=\frac{1}{2}\left(\a_u\c_u-\c_u(1+\cS)(\cS-1)^{-1}\c_u\right).
		\end{equation}
		\begin{proposition}\label{thm:Hsc0}
		The operator $H_0^{(\mathrm{sc})}$ is Poisson, and therefore Hamiltonian.
		\end{proposition}
		\begin{proof}
		In \cite[Proposition 7]{cw19-2} we proved that \eqref{eq:hamVolt} is a Poisson operator. Let us introduce the nonlocal variable
		$$
		\rho=(\cS-1)^{-1}(u\theta-\theta u)
		$$
		with the useful identity
		\begin{equation}\label{eq:rhoid}
		\rho_1:=\cS\rho=\rho+(u\theta-\theta u).
		\end{equation}
		The characteristics of the formal vector field $H\theta:=H_0^{(\mathrm{sc})}\theta$ is
		\begin{equation}
		H\theta=u\theta u-\theta u^2+\rho u-u\rho
		\end{equation}
		and its associated bivector is
		\begin{equation}
		2P=\int\tr\left(-u^2\theta\theta+\rho_1\rho\right).		
		\end{equation}
		To compute \eqref{eq:HamProp} we need to obtain an explicit form for the prolongation of the formal vector field applied to the nonlocal terms. For a generic formal vector field of characteristics $V$ (and degree 1 in $\theta$) we have, as illustrated with more detail in \cite{cw19-2},
		\begin{align}\label{eq:prvNL-1}
\mathbf{pr}_V\int\tr \rho_1 \rho&=\int\tr\left[\rho\mathbf{pr}_V(\rho_1)-\rho_1\mathbf{pr}_V(\rho)\right]\\
&=\tr\int\left[\rho\mathbf{pr}_V(\rho_1-\rho_{-1})\right]=\int\tr\left[\rho\mathbf{pr}_V\left((1+\cS^{-1})(u\theta-\theta u)\right)\right]\\
&=\int\tr\left[\left((1+\cS)\rho\right)\mathbf{pr}_V(u\theta-\theta u)\right]\\
&=\int\tr\left[2 \rho\mathbf{pr}_V(u\theta -\theta u)+(u\theta-\theta u)\mathbf{pr}_V(u\theta-\theta u)\right].
\end{align}
An explicit and a bit tedious computation gives us
\begin{equation}
\mathbf{pr}_{H\theta}P=\int\tr\left[\rho(u\theta-\theta u)^2+\rho^2(u\theta-\theta u)-u\theta^2 u^2\theta+\theta^2 u\theta u^2\right],
\end{equation} 
The identity $\int\tr\left(\rho_1\rho_1\rho_1-\rho\rho\rho\right)=0$ implies
\begin{equation}
\int\tr\left[\rho(u\theta-\theta u)^2+\rho^2(u\theta-\theta u)\right]=-\frac13\int\tr(u\theta-\theta u)^3
\end{equation}
from which we obtain
\begin{equation}
\mathbf{pr}_{H\theta}P=\int\tr\left[-\frac13(u\theta-\theta u)^3-u\theta^2 u^2\theta+\theta^2 u\theta u^2\right],
\end{equation}
which vanishes. Hence, $H_0^{(sc)}$ is a Poisson operator.
\end{proof}

In \cite{cw19-2} we used the result for \eqref{eq:hamVolt} to proof by induction that the operator
		\begin{equation}
		H_{NIB}=\sum_{i=1}^p \left(\r_u\cS^i\r_u-\l_u\cS^{-i}\l_u\right)
	-H^{(\mathrm{sc})}_0	\end{equation}
		is the Hamiltonian operator for the Narita-Itoh-Bogoyavlensky lattice. Note that also in this case the Hamiltonian operator is the sum of an operator which is not Hamiltonian, but reduces to the Hamiltonian one for the corresponding commutative system, and of a Hamiltonian operator vanishing in the commutative case.
				
		\subsubsection{``Null'' Hamiltonian operators}\label{ssec:H0}
		A similar pattern as the one we have just observed can be also found in two-component systems. In \cite{cw19-2} we studied the nonabelian 2D Toda system. Its first Hamiltonian structure is local, but the second one -- obtained applying the recursion operator to the first structure -- is nonlocal and reduces to the standard one in the commutative case. Similarly to the nonabelian Volterra Hamiltonian structure, that operator can be regarded as the direct promotion of the Abelian Hamiltonian operator to the noncommutative case (which is not Hamiltonian) plus a Hamiltonian operator, vanishing in the commutative case.
		
\begin{theorem}\label{thm:H02}
The operator
\begingroup
\begin{equation}\label{eq:H02}
		H=\left(\begin{array}{ccc} \r_u\c_u-\c_u(\cS-1)^{-1}\c_u && \r_u\c_v-\c_u(\cS-1)^{-1}\c_v\\ \c_v\r_u-\c_v(\cS-1)^{-1}\c_u&\phantom{12}& \c_v\r_v-\c_v(\cS-1)^{-1}\c_v\end{array}\right)
		\end{equation}
\endgroup
is Poisson.
\end{theorem}

\begin{proof}
We show that the bivector defined by $H$ is Poisson; we denote $\theta$ and $\zeta$ the basic univectors corresponding, respectively, to $u$ and $v$.

Let $(\rho,\sigma)$ be the nonlocal variables 
\begin{align}
(\cS-1)^{-1}(u\theta-\theta u)&=\rho,&(\cS-1)^{-1}(v\zeta-\zeta v)&=\sigma,
\end{align} 	
and write the characteristics of the formal bivector $H\Theta$ as
\begin{equation}\label{H0-op}
H\binom{\theta}{\zeta}=\left(\begin{array}{c}-\theta u^2+u\theta u - u\rho+\rho u-\zeta v u + v\zeta u + \sigma u-u \sigma\\
-\theta u v + v\theta u + \rho v-v \rho-\zeta v^2+v\zeta v + \sigma v-v \sigma\end{array}\right).
\end{equation}
The bivector $P$ is then written as
\begin{equation}\label{H0-biv}
\begin{split}
2P&=\int\tr \Theta^\dagger H \Theta\\
&=\int\tr\left(-\theta\theta u^2-\zeta \zeta v^2+u\theta v \zeta+\zeta v \theta u-u\theta\zeta v-v\zeta\theta u\right.\\
&\qquad\qquad+\left(\rho_1+\sigma_1\right)\left(\rho+\sigma\right)\Big),
\end{split}
\end{equation}
where $\rho_1=\cS\rho$ and so on.

The condition \eqref{eq:HamProp} that must be checked for $P$ and $H\theta$ is an element of $\hF$ of degree 3 in $\theta$ and $\zeta$: indeed, the Schouten torsion of a bivector is a trivector. Because of the symmetry of the bracket in the exchange of $(u,\theta)$ and $(v,\zeta)$, it is sufficient that the homogeneous components respectively of degree 3 in $\theta$, and degree 2 in $\theta$ and 1 in $\zeta$ vanish. Moreover, the vanishing of the coefficient in degree 3 in $\theta$ (resp.~$\zeta$) is equivalent to the claim that the $(1,1)$ (resp.~$(2,2)$) entry of $H$, is Poisson in its own right. Indeed, they are both the same as $H_0^{\mathrm{(sc)}}$ given in equation \eqref{eq:Hsc0}, which is Poisson according to Proposition \ref{thm:Hsc0}.

Let us denote by $P^{\theta\theta}$, $P^{\theta\zeta}$, $P^{\zeta\zeta}$ the homogeneous components of $P$. Similarly, we denote $H\theta^\theta$, $H\theta^\zeta$, $H\zeta^\theta$, $H\zeta^\zeta$ the linear components in $\theta$ and $\zeta$ of each component of $H\Theta$.

We need to compute the terms of degree 2 in $\theta$ and degree 1 in $\zeta$ (the vice versa is the same by the symmetry of $H$), that are obtained by
\begin{equation}\label{eq:SchH0qqz}
\mathbf{pr}_{H\theta^\zeta}P^{\theta\theta}+\mathbf{pr}_{H\theta^\theta+H\zeta^\theta}P^{\theta\zeta}.
\end{equation}
Similarly to identity \eqref{eq:prvNL-1}, we also have
\begin{multline}\label{eq:prvNL-2}
\mathbf{pr}_V \int\tr\left(\sigma_1 \rho+\rho_1\sigma\right)=\\
\int\tr\left[2\sigma\mathbf{pr}_V(u\theta-\theta u)+(v\zeta-\zeta v)\mathbf{pr}_V(u\theta-\theta u)\right.\\
+\left.2\rho\mathbf{pr}_V(v\zeta-\zeta v)+(u\theta-\theta u)\mathbf{pr}_V(v\zeta-\zeta v)\right].
\end{multline}
The computation is then longer but similar to the one performed in the proof of Proposition \ref{thm:Hsc0}. We obtain purely local terms, terms with a single nonlocal variable and terms with two nonlocal variables. The latter ones are
\begin{equation}\label{eq:SchH0qqz-doublenonloc}
2\int\tr\left[\sigma\rho(u\theta-\theta u)+\sigma(u\theta-\theta u)\rho+(v\zeta-\zeta v)\rho\rho\right],
\end{equation}
which can be rewritten, from $\int\tr(\sigma_1\rho_1\rho_1-\sigma\rho\rho)=0$ and $\sigma_1=\sigma+v\zeta-\zeta v$, $\rho_1=\rho+u\theta-\theta u$, as 
\begin{multline}
\text{\eqref{eq:SchH0qqz-doublenonloc}}=-\int\tr\left[\sigma(u\theta-\theta u)(u\theta-\theta u)+\rho(u\theta-\theta u)(v\zeta-\zeta v)\right.\\
\left.+\rho(v\zeta-\zeta v)(u\theta-\theta u)+(v\zeta-\zeta v)(u\theta-\theta u)(u\theta-\theta u)\right].
\end{multline}
These terms cancel with the remaining ones in the expression \eqref{eq:SchH0qqz}. By symmetry, the same happens for the terms quadratic in $\zeta$ and linear in $\theta$, fulfilling \eqref{eq:HamProp}
\end{proof}

The operator $H$ in \eqref{eq:H02} vanishes in the Abelian case. It would be interesting to investigate the existence of integrable equations defined in terms of it, because they would not have a commutative counterpart. However, as we anticipated the second Hamiltonian structure of the nonabelian Toda system (see \cite[Section 3.3]{cw19-2}) is given by
\begin{equation}\label{eq:HToda1}
H^{(2)}_{Toda}=\left(\begin{array}{ccc}\l_u\cS^{-1}\l_u-\r_u\cS\r_u&\quad& \l_u\r_v-\r_u\cS\r_v\\
-\l_v\r_u+\l_v\cS^{-1}\l_u&\quad& \cS^{-1}\l_u-\r_u\cS\end{array}\right)+H.
\end{equation}

The operator \eqref{eq:H02} is not the only two-component Hamiltonian operator vanishing when we assume that the variables do commute. We also have
\begin{theorem}\label{thm:H02-tilde}
The operator
\begin{multline}\label{eq:H02-tilde}
\tilde{H}=\left(\begin{array}{cc}-\c_{u^2} & \l_u\r_v-\l_v\r_u-\l_{uv}+\r_{vu}\\
\r_v\l_u-\r_u\l_v+\r_{uv}-\l_{vu} & \c_{v^2}\end{array}\right)\\
+\begin{pmatrix}\a_u\\-\a_v\end{pmatrix}\cS(\cS-1)^{-1}\begin{pmatrix}\c_u & \c_v\end{pmatrix}-\begin{pmatrix}\c_u\\ \c_v\end{pmatrix}(\cS-1)^{-1}\begin{pmatrix}\a_u & -\a_v\end{pmatrix}
\end{multline}
is Poisson.
\end{theorem}
\begin{proof}
As for Theorem \ref{thm:H02} this reduces to cumbersome computations. Note that both the local and nonlocal parts are skewsymmetric at sight, since $(A(\cS-1)^{-1}B)^\dagger=-B^\dagger \cS (\cS-1)^{-1}A^\dagger$, with $\c_u^\dagger=-\c_u$, $\a_u^\dagger=\a_u$.
We introduce nonlocal variables
\begin{align*}
\lambda&=(\cS-1)^{-1}u\theta,&\mu=(\cS-1)^{-1}v\zeta,\\
\rho&=(\cS-1)^{-1}\theta u,&\sigma=(\cS-1)^{-1}\zeta v,
\end{align*}
and explicitly compute
\begin{equation}\label{eq:H02-tilde-op}
\tilde{H}\begin{pmatrix}\theta\\ \zeta\end{pmatrix}=2\begin{pmatrix} \lambda u-u\rho+u\mu-\sigma u\\ \theta uv-vu\theta+\rho v-v\lambda+v\sigma-\mu v
\end{pmatrix}.
\end{equation}
Then, its associated bivector is
\begin{equation}\label{eq:H02-tilde-biv}
\tilde{P}=\int\tr\left(\rho\rho_1-\lambda\lambda_1+\mu\mu_1-\sigma\sigma_1-2 u\theta\sigma+2\theta u \mu\right).
\end{equation}
To compute \eqref{eq:HamProp} we need, as before, some additional identities. We have
\begin{align}\label{eq:id-1}
\mathbf{pr}_V\int\tr\rho\rho_1&=-\int\tr 2\rho\mathbf{pr}_V(\theta u)+\theta u\mathbf{pr}_V(\theta u),\\\notag
\mathbf{pr}_V\int\tr\theta u\mu&=\int\tr\left[\mu\mathbf{pr}_V(\theta u)-\theta u\mathbf{pr}_V((\cS-1)^{-1}v\zeta)\right]\\\notag
&=\int\tr\left[\mu\mathbf{pr}_V(\theta u)+\rho_1\mathbf{pr}_V(v\zeta)\right]\\\label{eq:id-2}
&=\int\tr\left[\mu\mathbf{pr}_V(\theta u)+(\theta u+\rho)\mathbf{pr}_V(v\zeta)\right],
\end{align}
and similarly for the other nonlocal variables. As previously, we only need to check the vanishing of the expression for $\theta \theta \theta$ and $\theta\theta\zeta$ (and their corresponding nonlocal variables).

For the $\theta\theta\theta$ part, the direct computation gives
\begin{equation}
\mathbf{pr}_{H\theta^\theta} \tilde{P}^{\theta\theta}=2\int\tr\left(\rho\rho\theta u-\lambda\lambda u \theta +\rho\theta u\theta u-\lambda u\theta u \theta\right).
\end{equation}
Using $\int\tr(\rho_1 \rho_1 \rho_1-\rho\rho\rho)=0$ and $\rho_1=\rho+\theta u$, we have $\int\tr\rho\rho \theta u=\int\tr(-\rho\theta u \theta u -\tfrac{1}{3}\theta u\theta u\theta u)$, and similarly for $\lambda$. Then
\begin{equation}
\mathbf{pr}_{H\theta^\theta} \tilde{P}^{\theta\theta}=\frac{2}{3}\int\tr(u\theta u\theta u\theta-\theta u \theta u \theta u)=0.
\end{equation}
For the $\theta\theta \zeta$ part the picture is similar: we use the identities $\int\tr (\rho_1\rho_1\mu_1-\rho\rho\mu)=0$ and $\int\tr (\lambda_1\lambda_1\sigma_1-\lambda\lambda\sigma)=0$ to simplify the expression and verify that it vanishes.
\end{proof}

The operator $\tilde{H}$ of Theorem \ref{thm:H02-tilde} can be combined with ultralocal, non-Hamiltonian operators to give a one-parameter family of Hamiltonian operators which provide the Hamiltonian structures for the nonabelian Ablowitz-Ladik and Chen-Lee-Liu integrable equations. They reduce to ultralocal brackets when the variables are assumed to commute.

\begin{theorem}\label{thm:Ham0pJP}
The operator
\begin{equation}\label{eq:OpJP}
H_\alpha=\left(\begin{array}{cc}0 & -2\r_{vu}+\alpha\\
2\l_{vu}-\alpha & 0\end{array}\right)+\tilde{H}
\end{equation}
is Poisson for any $\alpha\in\R$ (and, therefore, Hamiltonian).
\end{theorem}
\begin{proof}
We can consider a generic operator depending on two parameters
\begin{equation}
\check{H}=\begin{pmatrix} 0 & \beta\r_{vu}+\alpha\\-\beta\l_{vu}-\alpha & 0\end{pmatrix}
\end{equation}
with the associated bivector
\begin{equation}
\check{P}=2\int\tr\left(\beta vu\theta\zeta+\alpha\theta\zeta\right).
\end{equation}
The Poisson property \eqref{eq:HamProp} for $\check{P}+\tilde{P}$ reduces, since $\tilde{P}$ is Poisson on its own, to
\begin{equation}
\mathbf{pr}_{\tilde{H}}\check{P}+\mathbf{pr}_{\check{H}}\tilde{P}+\mathbf{pr}_{\check{H}}\check{P}=0.
\end{equation}
Relying on the identities \eqref{eq:id-1}, \eqref{eq:id-2} and the analogue ones for the remaining nonlocal terms in $\tilde{P}$, we can perform a straightforward computation that gives us
\begin{equation}
-\int\tr\left[\alpha(\beta+2)u\theta\zeta\theta+\beta(\beta+2)u\theta u\theta \zeta v\right]=0,
\end{equation}
which is satisfied if and only if $\beta=-2$, for any value of $\alpha$.
\end{proof}

By shift of variables, from $H_\alpha$ we can obtain a further Hamiltonian operator, linear in the variables $(u,v)$.
\begin{theorem}\label{thm:HK_Hamproof}
The operator
\begin{multline}\label{eq:HKaup}
H=\left(\begin{array}{cc}\c_u & \l_v+\r_u\\
-\r_v-\l_u & -\c_v\end{array}\right)\\
-\begin{pmatrix}1\\-1\end{pmatrix}\cS(\cS-1)^{-1}\begin{pmatrix}\c_u & \c_v\end{pmatrix}+\begin{pmatrix}\c_u\\ \c_v\end{pmatrix}(\cS-1)^{-1}\begin{pmatrix}1 & -1\end{pmatrix}
\end{multline}
is Poisson and, therefore, Hamiltonian. 
\end{theorem}
\begin{proof}
By the constant coordinate change $u\mapsto u-\eta$, $v\mapsto v-\eta$ in $H_\alpha$, we obtain
$\tilde H=H_{\alpha-2\eta^2}+2\eta H$.
Let us denote $P$ the bivector defined by $H$ and $P_{\alpha-2\eta^2}$ the one defined by $H_{\alpha-2\eta^2}$. Note that the latter is Poisson, as proved in Theorem \ref{thm:Ham0pJP} for any value of $\alpha$. Condition \eqref{eq:HamProp} for $\tilde{H}$ gives
\begin{equation}
0=\mathbf{pr}_{H_{\alpha-2\eta^2}\Theta} P_{\alpha-2\eta^2}+2\eta\left(\mathbf{pr}_{H_{\alpha-2\eta^2}\Theta} P+\mathbf{pr}_{H\Theta}P_{\alpha-2\eta^2}\right)+4\eta^2\mathbf{pr}_{H\Theta}P.
\end{equation}
The first term vanishes because $H_\alpha$ is Poisson for any value of the constant. Moreover, since the second term is either linear or of degree three in $\eta$, its vanishing is independent from the one of the third one, which corresponds to the Poisson property for $H$.
\end{proof}
\subsection{Hamiltonian structures for integrable nonabelian difference systems}\label{ssec:HamStr-List}
In this section we provide a Hamiltonian formulation for three nonabelian systems we introduced in \cite{cw19-2}. Their Hamiltonian structures belong to the class discussed in Section \ref{ssec:H0} 
\subsubsection{Nonabelian Kaup Lattice}
The nonabelian Kaup system (see \cite{ay94} for the Abelian version)
\begin{eqnarray*}
		\left\{ {\begin{array}{l} u_t=(u_1-u) (u+v) \\ v_t= (u+v) (v-v_{-1}) \end{array} } \right.
	\end{eqnarray*}
is Hamiltonian with respect to the structure $H$ given in \eqref{eq:HKaup} with Hamiltonian functional
\begin{equation}\label{eq:FKaup}
F=\int\tr\left(u_1v-uv\right).
\end{equation}
Note that the Hamiltonian functional has the same form of the one for the Abelian case, see \cite{kmw13}.
\subsubsection{Nonabelian Ablowitz-Ladik Lattice}
The nonabelian Ablowitz-Ladik lattice (first introduced in \cite{al} for the Abelian case)
\begin{equation}
		\left\{ {\begin{array}{l} u_t=\alpha(u_1-u_1vu)+\beta(uvu_{-1}-u_{-1}) \\ v_t= \alpha(v u v_{-1}-v_{-1})+\beta(v_1-v_1uv) \end{array} } \right.\qquad\alpha,\beta\in\R
	\end{equation}
is Hamiltonian with respect to the operator $H_{2}$ ($H_\alpha$ of equation \eqref{eq:OpJP} with $\alpha=2$) and Hamiltonian functional
\begin{equation}
G=\frac12\int\tr\left(\alpha u_1v-\beta uv_1\right).
\end{equation}
\subsubsection{Nonabelian Chen-Lee-Liu lattice}
The nonabelian Chen-Lee-Liu lattice (see \cite{ts02} for the Abelian case and \cite{kmw13} for its Abelian Hamiltonian formulation)
\begin{eqnarray*}
		\left\{ {\begin{array}{l} u_t=(u_1-u)(1+vu) \\ v_t= (1+vu)(v-v_{-1}) \end{array} } \right.
	\end{eqnarray*}
is Hamiltonian with respect to the operator $-H_{-2}$ ($H_\alpha$ of equation \eqref{eq:OpJP} with $\alpha=-2$) and Hamiltonian functional $\frac12 F$ with $F$ as in \eqref{eq:FKaup}.

\section{Discussion and further work}

The main purpose of this paper was investigating the notion of Hamiltoninan structure for noncommutative systems, in particular focussing on the differential-difference case (when these structures are given by difference operators).

In the literature, one finds two main approaches to the problem: on the one hand, one can require the existence of an operator defining a ``Poisson bracket'' (more precisely, a Lie bracket on the space of local functionals and an action of these by derivations, namely the Hamiltonian vector fields), see for example \cite{Kup00,OS98}: this is the most commonly adopted notion among the researchers active in Integrable Systems. On the other hand, one can define a suitable algebraic structure on the space of noncommutative local densities, and prove that this produces a Poisson bracket (in the sense above) among local functionals. This is the basic idea leading to the definition of double Poisson algebras \cite{vdb}, double Poisson vertex algebras \cite{dskv15}, and multiplicative double Poisson vertex algebra (see Section \ref{ssec:dmPVA}). In this paper, we show how the classical (and somehow geometric, in the sense that it exploits the language and machinery widely used in Poisson geometry) notion of Poisson bivector, and of Schouten brackets between polyvector fields, can be tailored to the functional  nonabelian case (functional polyvector fields for Abelian differential systems are very well known and long-established, see for instance \cite{O93}, and we have introduced them for Abelian differential-difference systems in \cite{CW19}). The Schouten brackets we defined in Section \ref{sec:geom} unify the two aforementioned languages, as well as the standard language of Poisson geometry.

One of the fundamental and basic notions of classical (in particular, commutative) Hamiltonian and Poisson structures is that the existence of a Poisson bracket is equivalent to the existence of a bivector (concretely, of an operator) with vanishing Schouten torsion (for a bivector $P$, this means  $[P,P]=0$), so that the bracket it defines is skewsymmetric (because a bivector is skewsymmetric by definition) and fulfils Jacobi identity (because of the Schouten condition). While this equivalence has been widely believed to exist in the nonabelian setting, too, we can now conclude that this is not the case -- the quasi-Poisson structure of Kontsevich's system being a clear counterexample. As demonstrated in this paper, the notions of double Poisson (vertex) algebras and of Poisson bivectors (defined by suitable multiplicative or difference operators) on the space $\A$ of noncommutative (difference) Laurent polynomials are equivalent, but it is possible to define Hamiltonian structures with weaker assumptions.

An obvious way to do so is introducing the notion of double \emph{quasi}-Poisson algebras. They are the noncommutative analogue of quasi-Poisson manifolds \cite{akm02}, which are endowed with a non-Poisson bivector, and yet the bracket it defines satisfies the Jacobi identity on the orbit space of a group action. In our noncommutative case, we can regard the cyclic permutations (the quotient with respect to which constitutes our ``trace operation'') as such a group action.

In Section \ref{ssec:pois} we have reviewed, in abstract terms, the relation between (Poisson) bivectors and Poisson brackets; in particular, we showed that the Jacobi identity is equivalent to
$$
[[[[P,P],f],g],h]=0
$$
for any triple of local functionals $f,g,h$, and that the Lie algebra morphism between vector fields and local functionals is guaranteed by
$$
[[[P,P],f],g]=0
$$
for any pairs of local functionals. Both this conditions are satistified if $P$ is a Poisson bivector, namely $[P,P]=0$, but the condition is not necessary as showed in Theorem \ref{thm:qp-gen}. Identifying the quasi-Poisson property as being equivalent to $[[P,P],f]=0$ for all $f$ clearly shows that a non-Poisson bivector can indeed define a Poisson bracket on an appropriate space.

There are also examples of cases when an operator defines a Hamiltonian structure even if it is not skewsymmetric itself; in \cite{ar15}, the author exhibits the double bracket
\begin{align*}
\ldb u,v\rdb&=-vu\otimes 1, & \ldb v,u\rdb&=uv\otimes 1,&\ldb u,u\rdb =\ldb v,v\rdb&=0,
\end{align*}
which can be used to define a different Hamiltonian structure for the Kontsevich system \eqref{eq:kont}. The bracket $\{a,b\}:=m(\ldb a,b\rdb)$ (note that this is not yet a ``Poisson bracket'' as the one we define in \eqref{eq:PoisB-PA-def}, since it is not defined between local functionals because of the lack of the trace operation) satisfies the Loday property, which implies the Jacobi identity when computed on local functionals. At the same time, the double bracket is not skewsymmetric in $\A_0\otimes\A_0$ but defines a skewsymmetric bracket on $\F_0$. Such a bracket is neither double Poisson nor double quasi-Poisson, and yet it defines a Poisson bracket.

This example, as well as the results on double quasi-Poisson algebras, motivates us to pursue further work in the direction of finding a more general class of operators which can give rise to Hamiltonian structures (then characterising integrable systems, maybe in the bi-Hamiltonian flavour). One may drop the skewsymmetry requirement, as Arthamonov showed, or -- keeping the description in terms of functional bivectors, that are necessarily skewsymmetric -- look for structures which fail to be quasi-Poisson ($[[P,P],f]\neq 0$) but can still play the same role (possibly by $[[[P,P],f],g]=0$), which seems to be the necessary condition to guarantee the existence of a ``Poisson action".

\section*{Acknowledgements}
	The paper is supported by the EPSRC grants	EP/P012698/1, and EP/V050451/1  under the small grant scheme. Both authors gratefully acknowledge the financial support. The authors are thankful to the referees for their many suggestions, which greatly improved the quality of the manuscript; in particular, the content of Remark 1 was given by one of the reviewers, who also  helped to streamline the proofs of Lemma 1 and Proposition 11 providing some derivation rules. JPW would like to thank A.~V.~Mikhailov for useful discussions. MC would like to thank M.~Fairon and D.~Valeri for the many useful and enjoyable discussions. 
	
\appendix
\section{Equivalence between Schouten property and Jacobi identity for double multiplicative PVAs}\label{app:equivPVAJac}
In this Appendix we show, as claimed in the proof of Theorem \ref{thm:doublePVA-Poiss}, that the Poisson property for a scalar difference operator is equivalent to the Jacobi identity of the corresponding double multiplicative PVA.
Since we are in the scalar case, we denote the only degree 1 variable as $\theta$ and $\cS^p\theta=\theta_p$.
We consider a skewsymmetric scalar difference operator of the form \eqref{eq:diffop-dmpva-pf1}, namely
$$
H=\sum \l_{H^{(\alpha_p)}_L}\r_{H^{(\alpha_p)}_R}\cS^{p}
$$
and its corresponding $\lambda$ bracket
\begin{equation}
\ldb u_{\lambda} u\rdb=\sum H^{(\alpha_p)}_L\otimes H^{(\alpha_p)}_R\lambda^p.
\end{equation}
Note that the skewsymmetry of the operator (and of the bracket) implies
\begin{equation}
\sum H^{(\alpha_p)}_L\otimes H^{(\alpha_p)}_R\lambda^p=-\sum\left(\cS^{-p} H^{(\alpha_p)}_R\right)\otimes\left(\cS^{-p} H^{(\alpha_p)}_L\right)\lambda^{-p}.
\end{equation}

We obtain the three terms of the double Jacobi identity by a straightforward computation. In analogy with \eqref{eq:JacLin_pf2} and denoting $\dev_{u_m}$ as $\dev_m$, we write
\begin{align}\label{eq:dpva1}
\ldb u_{\lambda}\ldb u_{\mu} u\rdb_{(\alpha_p)}\rdb_{(\beta_q),L}&=\left(\dev_m H^{(\alpha_p)}_L\right)'\left(\cS^m H^{(\beta_q)}_L\right)\otimes\left(\cS^m H^{(\beta_q)}_R\right)\left(\dev_m H^{(\alpha_p)}_L\right)''\otimes H^{(\alpha_p)}_R\lambda^{m+q}\mu^p,\\\label{eq:dpva2}
\ldb u_{\mu}\ldb u_{\lambda} u\rdb_{(\alpha_p)}\rdb_{(\beta_q),R}&=H^{(\alpha_p)}_L\otimes \left(\dev_m H^{(\alpha_p)}_R\right)'\left(\cS^m H^{(\beta_q)}_L\right)\otimes\left(\cS^m H^{(\beta_q)}_R\right)\left(\dev_m H^{(\alpha_p)}_R\right)''\lambda^p\mu^{m+q},\\\label{eq:dpva3}
\ldb\ldb u_{\lambda}u\rdb_{(\alpha_p),\lambda\mu}u\rdb_{(\beta_q),L}&=H^{(\beta_q)}_L\left(\cS^{q-m}\dev_m H^{(\alpha_p)}_L\right)''\otimes \left(\cS^{q-m}H^{(\alpha_p)}_R\right)\otimes \left(\cS^{q-m}\dev_m H^{(\alpha_p)}_L\right)'H^{(\beta_q)}_R\lambda^{p+q-m}\mu^{q-m} .
\end{align}
On the other hand, to compute \eqref{eq:HamProp} we have
\begin{equation}
P=\sum P^{(\alpha_p)}=\frac12\int\tr\sum \theta H^{(\alpha_p)}_L\theta_p H^{(\alpha_p)}_R,
\end{equation}
from which we obtain
\begin{equation}\label{eq:diffop-fmpva-pf3}
\begin{split}
-2\mathbf{pr}_{H^{(\beta_q)}\theta}P^{(\alpha_p)}&=\int\tr\left[\theta\left(\dev_m H^{(\alpha_p)}_L\right)'\left(\cS^m H^{(\beta_q)}_L\right)\theta_{q+m}\left(\cS^{m}H^{(\beta_q)}_R\right)\left(\dev_m H^{(\alpha_p)}_L\right)''\theta_p H^{(\alpha_p)}_R\right.\\
&\qquad\left.-\theta H^{(\alpha_p)}_L\theta_{p}\left(\dev_m H^{(\alpha_p)}_R\right)'\left(\cS^m H^{(\beta_q)}_L\right)\theta_{q+m}\left(\cS^m H^{(\beta_q)}_R\right)\left(\dev_m H^{(\alpha_p)}_R\right)''\right].
\end{split}
\end{equation}
As for the ultralocal case, the trace operation allows us to rewrite the RHS of \eqref{eq:diffop-fmpva-pf3} replacing it with all the cyclic permutations of its factors; moreover, the integral operation allows us to ``normalise'' each of the monomials we obtain by imposing the first $\theta$ to be taken without shifts.This gives us the six terms
\begin{align}\label{eq:diffop-fmpva-pf4}
3\text{\eqref{eq:diffop-fmpva-pf3}}&=\int\tr\left[\theta\left(\dev_m H^{(\alpha_p)}_L\right)'\left(\cS^m H^{(\beta_q)}_L\right)\theta_{q+m}\left(\cS^{m}H^{(\beta_q)}_R\right)\left(\dev_m H^{(\alpha_p)}_L\right)''\theta_p H^{(\alpha_p)}_R\right.\\
&\qquad-\theta\left(\cS^{-p}\dev_m H^{(\alpha_p)}_R\right)'\left(\cS^{m-p}H^{(\beta_q)}_L\right)\theta_{q+m-p}\left(\cS^{m-p}H^{(\beta_q)}_R\right)\left(\cS^{-p}\dev_m H^{(\alpha_p)}_R\right)''\theta_{-p}\left(\cS^{-p}H^{(\alpha_p)}_L\right)\\
&\qquad-\theta H^{(\alpha_p)}_L\theta_p\left(\dev_m H^{(\alpha_p)}_R\right)'\left(\cS^m H^{(\beta_q)}_L\right)\theta_{q+m}\left(\cS^m H^{(\beta_q)}_R\right)\left(\dev_m H^{(\alpha_p)}_R\right)''\\
&\qquad+\theta \left(\cS^{-p}H^{(\alpha_p)}_R\right)\theta_{-p}\left(\cS^{-p}\dev_m H^{(\alpha_p)}_L\right)'\left(\cS^{m-p}H^{(\beta_q)}_L\right)\theta_{q+m-p}\left(\cS^{m-p} H^{(\beta_q)}_R\right)\left(\cS^{-p}\dev_m H^{(\alpha_p)}_L\right)''\\
&\qquad - \theta\left(\cS^{-q} H^{(\beta_q)}_R\right)\left(\cS^{-q-m}\dev_m H^{(\alpha_p)}_R\right)''\theta_{-q-m}\left(\cS^{-q-m}H^{(\alpha_p)}_L\right)\theta_{p-q-m}\left(\cS^{-q-m}\dev_m H^{(\alpha_p)}_R\right)'\left(\cS^{-q} H^{(\beta_q)}_L\right)\\
&\qquad+\left.\theta\left(\cS^{-q}H^{(\beta_q)}_R\right)\left(\cS^{-q-m}\dev_m H^{(\alpha_p)}_L\right)''\theta_{p-q-m}\left(\cS^{-q-m}H^{(\alpha_p)}_R\right)\theta_{-q-m}\left(\cS^{-q-m}\dev_m H^{(\alpha_p)}_L\right)'\left(\cS^{-q}H^{(\beta_q)}_L\right)\right].
\end{align}
Note that these six terms can be paired; we can actually show that they cancel in such pairs, thanks to the skewsymmetry of the operator $H$ and the property $\cS\dev_m=\dev_{m+1}\cS$. Let us consider the first two lines: we can move the shift operators inside the derivatives in the second line, obtaining
\begin{equation}
-\theta\left(\dev_{m-p} \cS^{-p}H^{(\alpha_p)}_R\right)'\left(\cS^{m-p} H^{(\beta_q)}_L\right)\theta_{q+(m-p)}\left(\cS^{m-p}H^{(\beta_q)}_R\right)\left(\dev_{m-p} \cS^{-p}H^{(\alpha_p)}_R\right)''\theta_{-p}\left(\cS^{-p}H^{(\alpha_p)}_L\right).
\end{equation}
By using the skewsymmetry for $H^{(\alpha_p)}_{L,R}$ (which allows us to replace $\cS^{-p}H^{(\alpha_p)}_R\theta_{-p}\cS^{-p}H^{(\alpha_p)}_L$ with $-H^{(\alpha_p)}_L\theta_pH^{(\alpha_p)}_R$) and relabelling the indices, we obtain
\begin{equation}
\theta\left(\dev_m H^{(\alpha_p)}_L\right)'\left(\cS^m H^{(\beta_q)}_L\right)\theta_{q+m}\left(\cS^{m}H^{(\beta_q)}_R\right)\left(\dev_m H^{(\alpha_p)}_L\right)''\theta_p H^{(\alpha_p)}_R,
\end{equation}
which is another copy of the first term. We do the same for the term of the forth line and obtain the term on the third one. For the term in the fifth line, first we exploit the skewsymmetry inside the argument of the derivatives, obtaining
\begin{multline}
\theta\left(\cS^{-q} H^{(\beta_q)}_R\right)\left(\cS^{-q-m}\dev_m\cS^{-p} H^{(\alpha_p)}_L\right)''\theta_{-q-m}\left(\cS^{-q-m-p}H^{(\alpha_p)}_R\right)\theta_{-p-q-m}\left(\cS^{-q-m}\dev_m \cS^{-p}H^{(\alpha_p)}_L\right)'\left(\cS^{-q} H^{(\beta_q)}_L\right)\\
=\theta\left(\cS^{-q} H^{(\beta_q)}_R\right)\left(\cS^{-q-(m+p)}\dev_{m+p} H^{(\alpha_p)}_L\right)''\theta_{-q-m}\left(\cS^{-q-(m+p)}H^{(\alpha_p)}_R\right)\theta_{-q-(m+p)}\left(\cS^{-q-(m+p)}\dev_{m+p} H^{(\alpha_p)}_L\right)'\left(\cS^{-q} H^{(\beta_q)}_L\right)\\
=\theta\left(\cS^{-q} H^{(\beta_q)}_R\right)\left(\cS^{-q-m'}\dev_{m'} H^{(\alpha_p)}_L\right)''\theta_{-q-m'+p}\left(\cS^{-q-m'}H^{(\alpha_p)}_R\right)\theta_{-q-m'}\left(\cS^{-q-m'}\dev_{m'} H^{(\alpha_p)}_L\right)'\left(\cS^{-q} H^{(\beta_q)}_L\right).
\end{multline}
Then, using the skewsymmetry again, we obtain
\begin{equation}
-\theta H^{(\beta_q)}_L\left(\cS^{q-m}\dev_m H^{(\alpha_p)}_L\right)''\theta_{p+q-m}\left(\cS^{q-m}H^{(\alpha_p)}_R\right)\theta_{q-m}\left(\cS^{q-m}\dev_m H^{(\alpha_p)}_L\right)'H^{(\beta_q)}_R.
\end{equation}
Repeating this passage for the last line, our final result is
\begin{equation}\label{eq:diffop-fmpva-pf5}
\begin{split}
-3\mathbf{pr}_{H^{(\beta_q)}\theta} P^{(\alpha_p)}=\int\tr&\left[\theta\left(\dev_m H^{(\alpha_p)}_L\right)'\left(\cS^m H^{(\beta_q)}_L\right)\theta_{q+m}\left(\cS^{m}H^{(\beta_q)}_R\right)\left(\dev_m H^{(\alpha_p)}_L\right)''\theta_p H^{(\alpha_p)}_R\right.\\
&-\theta H^{(\alpha_p)}_L\theta_p\left(\dev_m H^{(\alpha_p)}_R\right)'\left(\cS^m H^{(\beta_q)}_L\right)\theta_{q+m}\left(\cS^m H^{(\beta_q)}_R\right)\left(\dev_m H^{(\alpha_p)}_R\right)''\\
&\left.-\theta H^{(\beta_q)}_L\left(\cS^{q-m}\dev_m H^{(\alpha_p)}_L\right)''\theta_{p+q-m}\left(\cS^{q-m}H^{(\alpha_p)}_R\right)\theta_{q-m}\left(\cS^{q-m}\dev_m H^{(\alpha_p)}_L\right)'H^{(\beta_q)}_R\right].
\end{split}
\end{equation}
Comparing \eqref{eq:diffop-fmpva-pf5} with the expression for the double Jacobi identity (see Definition \ref{def:dmpva}, explicitly obtained by $\text{\eqref{eq:dpva1}}-\text{\eqref{eq:dpva2}}-\text{\eqref{eq:dpva3}}$), we observe that each term of \eqref{eq:diffop-fmpva-pf5} is made of three elementary pieces kept separated by the $\theta$ variables, similarly to the fact that each of the summands in the double Jacobi identity is an element of $\A\otimes\A\otimes A$.  Moreover, they exactly match each of them line by line: for instance, the three factors in \eqref{eq:dpva1} are $\left(\dev_m H^{(\alpha_p)}_L\right)'\left(\cS^m H^{(\beta_q)}_L\right)$, $\left(\cS^{m}H^{(\beta_q)}_R\right)\left(\dev_m H^{(\alpha_p)}_L\right)''$, and $H^{(\alpha_p)}_R$. Finally, the degree of $\lambda$ and $\mu$ in each terms of the double Jacobi identity matches the number of shift of, respectively, the second and the third $\theta$'s in \eqref{eq:diffop-fmpva-pf5} (the order of the $\theta$'s in the expression is fixed by ``normalizing'' them leaving the first one without shifts). 

If we compare the results we computed with the structure of \eqref{eq:JacLin_pf2} and \eqref{eq:JacLin_pf1} respectively, we see that each of the summands for the Poisson property of the operator and of the double Jacobi identity for the $\lambda$ brackets coincide, making them equivalent. In particular, the vanishing of the former one is equivalent to the vanishing of the latter one, as claimed in Theorem \ref{thm:doublePVA-Poiss}.

\section{Graded Jacobi identity for the Schouten bracket}\label{app:Jacobi}
In this Appendix we prove the second half of Proposition \ref{thm:Sch_skew} of Section \ref{ssec:sch-ul},  namely that the bracket we have defined in \eqref{eq:schdef} satisfies the graded version of the Jacobi identity \eqref{eq:schJac}; together with Proposition \ref{thm:Sch_skew} and \ref{thm:Sch_vf}, this means that it is a \emph{bona fide} Schouten bracket. Part (iv) of Proposition \ref{thm:property-ds} of Section \ref{ssec:sch-l} is proved along the same lines: we do not repeat the proof because this one has the advantage of being more compact because of the single index on the generators and the absence of the formal variables $\lambda$ and $\mu$.

The proof is in two main steps: first, we prove a version of the vanishing of the graded triple bracket associated to a double Schouten bracket, similarly to \eqref{eq:tripleb}. To do so, we prove that it vanishes for local functionals and vector fields (namely for elements of degree 1); then we show by induction that it holds true for elements of arbitrary degree. Secondly, we prove (similarly to \cite[Corollary 2.4.4]{vdb}) that the Jacobi identity for the Schouten bracket follows from the vanishing of the graded triple bracket.

Let us start with the graded Jacobi-like identity for the double Schouten bracket. We introduce the notation
\begin{gather}
\lsb a,b\otimes c\rsb_L=(-1)^{(|a|-1)|c|}\lsb a,b\rsb\otimes c,\qquad \lsb a,b\otimes c\rsb_R=b\otimes\lsb a,c\rsb,\\ \label{eq:notationappB}
\lsb a\otimes b,c\rsb_L=\lsb a,c\rsb\otimes_1 b,\qquad\lsb a\otimes b,c\rsb_R=(-1)^{(|c|-1)|a|}a\otimes_1\lsb b,c\rsb.
\end{gather}
which will be instrumental in writing the Jacobi identity for the Schouten bracket, similarly to the one we introduced for the Jacobi identity of double $\lambda$ brackets in Section \ref{ssec:dmPVA}.
\begin{lemma}\label{lem:grad-id}
We have the following identities:
\begin{align}\label{eq:grad-id1}
\lsb b,\lsb a,c\rsb\rsb_R&=-(-1)^{(|a|-1)(|c|-1)}\tau\left(\lsb b,\lsb c,a\rsb\rsb_L\right),\\ \label{eq:grad-id2}
\lsb \lsb a,b\rsb,c\rsb_L&=-(-1)^{(|c|-1)(|a|+|b|)}\tau^2\left(\lsb c,\lsb a,b\rsb\rsb_L\right).
\end{align}
\end{lemma}
\begin{proof}
Straightforward computation. For instance, let $(ca)'=\left|\lsb c,a\rsb'\right|$ and $(ca)''=\left|\lsb c,a\rsb''\right|$. Then for \eqref{eq:grad-id1} we have
\begin{equation}
\begin{split}
\lsb b,\lsb a,c\rsb\rsb_R&=-(-1)^{(|a|-1)(|c|-1)}\lsb b,\lsb c,a\rsb^\sigma\rsb_R\\
&=-(-1)^{(|a|-1)(|c|-1)+(ca)'(ca)''}\lsb c,a\rsb''\otimes \lsb b,\lsb c,a\rsb'\rsb'\otimes\lsb b,\lsb c,a\rsb'\rsb''\\
&=-(-1)^{(|a|-1)(|c|-1)+(ca)'(ca)''+(ca)''(|b|+(ca)'+1)}\tau\left(\lsb b,\lsb c,a\rsb'\rsb'\otimes\lsb b,\lsb c,a\rsb'\rsb'\otimes\lsb c,a\rsb''\right)\\
&=-(-1)^{(|a|-1)(|c|-1)+(ca)''(|b|+1)+(ca)''(|b|+1)}\tau\left(\lsb b,\lsb c,a\rsb\rsb_L\right)\\
&=-(-1)^{(|a|-1)(|c|-1)}\tau\left(\lsb b,\lsb c,a\rsb\rsb_L\right).
\end{split}
\end{equation}
A similar computation yields \eqref{eq:grad-id2}.
\end{proof}

We prove the graded Jacobi identity for the double Schouten bracket by induction. The statement is given in the following Proposition:
\begin{proposition}\label{thm:dJac}
Let $\lsb -,-\rsb$ be the double Schouten bracket defined in \eqref{eq:ds_master_ul}. Then the following identity holds true for any $a,b,c$ in $\hat{\A}_0$:
\begin{equation}\label{eq:dJac-ul}
\lsb a,b,c\rsb:=\lsb a,\lsb b,c\rsb\rsb_L-(-1)^{(|a|-1)(|b|-1)}\lsb b,\lsb a,c\rsb\rsb_R-\lsb \lsb a,b\rsb,c\rsb_L=0.
\end{equation}
\end{proposition}

We first need to prove the initial cases of the induction, namely that identity \eqref{eq:dJac-ul} holds true if we consider local functionals ($0$-vectors) and $1$-vector fields.

It is obvious from \eqref{eq:ds_master_ul} that the bracket between two local functionals is 0: the Jacobi identity is then satisfied for triples of local functionals and for two local functionals and a 1-vector field (in this latter case, because the bracket of a 1-vector field with a local functional is a local functional, and hence the further bracket vanishes). We need to add to the initial cases the identity among two vector fields and a local functional, as well as the one among three vector fields. This is the result of the following two lemmas.

\begin{lemma}\label{lem:2v1f}
Let $X$ and $Y$ be local vector fields (elements of degree 1 in $\hat{\F}_0$) and $f$ a local functional. Then
\begin{equation}\label{eq:2v1f-1}
\lsb X,\lsb Y,f\rsb\rsb_L-\lsb Y,\lsb X,f\rsb\rsb_R=\lsb\lsb X,Y\rsb,f\rsb_L.
\end{equation}
\end{lemma}
\begin{proof}
Let us take for example $X=X^p\theta_p$ and $Y=Y^q\theta_q$ (the computation is similar and yields the same result if we consider more complicated forms for the 1-st order elements). For the LHS of \eqref{eq:2v1f-1} we have
\begin{multline}
\left(\frac{\dev(\dev_pf)'}{\dev u^q}\right)'\otimes X^q\left(\frac{\dev(\dev_pf)'}{\dev u^q}\right)''\otimes Y^p(\dev_pf)''-(\dev_p f)'\otimes(\dev_q X^p)'\otimes Y^q(\dev_q X^p)''(\dev_p f)''\\
-(\dev_q f)'\otimes X^q \left(\frac{\dev(\dev_qf)''}{\dev u^p}\right)'\otimes Y^p\left(\frac{\dev(\dev_qf)''}{\dev u^p}\right)'',
\end{multline}
where the sake of compactness, we adopt the shorthand notation
$$
(\dev_p Z^q)^{',''}:=\left(\frac{\dev Z^m}{\dev u^p}\right)^{',''}.
$$
On the RHS of \eqref{eq:2v1f-1} we obtain
\begin{equation}
-(\dev_p f)'\otimes(\dev_q X^p)'\otimes Y^q(\dev_q X^p)''(\dev_p f)''.
\end{equation}
The terms with first derivatives only immediately cancel. For the terms with the second derivatives, we are in a similar situation as \eqref{eq:comm_vf_pf4} in the proof of Lemma \ref{lem:comm_vf} -- here we don't have shifted variables and the expression is in $\A_0\otimes\A_0\otimes \A_0$ rather than in $\A_0$, but the terms are of the same form:
$$
\left(\frac{\dev(\dev_pf)'}{\dev u^q}\right)'\otimes X^q\left(\frac{\dev(\dev_pf)'}{\dev u^q}\right)''\otimes Y^p(\dev_pf)''-(\dev_q f)'\otimes X^q \left(\frac{\dev(\dev_qf)''}{\dev u^p}\right)'\otimes Y^p\left(\frac{\dev(\dev_qf)''}{\dev u^p}\right)''.
$$
They vanish for the same reason explained in Lemma \ref{lem:comm_vf}.
\end{proof}

\begin{lemma}\label{thm:Jac_vf}
Let $X= X^i\theta_i$, $Y=Y^j\theta_j$, $Z=\theta_kZ^k$ \emph{(note the different position of $\theta$: we do this to show that it is not relevant)} be densities of 1-vector fields. Then
\begin{equation}\label{eq:Jac_vf_th}
\lsb X,Y,Z\rsb=\lsb X,\lsb Y,Z\rsb\rsb_L-\lsb Y,\lsb X,Z\rsb\rsb_R-\lsb\lsb X,Y\rsb,Z\rsb_L=0.
\end{equation}
\end{lemma}
\begin{proof}
A direct computation gives
\begin{align}\label{eq:Jac_vf_pf1}
\lsb X,\lsb Y,Z\rsb\rsb_L&=\theta_k\left(\frac{\dev (\dev_m Z^k)'}{\dev u^l}\right)'\otimes X^l\left(\frac{\dev (\dev_m Z^k)'}{\dev u^l}\right)''\otimes Y^m(\dev_m Z^k)''\\
&\quad-(\dev_lX^k)''\theta_k\otimes(\dev_l X^k)'(\dev_m Z^l)'\otimes Y^m(\dev_m Z^k)''\\
&\quad-\left(\frac{\dev (\dev_m Y^k)''}{\dev u^l}\right)'\otimes X^l\left(\frac{\dev (\dev_m Y^k)''}{\dev u^l}\right)''\theta_k\otimes(\dev_mY^k)'Z^m\\
&\quad+(\dev_mY^l)''(\dev_lX^k)''\theta_k\otimes(\dev_l X^k)'\otimes(\dev_m Y^l)'Z^m,
\end{align}
\begin{align}\label{eq:Jac_vf_pf2}
\lsb Y,\lsb X,Z\rsb\rsb_R&=\theta_k(\dev_m Z^k)'\otimes(\dev_l X^m)'\otimes Y^l(\dev_l X^m)''(\dev_m Z^k)''\\
&\quad+\theta_k(\dev_m Z^k)'\otimes X^m\left(\frac{\dev (\dev_m Z^k)''}{\dev u^l}\right)'\otimes Y^l\left(\frac{\dev (\dev_m Z^k)''}{\dev u^l}\right)''\\
&\quad-(\dev_m X^k)''\theta_k\otimes\left(\frac{\dev (\dev_m X^k)'}{\dev u^l}\right)'\otimes Y^l\left(\frac{\dev (\dev_m X^k)'}{\dev u^l}\right)''Z^m\\
&\quad-(\dev_m X^k)''\theta_k\otimes(\dev_m X^k)'(\dev_l Z^m)'\otimes Y^l(\dev_l Z^m)''
\end{align}
and
\begin{align}\label{eq:Jac_vf_pf3}
\lsb \lsb X,Y\rsb,Z\rsb_L&=-\left(\frac{\dev (\dev_m Y^k)'}{\dev u^l}\right)''\otimes X^m(\dev_m Y^k)''\theta_k\otimes\left(\frac{\dev (\dev_m Y^k)'}{\dev u^l}\right)'Z^l\\
&\quad-\theta_k(\dev_lZ^k)'\otimes(\dev_mX^l)'\otimes Y^m(\dev_m X^l)''(\dev_l Z^k)''\\
&\quad+(\dev_lY^m)''(\dev_m X^k)''\theta_k\otimes(\dev_m X^k)'\otimes(\dev_lY^m)'Z^l\\
&\quad+\left(\frac{\dev (\dev_m X^k)''}{\dev u^l}\right)''\theta_k\otimes(\dev_m X^k)'\otimes Y^m\left(\frac{\dev (\dev_m X^k)''}{\dev u^l}\right)'Z^l.
\end{align}
We observe that, taking the three expression with the signs given in \eqref{eq:Jac_vf_th}, the terms with only first derivatives vanish (for instance, the second line of \eqref{eq:Jac_vf_pf1} cancels out with (minus) the fourth line of \eqref{eq:Jac_vf_pf2} upon the exchange of the indices $(l,m)$).
For the terms with second derivatives we have, for instance, that \eqref{eq:Jac_vf_pf1} and \eqref{eq:Jac_vf_pf2} yield
\begin{equation}
\left(\theta_k\otimes X^l\otimes Y^m\right)\left(\left(\frac{\dev}{\dev u^l}\otimes 1\right)\circ\frac{\dev Z^k}{\dev u^m}-\left(1\otimes\frac{\dev}{\dev u^m}\right)\circ\frac{\dev Z^k}{\dev u^l}\right),
\end{equation}
and the terms in the bracket vanish because of the commutation of the double derivatives \cite[Lemma 2.6]{dskv15}.
\end{proof}
Similar computations shows that the identity holds true regardless of the relative position of $\theta$ inside $X$, $Y$, or $Z$.

Finally, let \eqref{eq:dJac-ul} be our inductive hypothesis. We show that raising the degree of the vector fields in the Jacobi identity by multiplying an entry by a 1-vector gives us a version of the identity with signs in accordance with our claim.

\begin{lemma}\label{thm:Sch_Jac_ind}
Given $x$ of degree $1$, we have
\begin{enumerate}[label=(\roman*)]
\item $\lsb a,\lsb b,cx\rsb\rsb_L-(-1)^{(|a|-1)(|b|-1)}\lsb b,\lsb a,cx\rsb\rsb_R-\lsb\lsb a,b\rsb,cx\rsb_L=0$;
\item $\lsb a,\lsb bx,c\rsb\rsb_L-(-1)^{((|a|-1)|b|}\lsb bx,\lsb a,c\rsb\rsb_R-\lsb\lsb a,bx\rsb,c\rsb_L=0$;
\item $\lsb ax,\lsb b,c\rsb\rsb_L-(-1)^{|a|(|b|-1)}\lsb b,\lsb ax,c\rsb\rsb_R-\lsb\lsb ax,b\rsb,c\rsb_L=0$.
\end{enumerate}
\end{lemma}
\begin{proof}
The proof is essentially computational, exploiting the skewsymmetry and the Leibniz property for the double Schouten bracket. We show here the detailed derivation of (i). For the first term we have
\begin{equation}
\begin{split}
\lsb a,\lsb b,cx\rsb\rsb_L&=\lsb a,c\lsb b,x\rsb\rsb_L+(-1)^{|b|-1}\lsb a,\lsb b,c\rsb x\rsb_L\\
&=\lsb a,c\lsb b,x\rsb'\otimes\lsb b,x\rsb''\rsb_L+(-1)^{|b|-1}\lsb a,\lsb b,c\rsb'\otimes\lsb b,c\rsb''x\rsb_L\\
&=(-1)^{(|a|-1)(bx)''}\lsb a,c\lsb b,x\rsb'\rsb\otimes \lsb b,x\rsb''\\
&\quad+(-1)^{|b|-1+(|a|-1)((bc)''+1)}\lsb a,\lsb b,c\rsb'\rsb\otimes \lsb b,c\rsb''x,
\end{split}
\end{equation}
where we denote $(bx)''=|\lsb b,x\rsb''|$ and $(bc)''=|\lsb b,c\rsb''|$. Using once again the Leibniz property and the definitions of $\lsb-,-\rsb_L$ we obtain
\begin{equation}
\begin{split}
\lsb a,\lsb b,cx\rsb\rsb_L&=c\lsb a,\lsb b,x\rsb\rsb_L+(-1)^{(|a|-1)|b|}\lsb a,c\rsb\lsb b,x\rsb\\
&\quad +(-1)^{|a|+|b|}\lsb a,\lsb b,c\rsb\rsb_L x.
\end{split}
\end{equation}
Similar computations for the remaining terms give
\begin{align}
\lsb b,\lsb a,cx\rsb\rsb_R&=c\lsb b,\lsb a,x\rsb\rsb_R+(-1)^{|a|-1}\lsb a,c\rsb\lsb b,x\rsb\\
&\quad+(-1)^{|a|+|b|}\lsb b,\lsb a,c\rsb\rsb_R x,\\
\lsb \lsb a,b\rsb,cx\rsb_L&=c\lsb\lsb a,b\rsb,x\rsb_R+(-1)^{|a|+|b|}\lsb \lsb a,b\rsb,c\rsb_Lx. 
\end{align}
Upon multiplication of the second term by $-(-1)^{(|a|-1)(|b|-1)}$ and of the third by $(-1)$ we observe the vanishing of the two summands with two double brackets and can collect
\begin{multline}
c\left(\lsb a,\lsb b,x\rsb\rsb_L-(-1)^{(|a|-1)(|b|-1)}\lsb b,\lsb a,x\rsb\rsb_R-\lsb \lsb a,b\rsb,x\rsb_L\right)\\
+\left(\lsb a,\lsb b,c\rsb\rsb_L-(-1)^{(|a|-1)(|b|-1)}\lsb b,\lsb a,c\rsb\rsb_R-\lsb \lsb a,b\rsb,c\rsb_L\right)x.
\end{multline}
The terms in the two parentheses vanish by inductive hypothesis, proving our statement by induction. To prove the remaining two identities, we can observe that, due to Lemma \ref{lem:grad-id}, the first one can be alternatively  written as
\begin{multline}\label{eq:dJac-ind-sym}
(-1)^{(|a|-1)|c|}\lsb a,\lsb b,cx\rsb\rsb_L+(-1)^{(|a|-1)(|b|-1)}\tau\left(\lsb b,\lsb cx,a\rsb\rsb_L\right)\\
+(-1)^{(|b|-1)|c|}\tau^2\left(\lsb cx,\lsb a,b\rsb\rsb_L\right)=0.
\end{multline}
This form is apparently cyclically symmetric: then (2) and (3) can be easily brought to the form \eqref{eq:dJac-ind-sym} by the suitable cyclic permutation.
\end{proof}
Lemma \ref{thm:Sch_Jac_ind}, together with the initial cases (the obvious ones, Lemma \ref{lem:2v1f} and Lemma \ref{thm:Jac_vf}), completes the proof by induction of Proposition \ref{thm:dJac}. We can now move to the main result, holding true for the ``actual'' Schouten bracket. 
\begin{theorem}\label{thm:schJac}
Let $A$, $B$ and $C$ be, respectively, $a$-, $b$-, and $c$-vector fields. Then their Schouten bracket \eqref{eq:schdef} satisfies the Jacobi identity
\begin{equation}\label{eq:schJac_th}
[A,[B,C]]=[[A,B],C]+(-1)^{(a-1)(b-1)}[B,[A,C]].
\end{equation}
\end{theorem}
\begin{proof}
First, we rewrite \eqref{eq:schJac_th} from its definition in terms of double Schouten brackets.
\begin{equation}
\tr m\left(\lsb A,\tr m\lsb B,C\rsb \rsb-(-1)^{(a-1)(b-1)}\lsb B,\tr m  \lsb A,C\rsb \rsb-\lsb \tr m\lsb A,B\rsb,C\rsb\right)=0.
\end{equation}
Since we have learnt in Proposition \ref{def:sch} that the Schouten bracket is well-defined, we can drop the trace operation inside the brackets, and focus on
\begin{equation}\label{eq:sch_Jac_pf1}
\tr m\left(\lsb A, m\lsb B,C\rsb \rsb-(-1)^{(a-1)(b-1)}\lsb B, m  \lsb A,C\rsb \rsb-\lsb  m\lsb A,B\rsb,C\rsb\right).
\end{equation}
The quantity \eqref{eq:sch_Jac_pf1} can be written, using the usual Sweedler's notation, as
\begin{equation}
\tr m\left(\lsb A,\lsb B,C\rsb'\lsb B,C\rsb''\rsb-(-1)^{(a-1)(b-1)}\lsb B,\lsb A,C\rsb'\lsb A,C\rsb''\rsb-\lsb\lsb A,B\rsb'\lsb A,B\rsb'',C\rsb\right).
\end{equation}
By the Leibniz properties for the double Schouten bracket, this is in turn
\begin{multline}
\tr m\left(\lsb B,C\rsb'\lsb A,\lsb B,C\rsb''\rsb+(-1)^{(a-1)(bc)''}\lsb A,\lsb B,C\rsb'\rsb\lsb B,C\rsb''\right.\\
-(-1)^{(a-1)(b-1)}\lsb A,C\rsb'\lsb B,\lsb A,C\rsb''\rsb-(-1)^{(a-1+(ac)'')(b-1)}\lsb B,\lsb A,C\rsb'\rsb\lsb A,C\rsb''\\\left.
-\lsb\lsb A,B\rsb',C\rsb\star\lsb A,B\rsb''-(-1)^{(ab)'(c-1)}\lsb A,B\rsb'\star\lsb \lsb A,B\rsb'',C\rsb\right).
\end{multline}
We now recall the definitions of $\lsb-,-\rsb_{L,R}$ which allow, together with the multiplication map, to rewrite the previous expression as
\begin{multline}\label{eq:sch_Jac_pf2}
\tr m\left((m\otimes 1)\lsb A,\lsb B,C\rsb\rsb_R+(1\otimes m)\lsb A,\lsb B,C\rsb\rsb_L\right.\\
-(-1)^{(a-1)(b-1)}(m\otimes 1)\lsb B,\lsb A,C\rsb\rsb_R-(-1)^{(a-1)(b-1)}(1\otimes m)\lsb B,\lsb A,C\rsb\rsb_L\\
\left.-(m\otimes 1)\lsb\lsb A,B\rsb,C\rsb_L-(-1)^{(ab)'(c-1)}\lsb A,B\rsb'\star\lsb \lsb A,B\rsb'',C\rsb\right).
\end{multline}
The last term requires a closer inspection. By definition we have
\begin{equation}
\begin{split}
\lsb A,B\rsb'\star\lsb\lsb A,B\rsb'',C\rsb&=(-1)^{(ab)'((ab)''c)'}\lsb \lsb A,B\rsb'',C\rsb'\otimes \lsb A,B\rsb'\lsb \lsb A,B\rsb'',C\rsb''\\
&=(-1)^{(ab)'(((ab)''c)'+((ab)''c)'')}(1\otimes m)\lsb \lsb A,B\rsb'',C\rsb'\otimes \lsb \lsb A,B\rsb'',C\rsb''\otimes_1\lsb A,B\rsb'\\
&=(-1)^{(ab)'(ab)''+(ab)'(c-1)}(1\otimes m)\lsb \lsb A,B\rsb''\otimes\lsb A,B\rsb',C\rsb_L\\
&=-(-1)^{(a-1)(b-1)+(ab)'(c-1)}(1\otimes m)\lsb \lsb B,A\rsb,C\rsb_L. 
\end{split}
\end{equation}
This means that \eqref{eq:sch_Jac_pf2} is finally equal to
\begin{multline}\label{eq:sch_Jac_pf3}
\tr m\left((m\otimes 1)\lsb A,\lsb B,C\rsb\rsb_R+(1\otimes m)\lsb A,\lsb B,C\rsb\rsb_L\right.\\
-(-1)^{(a-1)(b-1)}(m\otimes 1)\lsb B,\lsb A,C\rsb\rsb_R-(-1)^{(a-1)(b-1)}(1\otimes m)\lsb B,\lsb A,C\rsb\rsb_L\\
\left.-(m\otimes 1)\lsb A,B\rsb,C\rsb_L+(-1)^{(a-1)(b-1)}(1\otimes m)\lsb \lsb B,A\rsb,C\rsb_L\right..
\end{multline}
Now, from $m(m\otimes 1)(a\otimes b\otimes c)=m(1\otimes m)(a\otimes b\otimes c)$ we can rearrange the terms in \eqref{eq:sch_Jac_pf3} and obtain
\begin{multline}\label{eq:sch_Jac_pf4}
\tr m\left((m\otimes 1)\left(\lsb A,\lsb B,C\rsb\rsb_L-(-1)^{(a-1)(b-1)}\lsb B,\lsb A,C\rsb\rsb_R-\lsb \lsb A,B\rsb,C\rsb_L\right)\right.\\
\left. -(-1)^{(a-1)(b-1)}(1\otimes m)\left(\lsb B,\lsb A,C\rsb\rsb_L-(-1)^{(a-1)(b-1)}\lsb A,\lsb B,C\rsb\rsb_R-\lsb\lsb B,A\rsb\rsb C\rsb_L\right)\right).
\end{multline}
Comparing this expression with \eqref{eq:dJac-ul}, we observe that it is
\begin{equation}
\tr m\left((m\otimes 1)\lsb A,B,C\rsb-(-1)^{(a-1)(b-1)}(1\otimes m)\lsb B,A,C\rsb\right),
\end{equation}
which vanishes because of Proposition \ref{thm:dJac}.
\end{proof}

\end{document}